\documentclass[11pt,a4paper]{article}
\usepackage[latin1]{inputenc}
\usepackage[english]{babel}
\usepackage{amsmath,amsthm}
\usepackage{amsfonts}
\usepackage{amssymb,bbm}
\usepackage{graphicx}
\usepackage{hyperref,mathtools}
\usepackage{bbm}
\usepackage[usenames,dvipsnames]{color}
\usepackage[nottoc]{tocbibind}

\swapnumbers
\theoremstyle{plain}
\newtheorem{thm}{Theorem}[subsection]
\newtheorem{lem}[thm]{Lemma}
\newtheorem{prop}[thm]{Proposition}
\newtheorem{cor}[thm]{Corollary}

\theoremstyle{definition}
\newtheorem{defn}[thm]{Definition}

\newtheorem{exmp}[thm]{Example}

\theoremstyle{remark}
\newtheorem{rem}[thm]{Remark}

\begin{document}

\begin{titlepage}
\begin{center}

{\Huge{\bfseries Topics in\\[.3cm] Stochastic Portfolio Theory}}

\vspace{25mm}

\includegraphics[width=60mm]{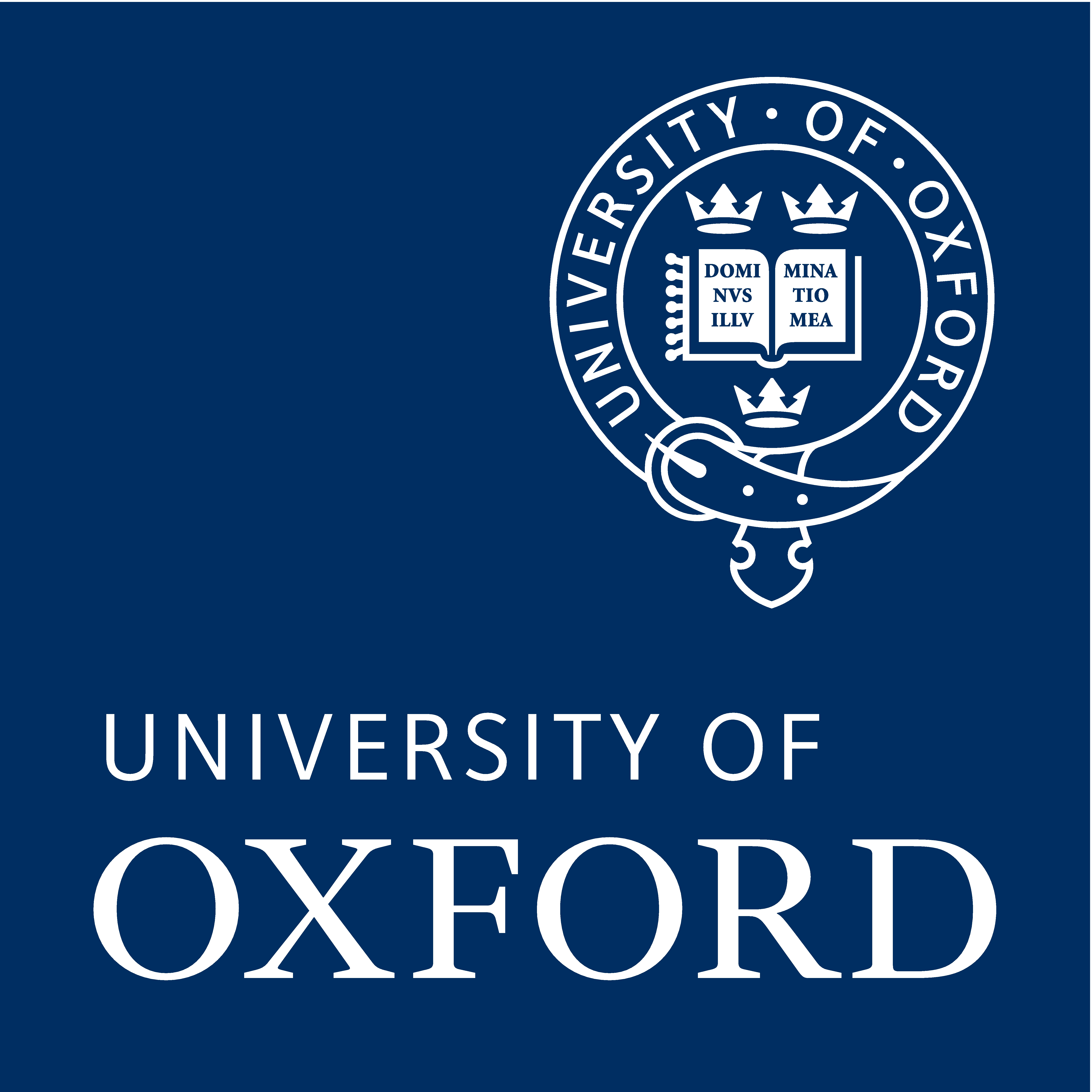}

\vspace{25mm}

{\Large Alexander Vervuurt}\\[0.3cm]
{\large 
    {The Queen's College}}\\[0.3cm]
{\large 
    {University of Oxford}}\\[0.3cm]
\vspace{25mm}

{\large {Thesis submitted for transfer from PRS to DPhil status}}\\[0.6cm]
{\large 
    {31 October 2014}}

\end{center}
\end{titlepage}

\tableofcontents

\section{Introduction}


Stochastic Portfolio Theory (SPT) is a framework in which the \emph{normative} assumptions from `classical' financial mathematics are not made\footnote{See Section 0.1 of \cite{kard08} for a motivation by Kardaras.}, but in which one takes a \emph{descriptive} approach to studying properties of markets that follow from empirical observations. More concretely, one does not assume the existence of an equivalent local martingale measure (ELMM), or, equivalently (by the First Theorem of Asset Pricing as proved by Delbaen and Schachermayer \cite{ds94}), the No Free Lunch with Vanishing Risk (NFLVR) assumption. 
Instead, in SPT one places oneself in a general It\^o model and assumes only the weaker No Unbounded Profit with Bounded Risk (NUPBR) condition, which was first defined in \cite{fk05}. 
The aim then is to find investment strategies which outperform the market in a pathwise fashion, and in particular ones that avoid making assumptions about the expected returns of stocks, which are notoriously difficult to estimate (see \cite{lcgr01}, for example).
SPT was initiated by Robert Fernholz (see \cite{f95}, \cite{f99}, \cite{f01} and the book \cite{f02}), and a major review of the area was made in 2009 by Fernholz and Karatzas in \cite{fk09}. In this review, the authors described the progress made thus far regarding the problem of finding so-called relative arbitrages, and listed several open questions, some of which have been solved since then, and some of which remain unsolved.

The objective taken in the framework of SPT is that of finding investment strategies with a good \emph{pathwise} and \emph{relative} performance compared to the entire market, that is, strategies which almost surely outgrow the market index (usually by a given time); these are portfolios which `beat the market'. Fernholz defines such portfolios as \emph{relative arbitrages}, and constructively proves the existence of such investment opportunities in certain types of markets. These model classes are general It\^o models with additional assumptions on the volatility structure and on the behaviour of the \emph{market weights} of the stocks that the investor is allowed to invest in, i.e.~the ratios of company capitalisations and the total market capitalisation. Several such classes, corresponding to different assumptions on market behaviour (which arise from empirical observations), have been introduced and studied in SPT; these are:
\begin{enumerate}
\item diverse models --- here, the market weights are bounded from above by a number smaller than one, meaning that no single company can capitalise the entire market;
\item `intrinsically volatile' models --- here, a certain process related to the volatility of the entire market (which depends on both the market weights and the volatilities of the stocks) is required to be bounded away from zero;
\item rank-based models --- here, the drift and volatility processes of each stock are made to depend on the stock's rank according to its capitalisation. 
\end{enumerate}

Diversity is clearly observed in real markets, and its validity is guaranteed by the fact that anti-trust regulations are typically in place. This assumption was first studied in detail in the context of SPT by Fernholz, Karatzas and Kardaras \cite{fkk05}, who defined and studied different forms of diversity, and proved that under an additional nondegeneracy condition on the stock volatilities, relative arbitrages exist in such markets --- both over sufficiently long time horizons, as well as over arbitrarily short time horizons.

The property of `sufficient intrinsic volatility' has also been argued to hold for real markets in \cite{f02}. Without any additional assumptions, \cite{fk05} showed that there exists relative arbitrage over sufficiently long time horizons in models with this property, with the size of the time horizon required to beat the market depending on the size of the lower bound for average market volatility. It remains a major open problem whether a relative arbitrage over arbitrarily short time horizons exists in such models --- though it has been shown to exist in some special cases of sufficiently volatile markets, namely volatility-stabilised markets (VSMs; see \cite{bf08}), which have been studied in detail\footnote{See, for instance, \cite{pal11}, in which the dynamics of market weights in VSMs are studied.}, generalised VSMs (see \cite{pic13}), and Markovian intrinsically volatile models (see Proposition 2 and the following Corollary of \cite[pp. 1194--1195]{fernk10}).

Rank-based models were introduced to model the observation that the distribution of capital according to rank by capitalisation has been very stable over the past decades, as illustrated in \cite{f02}. The dynamics of stocks in these models have been studied extensively, but the question of existence of (asymptotic) relative arbitrage has not been addressed yet. A very simple case of a rank-based model, the Atlas model, was introduced and studied in \cite{bfk05}, and an extension was proposed in \cite{ipbkf11}. Large market limits and mean-field versions of this model have been studied. In \cite{f01}, Fernholz first introduced a framework for studying the performance of portfolios which put weights on stocks based on their rank instead of their name, allowing him to theoretically explain certain phenomena observed in real markets.

The main strength of SPT lies in the fact that it does not require any drift estimation, making it much more robust than `classical' approaches to portfolio optimisation, such as mean-variance optimisation or utility maximisation. Crucial in the construction of relative arbitrages are so-called \emph{functionally generated portfolios}, which are portfolios which depend only on the current market weights in a simple way, and are thus very easily implementable (ignoring transaction costs, a crucial caveat).

Although the portfolio selection criterion described above is not one of optimisation, there have been attempts at finding the `best' relative arbitrage by \cite{fernk10} (which gives a characterisation of the optimal relative arbitrage in complete Markovian NUPBR markets) and \cite{fk11} (in which this result is extended to markets with `Knightian' uncertainty). Although not possible in general SPT models, in volatility-stabilised markets the log-optimal or num\'eraire portfolio can be characterised explicitly. First steps towards the optimisation of functionally generated portfolios have been made by Pal and Wong \cite{palw13}.

Besides the above, numerous other topics related to SPT have been studied over the past decade and a half. Some progress has been made regarding the hedging of claims in markets in which NUPBR holds but NFLVR is allowed to fail --- \cite{ruf11} and \cite{ruf13a} show that the cheapest way of hedging a European claim in a Markovian market is to delta hedge, and in \cite{bkx07} the authors solve the problem of valuation and optimal exercise of American call options, resolving an open problem posed in \cite{fk09}.
Furthermore, several articles present and study certain nonequivalent changes of measures with the goal of constructing a market with certain properties: Osterrieder and Rheinl\"ander \cite{or06} create a diverse model this way, and prove the existence of a real arbitrage in this model under a nondegeneracy condition; in \cite{ct13}, Chau and Tankov proceed similarly, but instead change measure to incorporate an investor's belief of a certain event not happening, leading to arbitrage opportunities, of which the authors characterise the one which is optimal in terms of having the largest lower bound on terminal wealth; and in \cite{rr13}, Ruf and Runggaldier describe a systematic way of constructing market models which satisfy NUPBR but in which NFLVR fails.

We discuss these topics in the order in which they are mentioned above. We start our critical literature review with several necessary definitions in Section \ref{secsetup}, followed by a section discussing the relations between the different types of arbitrage --- see Section \ref{secarb}. Section \ref{secdiverse} discusses the literature regarding diversity, Section \ref{secintvol} is about intrinsically volatile models, and Section \ref{secrank} reviews the current state of the field studying rank-based models and coupled diffusions. The remaining sections treat the other topics: Section \ref{secoptarb} treats the developments in portfolio optimisation in SPT, Section \ref{sechedge} discusses the hedging of both European as well as American options in NUPBR markets, and Section \ref{secmeasure} discusses the absolutely continuous changes of measure that have been studied in the articles mentioned earlier. What follows is a section describing research results we have had so far, see Section \ref{secme}, and we finish off with a list of ideas for possible research directions in Section \ref{secdirections}.

\section{Set-up} \label{secsetup}
\subsection{Definitions}
We proceed as in \cite{fk09}, and place ourselves in a general continuous-time It\^o model without frictions (i.e.~there are no transaction costs, trading restrictions, or any other imperfections\footnote{These assumptions of frictionlessness restrict the implementability of this theory, and are important areas of future research --- see Section \ref{realimpl}.}); let the price processes $X_i(\cdot)$ of stocks $i=1,\ldots,n$ under the physical measure $\mathbb{P}$ be given by 
\begin{align} \label{model}
\mathrm{d}X_i(t) &= X_i(t)\Big(b_i(t)\mathrm{d}t + \sum_{\nu=1}^d \sigma_{i\nu}(t) \mathrm{d}W_\nu(t)  \Big),\qquad i=1,\ldots,n\\
X_i(0)&=x_i>0.\nonumber
\end{align}
Here, $W(\cdot)=(W_1(\cdot),\ldots,W_d(\cdot))$ is a $d$-dimensional $\mathbb{P}$-Brownian motion, and we assume $d\geq n$. We furthermore assume our filtration $\mathbb{F}$ to contain the filtration $\mathbb{F}^W$ generated by $W(\cdot)$, and the drift rate processes $b_i(\cdot)$ and matrix-valued volatility process $\sigma(\cdot)=(\sigma_{i\nu}(\cdot))_{i=1,\ldots,n, \nu=1,\ldots,d}$ to be $\mathbb{F}$-progressively measurable and to satisfy the integrability condition $$\sum_{i=1}^n \int_0^T \left( |b_i(t)| + \sum_{\nu=1}^d(\sigma_{i\nu}(t))^2 \right)\mathrm{d}t<\infty\quad \forall T\in(0,\infty).$$ We define the \emph{covariance process} $a(t)=\sigma(t)\sigma'(t)$, with the apostrophe denoting a transpose. Note that $a(\cdot)$ is a positive semi-definite matrix-valued process.
Finally, we assume the existence of a \emph{riskless asset} $X_0(t)\equiv 1,\ \forall t\geq0$; namely, without loss of generality we assume a zero interest rate, by discounting the stock prices by the bond price.

Now, let us consider the log-price processes; by It\^o's formula, we have
\begin{align} \label{logX}
\mathrm{d}\log X_i(t)&=\left(b_i(t)-\frac{1}{2}a_{ii}(t)\right)\mathrm{d}t + \sum_{\nu=1}^d \sigma_{i\nu} \mathrm{d}W_\nu(t) \nonumber \\&=\gamma_i(t)\mathrm{d}t + \sum_{\nu=1}^d \sigma_{i\nu} \mathrm{d}W_\nu(t),
\end{align}
where we have defined the \emph{growth rates} $\gamma_i(t):=b_i(t) - \frac{1}{2}a_{ii}(t)$. This name is justified by the fact that $$\lim_{T\to\infty} \frac{1}{T} \left( \log X_i(T) - \int_0^T \gamma_i(t) \mathrm{d}t  \right)=0\qquad \mathbb{P}\text{-a.s.};$$ see, for instance, Corollary 2.2 of \cite{f99}.

We proceed by defining which investment rules are allowed in our framework.
\begin{defn} \label{portfolio}
Define a \emph{portfolio} as an $\mathbb{F}$-progressively measurable vector process $\pi(\cdot)$, uniformly bounded in $(t,\omega)$, where $\pi_i(t)$ represents the proportion of wealth invested in asset $i$ at time $t$, and satisfying $\sum_{i=1}^n \pi_i(t)=1\ \forall t\geq0$.
We say that $\pi(\cdot)$ is a \emph{long-only} portfolio if $\pi_i(t)\geq0\ \forall i=1,\ldots,n$.
For future reference, we also define the set
\begin{equation} \label{delta+}
\Delta^n_+:=\{ x\in\mathbb{R}^n :  x_i>0\ \forall i=1,\ldots,n \}.
\end{equation}
We denote the \emph{wealth process} of an investor investing according to portfolio $\pi(\cdot)$, with initial wealth $w>0$, by $V^{w,\pi}(\cdot).$
\end{defn}
Note that portfolios are self-financing by definition.
We also define a more general class of investment rules, which we shall call trading strategies.
\begin{defn} \label{tradingstrat} 
A \emph{trading strategy} is an $\mathbb{F}$-progressively measurable process $h(\cdot)$ that takes values in $\mathbb{R}^n$ and satisfies the integrability condition $$\sum_{i=1}^n \int_0^T \Big( |h_i(t)b_i(t)| + h^2_i(t)a_{ii}(t) \Big) \mathrm{d}t< \infty\qquad \mathbb{P}\text{-a.s.}$$
For any $t$, $h_i(t)$ is the amount of money invested in stock $i$. Again, we let $V^{w,h}(\cdot)$ denote the wealth process of an investor following the trading strategy $h(\cdot)$ and starting with initial wealth $w\geq0$. We write $V^{h}(\cdot):=V^{1,h}(\cdot)$. We require $h(\cdot)$ to be $x$-admissible for some $x\geq0$, written as $h(\cdot)\in\mathcal{A}_x$, meaning that $V^{0,h}(t)\geq -x$ $\forall\, t\in[0,T]$ a.s. We shall write $\mathcal{A}:=\mathcal{A}_0$.
\end{defn}
Note that each portfolio generates a trading strategy by setting $h_i(t)=\pi_i(t)V^{w,\pi}(t)\ \forall t\in[0,T].$ We assume the admissibility condition to exclude doubling strategies. On the contrary, one can define a trading strategy $h(\cdot)\in\mathcal{A}_x$ by specifying it as the \emph{proportions} invested in stocks at each time, $\pi_i(t)=h_i(t)/V^{w,h}(t)$, provided that $w>x$ and similarly to a portfolio but with the exception that in general now $\sum_{i=1}^n \pi_i(t)\neq1$; that is, there is a non-zero holding of cash $\pi_0(t)$.

The wealth process associated to a portfolio $\pi(\cdot)$ and initial wealth $w\in\mathbb{R}_+$ can be seen to evolve as
\begin{equation}
\frac{\mathrm{d}V^{w,\pi}(t)}{V^{w,\pi}(t)}=\sum_{i=1}^n \pi_i(t) \frac{\mathrm{d}X_i(t)}{X_i(t)} = b_\pi(t)\mathrm{d}t + \sum_{\nu=1}^d \sigma_{\pi\nu}(t)\mathrm{d}W_\nu(t),
\end{equation}
with the portfolio's rate of return $b_\pi(t):= \sum_{i=1}^n \pi_i(t)b_i(t)$ and its volatility coefficients $\sigma_{\pi\nu}(t):=\sum_{i=1}^n \pi_i(t)\sigma_{i\nu}(t)$ (very slightly abusing notation). Hence we have, by It\^o's formula, that
\begin{align} \label{logV}
\mathrm{d}\log V^{w,\pi}(t)&=\left( b_\pi(t)-\frac{1}{2} \sum_{\nu=1}^d(\sigma_{\pi\nu}(t))^2 \right)\mathrm{d}t + \sum_{\nu=1}^d \sigma_{\pi \nu}(t)\mathrm{d}W_\nu(t)\nonumber\\
&=\gamma_\pi(t)\mathrm{d}t + \sum_{\nu=1}^d \sigma_{\pi \nu}(t)\mathrm{d}W_\nu(t),
\end{align}
where $\gamma_\pi(t):=b_\pi(t)-\frac{1}{2} \sum_{\nu=1}^d(\sigma_{\pi\nu}(t))^2$ is the \emph{growth rate} of the portfolio $\pi$. Note the disappearance of the drift processes from this expression (in \eqref{eqactuseful}); since we may write
\begin{align*}
\gamma_\pi(t)&=\sum_{i=1}^n\pi_i(t)b_i(t)-\frac{1}{2}\sum_{i,j=1}^n\pi_i(t)a_{ij}(t)\pi_j(t)=  \sum_{i=1}^n\pi_i(t)\gamma_i(t) + \gamma^*_\pi(t),
\end{align*}
where the \emph{excess growth rate} is defined as
\begin{equation}\label{defexcess}
\gamma^*_\pi(t):=\frac{1}{2}\left(\sum_{i=1}^n \pi_i(t)a_{ii}(t) - \sum_{i,j=1}^n \pi_i(t) a_{ij}(t) \pi_j(t) \right),
\end{equation}
it follows directly from equations \eqref{logX} and \eqref{logV} that
\begin{equation} \label{eqactuseful}
\mathrm{d}\log V^\pi(t)= \gamma^*_\pi(t)\mathrm{d}t + \sum_{i=1}^n \pi_i(t) \mathrm{d}\log X_i(t),
\end{equation}
which also motivates the nomenclature for $\gamma^*_\pi(\cdot)$.

We define a particular portfolio, the \emph{market portfolio} $\mu(\cdot)$, by
\begin{equation} \label{defXX}
\mu_i(t):=\frac{X_i(t)}{X(t)},\qquad X(t):= \sum_{i=1}^n X_i(t).
\end{equation}
We assume there is only one share per company (or, equivalently, that $X_i(\cdot)$ is the capitalisation process of company $i$), so $\mu_i(t)$ is the \emph{relative market weight} of company $i$ at time $t$. The wealth process associated to the market portfolio is
\begin{equation}
\frac{\mathrm{d}V^{w,\mu}(t)}{V^{w,\mu}(t)}=\sum_{i=1}^n \mu_i(t) \frac{\mathrm{d}X_i(t)}{X_i(t)} =\sum_{i=1}^n  \frac{X_i(t)}{X(t)}\frac{\mathrm{d}X_i(t)}{X_i(t)} = \frac{\mathrm{d} X(t)}{X(t)},
\end{equation}
and hence 
\begin{equation} \label{vx}
V^{w,\mu}(t)=\frac{w}{X(0)}X(t).
\end{equation}
The wealth resulting from the market portfolio is therefore simply equal to a constant times the total market size: $\mu(\cdot)$ is a buy-and-hold strategy. In SPT, one measures the performance of portfolios with respect to the market portfolio (i.e.~one uses the market portfolio as a `benchmark' --- this is similar to the approach taken in the Benchmark Approach to finance, developed by Platen and Heath \cite{ph06}). The market portfolio is therefore of great importance.

Equation \eqref{logV} gives that
\begin{equation}
\mathrm{d}\log V^{w,\mu}(t)=\gamma_\mu(t)\mathrm{d}t + \sum_{\nu=1}^d \sigma_{\mu \nu}(t)\mathrm{d}W_\nu(t),
\end{equation}
which, together with equations \eqref{logX} and \eqref{vx}, gives that 
\begin{equation} \label{dynlogmu}
\mathrm{d}\log \mu_i(t)=\left(\gamma_i(t)-\gamma_\mu(t)  \right)\mathrm{d}t +\sum_{\nu=1}^d \left(\sigma_{i \nu}(t)-\sigma_{\mu \nu}(t)\right)\mathrm{d}W_\nu(t).
\end{equation}
Equivalently, the relative market weights evolve as
\begin{align}\label{dynmu}
\frac{\mathrm{d}\mu_i(t)}{\mu_i(t)}&=\Big(\gamma_i(t)-\gamma_\mu(t) +\frac{1}{2}\sum_{\nu=1}^d \left(\sigma_{i \nu}(t)-\sigma_{\mu \nu}(t)\Big)^2 \right)\mathrm{d}t +\sum_{\nu=1}^d \Big(\sigma_{i \nu}(t)-\sigma_{\mu \nu}(t)\Big)\mathrm{d}W_\nu(t)\nonumber\\
&= \Big(\gamma_i(t)-\gamma_\mu(t) +\frac{1}{2}\tau_{ii}^\mu(t) \Big)\mathrm{d}t +\sum_{\nu=1}^d \Big(\sigma_{i \nu}(t)-\sigma_{\mu \nu}(t)\Big)\mathrm{d}W_\nu(t).
\end{align}
Here, we have defined the matrix-valued \emph{covariance process of the stocks relative to the portfolio $\pi(\cdot)$} as
\begin{align}\label{deftau}
\tau^\pi_{ij}(t):&=\sum_{\nu=1}^d (\sigma_{i \nu}(t)-\sigma_{\pi \nu}(t))(\sigma_{j \nu}(t)-\sigma_{\pi \nu}(t))=(\pi(t)-e_i)'a(t)(\pi(t)-e_j)\nonumber\\
&=a_{ij}(t)-a_{\pi i}(t)-a_{\pi j}(t)+a_{\pi\pi}(t),
\end{align}
where $e_i$ is the $i$-th unit vector in $\mathbb{R}^n$, and
\begin{equation}
a_{\pi i}(t):= \sum_{j=1}^n \pi_j(t)a_{ij}(t),\qquad a_{\pi \pi}(t):= \sum_{i,j=1}^n \pi_i(t)\pi_j(t)a_{ij}(t).
\end{equation}
Note that we have the following relation:
\begin{align} \label{zeroo}
\sum_{j=1}^n \pi_j(t)\tau^\pi_{ij}(t)&=\sum_{j=1}^n \pi_j(t)a_{ij}(t)-a_{\pi i}(t)-\sum_{j=1}^n \pi_j(t)a_{\pi j}(t)+a_{\pi\pi}(t)\nonumber\\
&=0,\qquad  i=1,\ldots,n
\end{align}
since the first two and last two terms cancel each other. Finally, note also that
\begin{equation} \label{dmumu}
\tau^\mu_{ij}(t) = \frac{\mathrm{d}\left<\mu_i,\mu_j\right>(t)}{\mu_i(t)\mu_j(t)\mathrm{d}t},\qquad  1\leq i,j\leq n.
\end{equation}

We now give the definition of a relative arbitrage:
\begin{defn} \label{defRA}
\textbf{(Relative arbitrage)} Let $h(\cdot)$ and $k(\cdot)$ be trading strategies. Then $h(\cdot)$ is called a \emph{relative arbitrage} (RA) over $[0,T]$ with respect to $k(\cdot)$ if their associated wealth processes satisfy $$V^h(T) \geq V^k(T)\ \text{ a.s.,} \qquad \mathbb{P}(V^h(T) > V^k(T))>0.$$
\end{defn}
Usually, we will only consider and construct relative arbitrages using portfolios that do not invest in the riskless asset at all. However, it is also possible to create a RA using a trading strategy that has a non-trivial position in the riskless asset, as we show in the following example, which uses results from Ruf \cite{ruf13a} on hedging European claims in Markovian markets where NA is allowed to fail (see Section \ref{hedgeEU}).

\begin{exmp} \label{rufex}
Define an auxiliary process $R(\cdot)$ as a Bessel process with drift $-c$, i.e.
$$\textrm{d}R(t)=\Big(\frac{1}{R(t)}-c  \Big) \mathrm{d}t + \mathrm{d}W(t)$$
for $t\in[0,T]$, $c\geq0$ constant and $W(\cdot)$ a BM. We have that the Bessel process $R(\cdot)$ is strictly positive. Define a stock price process by
$$\textrm{d}S(t)=\frac{1}{R(t)} \mathrm{d}t + \mathrm{d}W(t),\qquad S(0)=R(0)>0$$
for $t\in[0,T]$, so $S(t) = R(t)+ct>0\ \forall t\in[0,T]$. The market price of risk is $\theta(t,s)=1/(s-ct)$ for $(t,s)\in[0,T]\times\mathbb{R}_+$ with $s>ct$. 
Thus, by Corrolary 5.2 of \cite{ruf13a}, the reciprocal $1/Z^\theta(\cdot)$ of the local martingale deflator (see Definition \ref{defZ}) hits zero exactly when $S(t)$ hits $ct$.
For a general payoff function $p$, and $(t,s)\in[0,T]\times\mathbb{R}_+$ with $s>ct$, Theorem 5.1 of \cite{ruf13a} implies that a claim paying $p(S(T))$ at time $t=T$ has value function
\begin{align} \label{eq}
h^p(t,s) :&= \mathbb{E}^{t,s}[\tilde{Z}^{\theta,t,s}(T)p(S(T))]\\
&= \mathbb{E}^\mathbb{Q}[ p(S(T))\textbf{1}_{\{ \min_{t\leq u \leq T}\{S(u)-cu\}>0  \}}  \big| \mathcal{F}(t)] \, \Big|_{S(t)=s}\nonumber\\
&= \int_{\frac{cT-s}{\sqrt{T-t}}}^\infty \frac{e^{-z^2/2}}{\sqrt{2\pi}} p(z\sqrt{T-t}+s)\textrm{d}z\\
&\qquad-e^{2c(s-ct)}\int_{\frac{cT-2ct+s}{\sqrt{T-t}}}^\infty \frac{e^{-z^2/2}}{\sqrt{2\pi}} p(z\sqrt{T-t}-s+2ct)\textrm{d}z.\nonumber
\end{align}

Now define another stock price process by
$$\textrm{d}\tilde{S}(t)=-\tilde{S}^2(t)\mathrm{d}W(t),$$
so $\mathbb{P}$ is already a martingale measure for $\tilde{S}(\cdot)$. We have $\tilde{S}(\cdot)=1/S(\cdot)$, with $c=0$, and also $\tilde{\theta}(\cdot)\equiv 0$, so $Z^{\tilde{\theta}}(\cdot)\equiv1$.
Applying It\^{o}' formula, note that
$$\mathrm{d}\log{\tilde{S}(t)}=-\tilde{S}(t)\mathrm{d}W(t)-\frac{1}{2}\tilde{S}^2(t)\mathrm{d}t=\mathrm{d}\log{Z^\theta(t)};$$
hence $\tilde{S}(t)=\tilde{S}(0)Z^\theta(t)$ and $$\tilde{Z}^{\theta,t,s}(T)=\frac{\tilde{S}(T)}{\tilde{S}(t)}\, \Bigg|_{\tilde{S}(t)=1/s}.$$
Thus, using Theorem 4.1 of \cite{ruf13a} and \eqref{eq} with $c=0$, we may compute the hedging price of one unit of this stock as
\begin{align*}
\nu^1(t,s):&=\mathbb{E}^{t,s}[\tilde{Z}^{\tilde{\theta},t,s}(T)\tilde{S}(T)] = \mathbb{E}^{t,s}[\tilde{S}(T)]\\ &= \mathbb{E}^{t,s}[\tilde{Z}^{\theta,t,1/s}(T)\tilde{S}(t)] = s\mathbb{E}^{t,s}[\tilde{Z}^{\theta,t,1/s}(T)]\\
&= s\cdot\left( \int_{\frac{-1/s}{\sqrt{T-t}}}^\infty \frac{e^{-z^2/2}}{\sqrt{2\pi}} \textrm{d}z-\int_{\frac{1/s}{\sqrt{T-t}}}^\infty \frac{e^{-z^2/2}}{\sqrt{2\pi}} \textrm{d}z\right) \\
&= 2s\Phi\left(\frac{1}{s\sqrt{T-t}}\right)-s<s.
\end{align*}
In other words, this stock has a ``bubble''. By Theorem 4.1 of \cite{ruf13a}, the corresponding optimal strategy (expressed in the \emph{number} of stocks the investor holds) is the derivative of the hedging price with respect to $s$, i.e.
$$\eta^1(t,s)=2\Phi\left(\frac{1}{s\sqrt{T-t}}\right)-1-\frac{2}{s\sqrt{T-t}}\phi\left(\frac{1}{s\sqrt{T-t}}\right)<1$$
for $(t,s)\in[0,T)\times\mathbb{R}_+.$

Now $\eta^1(\cdot,\cdot)$ is a relative arbitrage with respect to $\eta^2(t,s):=1$ (i.e.~just holding the stock). Namely, define $\bar{\nu}:=\nu^1(0,\tilde{S}_0)$; since $${V}^{\bar{\nu},\eta^2}(T)=\bar{\nu}\tilde{S}(T)<\tilde{S}(T)={V}^{\bar{\nu},\eta^1}(T)\ \text{ a.s.},$$
we see that $\eta^1(\cdot,\cdot)$ is a relative arbitrage with respect to $\eta^2(\cdot,\cdot)$. However, it is \emph{not} a `real' arbitrage, since for $\hat{\eta}(\cdot,\cdot):=\eta^1(\cdot,\cdot)-\eta^2(\cdot,\cdot)$ we have ${V}^{\bar{\nu},\hat{\eta}}(0)=0$ and ${V}^{\bar{\nu},\hat{\eta}}(T)=(1-\bar{\nu})\tilde{S}(T)>0$, but since $\eta^1(\cdot,\cdot)<1$ for $t\in[0,T)$, we get that $\hat{\eta}(\cdot,\cdot)<0$ for $t\in[0,T)$ and thus the wealth process is unbounded below; i.e.~$\hat{\eta}$ is not admissible.

The holding in the riskless asset $\phi(\cdot)$ corresponding to strategy $\eta^1(\cdot,\cdot)$ can be computed using the self-financing equation $\mathrm{d}{V}=\phi \mathrm{d}B + \eta^1 \mathrm{d}\tilde{S}= \eta^1 \mathrm{d}\tilde{S}$ and ${V}=\phi B + \eta \tilde{S}$, which gives that $$\phi(t)={V}(t)-\eta^1(t,\tilde{S}(t))\tilde{S}(t)=\int_0^t \eta^1(u,\tilde{S}(u))\textrm{d}\tilde{S}(u)-\eta^1(t,\tilde{S}(t))\tilde{S}(t),$$
which can, given the history up to time $t$, be computed. Note that $\phi(\cdot)$ is not Markovian, and is in general non-zero. \qed

\end{exmp}


\subsection{Derivation of some useful properties} \label{deruseful}
We now give the proofs from \cite{fk09} of two lemmas which will be essential in constructing relative arbitrages later. Let us start by defining the \emph{relative returns process} of stock $i$ with respect to portfolio $\pi(\cdot)$ as
\begin{equation}
R^\pi_i(t):=\log\bigg(\frac{X_i(t)}{V^{w,\pi}(t)} \bigg)\bigg|_{w=X_i(0)}.
\end{equation}
We shall use this process to show that the variance of a stock with respect to a portfolio is always positive, which will be useful in Lemma \ref{lem2}.
\begin{lem} \label{lem1}
We have that $\tau^\pi_{ii}(t)=\frac{\mathrm{d}}{\mathrm{d}t}\left< R^\pi_i \right>(t)\geq 0$.
\end{lem}
\begin{proof}
Using equations \eqref{logX} and \eqref{logV}, we get that
$$\textrm{d} R^\pi_i(t)=\left(\gamma_i(t)-\gamma_\pi(t)  \right)\mathrm{d}t +\sum_{\nu=1}^d \left(\sigma_{i \nu}(t)-\sigma_{\pi \nu}(t)\right)\mathrm{d}W_\nu(t).$$
From this and defining equation \eqref{deftau}, we see that $\tau^\pi_{ij}(t)=\frac{\textrm{d}}{\mathrm{d}t}\left< R^\pi_i,R^\pi_j\right>(t),$ and thus $$\tau^\pi_{ii}(t)=\frac{\textrm{d}}{\mathrm{d}t}\left< R^\pi_i \right>(t)\geq 0.$$
\end{proof}
We use this to prove the following, which says that we may simply replace the covariance matrix in \eqref{defexcess} by the matrix of covariances \emph{relative to any portfolio}:
\begin{lem} \label{lem2}
We have the \emph{num\'eraire-invariance} property
\begin{equation}\label{numinv}
\gamma^*_\pi(t)=\frac{1}{2}\Big( \sum_{i=1}^n\pi_i(t)\tau^\rho_{ii}(t)-\sum_{i,j=1}^n \pi_i(t) \pi_j(t) \tau^\rho_{ij}(t) \Big) 
\end{equation}
for any two portfolios $\pi(\cdot)$ and $\rho(\cdot)$. In particular, we have that 
\begin{equation}\label{gamma+}
\gamma^*_\pi(t)= \frac{1}{2}\sum_{i=1}^n \pi_i(t) \tau^\pi_{ii}(t),
\end{equation}
which is non-negative for any long-only portfolio $\pi(\cdot)$.
\end{lem}
\begin{proof}
By definition of $\tau^\rho_{ij}(t)$, equation \eqref{deftau}, we have that
\begin{equation} \label{sumone}
\sum_{i=1}^n\pi_i(t)\tau^\rho_{ii}(t)=\sum_{i=1}^n\pi_i(t)a_{ii}(t)-2\sum_{i=1}^n\pi_i(t)a_{\rho i}(t)+a_{\rho\rho}(t)
\end{equation}
and
\begin{equation} \label{sumtwo}
\sum_{i,j=1}^n \pi_i(t) \pi_j(t) \tau^\rho_{ij}(t)=\sum_{i=1}^n\pi_i(t)\pi_j(t)a_{ij}(t)-2\sum_{i=1}^n\pi_i(t)a_{\rho i}(t)+a_{\rho\rho}(t).
\end{equation}
Putting equations \eqref{sumone} and \eqref{sumtwo} together, we see that
\begin{align*}
\frac{1}{2}\Big( \sum_{i=1}^n\pi_i(t)\tau^\rho_{ii}(t)-\sum_{i,j=1}^n \pi_i(t) \pi_j(t) \tau^\rho_{ij}(t) \Big)&=\frac{1}{2}\Big( \sum_{i=1}^n\pi_i(t)a_{ii}(t)-\sum_{i,j=1}^n \pi_i(t) \pi_j(t) a_{ij}(t) \Big)\\&=\gamma^*_\pi(t)
\end{align*}
by definition \eqref{defexcess}, proving the first statement.
Now, choosing $\rho(\cdot)=\pi(\cdot)$, and recalling relation \eqref{zeroo}, we see that we may write the excess growth rate as
\begin{equation} \label{gammaWA}
\gamma^*_\pi(t)=\frac{1}{2} \sum_{i=1}^n\pi_i(t)\tau^\pi_{ii}(t);
\end{equation}
that is, as the weighted average of the stocks' variances relative to $\pi(\cdot)$.
Finally, using equation \eqref{gammaWA}, Definition \ref{portfolio} and Lemma \ref{lem1}, we conclude that for all long-only portfolios we have $\gamma^*_\pi(t)\geq0.$ 
\end{proof}


Note that for $\pi(\cdot)=\mu(\cdot)$, we get from equation \eqref{gammaWA} that the excess growth rate of the market portfolio is
\begin{equation} \label{intvol}
\gamma^*_\mu(t)=\frac{1}{2}\sum_{i=1}^n \mu_i(t) \tau^\mu_{ii}(t),
\end{equation}
namely the weighted average of the stocks' variances relative to the market. This is interpreted as a measure of the `intrinsic volatility' of the market.

As this will be useful later, let us introduce some notation:
\begin{defn} \label{defrevorder}
We shall use the reverse-order-statistics notation, defined by
\begin{align}\label{deforder}
\theta_{(1)}(t)&:=\max_{1\leq i \leq n}\{\theta_i(t)\}\nonumber \\
\theta_{(i)}(t)&:=\max \big( \{\theta_1(t),\ldots,\theta_n(t)  \}\setminus\{ \theta_{(1)}(t),\ldots,\theta_{(i-1)}(t) \} \big), \quad i=2,\ldots,n
\end{align}
for any $\mathbb{R}^n$-valued process $\theta(\cdot)$. Thus we have
\begin{equation} 
 \theta_{(1)}(t)\geq \theta_{(2)}(t)\geq \ldots \geq \theta_{(n)}(t).
\end{equation}
\end{defn}

\subsection{Functionally generated portfolios} \label{fgportfolios}
The biggest advantage of SPT over classical approaches to constructing well-performing portfolios is that in general it does not require estimation of the drifts or volatilities of the stocks. The machinery of SPT, i.e., the way in which virtually all relative arbitrages are constructed, involves what Fernholz (see Definition 3.1 in \cite{f95}) has called \emph{functionally generated portfolios} (FGPs):\footnote{We denote the derivative with respect to the $i$th coordinate by $D_i$.}
\begin{defn} \label{defFGP}
Let $U\subset \Delta^n_+$ be a given open set. Call $\mathbf{G}\in\mathcal{C}^2(U,(0,\infty))$ a \emph{generating function} for the portfolio $\pi(\cdot)$ if $\mathbf{G}$ is such that $x\mapsto x_iD_i\log\mathbf{G}(x)$ is bounded on $U$, and if there exists a measurable, adapted process $\mathfrak{g}(\cdot)$ such that 
\begin{equation} \label{defeqfgp}
\mathrm{d}\log \left(\frac{V^\pi(t)}{V^\mu(t)}\right)=\mathrm{d}\log\mathbf{G}(\mu(t))+ \mathfrak{g}(t)\mathrm{d}t, \qquad \forall t\geq0,\quad \text{a.s.}
\end{equation}
\end{defn}
We can interpret the above equation as follows: the process measuring the performance of the portfolio $\pi(\cdot)$ relative to the market (the LHS of \eqref{defeqfgp}) can be decomposed into a stochastic part of infinite variation, written as a deterministic function of the market weights process, plus a finite variation part $\mathfrak{g}(t)\mathrm{d}t$.
In fact, Theorem 3.1 of \cite{f95} shows that Definition \ref{defFGP} is equivalent to the following:
\begin{prop} \label{propfern}
Let a function $\mathbf{G}$ as in Definition \ref{defFGP} generate the portfolio $\pi(\cdot)$. Then we have the following expression, for $i=1,\ldots,n$:
\begin{equation} \label{deffgp}
\pi_i(t)=\Big( D_i\log\mathbf{G}(\mu(t)) + 1 - \sum_{j=1}^n \mu_j(t)D_j\log\mathbf{G}(\mu(t)) \Big)\cdot \mu_i(t).
\end{equation}
\end{prop}
Note that this defines a portfolio indeed, in particular, $\sum_{i=1}^n\pi(t) =1$.
We present the proof of the reverse direction to Proposition \ref{propfern}, as given in \cite{fk09}.
\begin{lem} \label{master}
For a portfolio $\pi(\cdot)$ satisfying \eqref{deffgp}, we have that $\pi(\cdot)$ is generated by $\mathbf{G}$, i.e.
\begin{equation} \label{eqmaster} 
\log\left(\frac{V^\pi(T)}{V^\mu(T)} \right) =\log \left(\frac{\textnormal{\textbf{G}}(\mu(T))}{\textnormal{\textbf{G}}(\mu(0))} \right) + \int_0^T \mathfrak{g}(t)\mathrm{d}t\ \textnormal{ a.s.},
\end{equation}
where 
\begin{equation} \label{defdrift}
\mathfrak{g}(t):= \frac{-1}{2\textnormal{\textbf{G}}(\mu(t))}\sum_{i,j=1}^n D_{ij}^2 \textnormal{\textbf{G}}(\mu(t))\mu_i(t)\mu_j(t)\tau^\mu_{ij}(t)
\end{equation} 
is called the \emph{drift process}.
\end{lem}
\begin{proof}
\textbf{Step I} First, let us prove a useful expression for the term on the LHS of \eqref{eqmaster}, namely equation \eqref{usefe}. In general, we have from equation \eqref{eqactuseful} that 
\begin{equation}
\mathrm{d}\log\left(\frac{V^\pi(T)}{V^\rho(T)} \right) = \gamma^*_\pi(t)\mathrm{d}t + \sum_{i=1}^n \pi_i(t) \mathrm{d}\log\left(\frac{ X_i(t)}{ V^\rho(T)}\right).
\end{equation}
Setting $\rho(\cdot)=\mu(\cdot)$ and recalling equation \eqref{vx}, this becomes
\begin{equation}
\mathrm{d}\log\left(\frac{V^\pi(T)}{V^\mu(T)} \right) = \gamma^*_\pi(t)\mathrm{d}t + \sum_{i=1}^n \pi_i(t) \mathrm{d}\log\mu_i(t).
\end{equation}
Now, recall expressions \eqref{dynlogmu} and \eqref{dynmu} for the dynamics of $\log \mu_i(\cdot)$ and $\mu_i(\cdot)$, respectively, and apply the num\'eraire-invariance property \eqref{numinv} to get
\begin{align}\label{usefe}
\mathrm{d}\log\left(\frac{V^\pi(T)}{V^\mu(T)} \right) = \sum_{i=1}^n \frac{\pi_i(t)}{\mu_i(t)}\mathrm{d}\mu_i(t)-\frac{1}{2} \sum_{i,j=1}^n {\pi_i(t)\pi_j(t)}\tau^{\mu}_{ij}(t)\, \mathrm{d}t.
\end{align} 

\textbf{Step II} In order for us to relate $V^\pi(T)/V^\mu(T)$ to $\mathbf{G}(\mu(0))$ and $\mathbf{G}(\mu(T))$, we need to derive a useful expression for the dynamics of $\log \mathbf{G}(\mu(\cdot))$. Note the relation $$D^2_{ij}\log \mathbf{G}(\mu(t))=\frac{D^2_{ij}\mathbf{G}(\mu(t))}{\mathbf{G}(\mu(t))}-D_i \mathbf{G}(\mu(t))\cdot D_j\mathbf{G}(\mu(t))$$ and introduce the notation $g_i(t):=D_i\log \mathbf{G}(\mu(t))$, $N(t):=1-\sum_{j=1}^n \mu_j(t)g_j(t)$; then we have, using relation \eqref{dmumu}, that
\begin{align}\label{dlogG}
\mathrm{d}\log \mathbf{G}(\mu(t))&= \sum_{i=1}^n g_i(t)\mathrm{d}\mu_i(t) + \frac{1}{2}\sum_{i,j=1}^n D^2_{ij}\log \mathbf{G}(\mu(t))\, \mathrm{d}\left<\mu_i,\mu_j\right>(t)\\
&=  \sum_{i=1}^n g_i(t)\mathrm{d}\mu_i(t) + \frac{1}{2}\sum_{i,j=1}^n \left( \frac{D^2_{ij}\mathbf{G}(\mu(t))}{\mathbf{G}(\mu(t))}-g_i(t)g_j(t) \right) \mu_i(t)\mu_j(t)\tau^\mu_{ij}(t)\, \mathrm{d}t.\nonumber
\end{align} 

\textbf{Step III} Finally, note that with our temporary notation, the defining equation \eqref{deffgp} becomes $\pi_i(t)=(g_i(t)+N(t))\mu_i(t)$; we compute 
\begin{equation}
\sum_{i=1}^n \frac{\pi_i(t)}{\mu_i(t)}\mathrm{d}\mu_i(t) = \sum_{i=1}^n g_i(t)\mathrm{d}\mu_i(t) + N(t)\mathrm{d}\Big( \sum_{i=1}^n \mu_i(t) \Big) = \sum_{i=1}^n g_i(t)\mathrm{d}\mu_i(t) 
\end{equation}
and
\begin{align}
\sum_{i,j=1}^n {\pi_i(t)\pi_j(t)}\tau^{\mu}_{ij}(t) &=  \sum_{i,j=1}^n (g_i(t)+N(t))(g_j(t)+N(t)) \mu_i(t)\mu_j(t)\tau^\mu_{ij}(t)\nonumber\\
&= \sum_{i,j=1}^n g_i(t)g_j(t) \mu_i(t)\mu_j(t)\tau^\mu_{ij}(t),
\end{align}
where we use relation \eqref{zeroo} in the final step. Hence, equation \eqref{usefe} becomes
\begin{align*}
\mathrm{d}\log\left(\frac{V^\pi(T)}{V^\mu(T)} \right) = \sum_{i=1}^n g_i(t)\mathrm{d}\mu_i(t) -\frac{1}{2} \sum_{i,j=1}^n g_i(t)g_j(t) \mu_i(t)\mu_j(t)\tau^\mu_{ij}(t)\, \mathrm{d}t,
\end{align*} 
and the result follows by comparison with equation \eqref{dlogG} and definition \eqref{defdrift}.
\end{proof}

\begin{rem}
The importance of this result cannot be overstated, as it allows us to relate observed properties of markets (and thus conditions on the behaviour of certain processes over time) to the relative performance of a portfolio compared to the market portfolio. By choosing a suitable generating function $\textbf{G}$, the first term on the RHS of \eqref{eqmaster} can be bounded from below. Furthermore, the volatility processes only appear in the drift process $\mathfrak{g}(\cdot)$, and the drift processes do not appear at all in \eqref{eqmaster}. This will be our method of constructing relative arbitrage opportunities.
\end{rem}

Generalisations of FGPs have been proposed in \cite{strong13}, in which the author demonstrates how the generating function might be made to depend on additional arguments which are processes of finite variation (for instance time, or live information from twitter feeds), how one could benchmark with respect to a portfolio different to the market portfolio, and how such changes would modify the master equation \eqref{eqmaster}. These generalised FGPs have not found an application in the literature yet, and could possibly offer a framework for studying FGPs which incorporate insider information or observations.

One open problem put forward in \cite[Remark 11.5]{fk09} is whether there exist relative arbitrages that are not functionally generated. This question has been answered by Pal and Wong in two different ways, depending on how the question is interpreted:
\begin{itemize}
\item In their paper \cite{palw13}, the authors take an information theoretic approach to portfolio performance analysis (discussed in Section \ref{futinfo}), and show that, under certain assumptions, there definitely do exist so-called energy-entropy portfolios which beat the market for sufficiently long time horizons, but are of finite variation and depend on the entire history of the stock prices, and are therefore \emph{not} functionally generated;
\item In the paper \cite{palw14} it is proven that if one restricts to the class of portfolios that merely depend on the current market capitalisations, a slight generalisation of functionally generated portfolios is the only class that can lead to a relative arbitrage.
\end{itemize}

\section{No-arbitrage conditions}  \label{secarb} 
There are several notions of arbitrage, and corresponding assumptions of the non-existence of these, which are relevant in the context of SPT. Relations between various no arbitrage conditions and the existence of local martingale deflators have been proved in several papers --- Fontana \cite{font13} summarises and reproves many of these relations.

\subsection{Notions of arbitrage and deflators} \label{notionsarbitrage}
We first define the relevant types of arbitrage, using Definition 4.1 of \cite{kk07}, in which the concept of an Unbounded Profit with Bounded Risk (UPBR) was first put forward. 

\begin{defn} \label{defarb}
Consider a time horizon $[0,T]$, where $T\leq \infty$. We define the following notions of arbitrage:
\begin{itemize}
\item A strategy $h(\cdot)\in\mathcal{A}_x$, $x\geq0$, is an $x$-\emph{arbitrage} if $V^{x,h}(T)\geq x,\ \mathbb{P}$-a.s., and $\mathbb{P}(V^{x,h}(T)>x)>0$, and a \emph{strong} or \emph{scalable arbitrage} if this holds for $x=0$.
\item A market satisfies \emph{No Unbounded Profit with Bounded Risk} (NUPBR) if\footnote{Recall, from Definition \ref{tradingstrat}, that $\mathcal{A}$ is the set of admissible trading strategies with zero initial wealth, i.e.~those with nonnegative corresponding wealth processes.} $$\lim_{c\to\infty}\sup_{h\in\mathcal{A}}\mathbb{P}(V^{0,h}(T)>c)=0.$$
\item A sequence $(h_n(\cdot))_{n\in\mathbb{N}}\in\mathcal{A}$ is said to be a \emph{free lunch with vanishing risk} (FLVR) if there exist an $\varepsilon > 0$ and an increasing sequence $(\delta_n)_{n\in\mathbb{N}}$ with $0\leq\delta_n \uparrow 1$, such that $V^{0,h_n}(T)>\delta_n\ \mathbb{P}$-a.s.\! and $\mathbb{P}(V^{0,h_n}(T) > 1+\varepsilon)\geq \varepsilon$. 
\item A market allows Immediate Arbitrage (IA) if there exists a stopping time $\tau$ with $\mathbb{P}(\tau<T)>0$, and a trading strategy $h(\cdot)\in\mathcal{A}$ supported by $(\tau,T]$, i.e.~$h(t)=h(t)\mathbbm{1}_{(\tau,T]}$, such that $V^{0,h}(t)>0\ \mathbb{P}$-a.s. $\forall\, t\in(\tau,T]$.
\end{itemize}
\end{defn}

We recall that an equivalent local martingale measure (ELMM) is a probability measure $\mathbb{Q}$ equivalent to the physical measure $\mathbb{P}$ with the property that the discounted price process is a local martingale under $\mathbb{Q}$. By the Fundamental Theorem of Asset Pricing (FTAP), see Corollary 1.2 in \cite{ds94}, the NFLVR condition is equivalent to the existence of an ELMM. Also, as is shown in Proposition 4.2 of \cite{kk07}, the NFLVR condition holds if and only if both the no-arbitrage (NA) and NUPBR conditions hold. \cite{kard12} shows that an UPBR is equivalent to the perhaps more familiar arbitrage of the first kind. Furthermore, Lemma 3.1 of \cite{ct13} proves that NUPBR implies NIA, and from Lemma 3.1 of \cite{ds95b} we conclude that the only possible arbitrage opportunity in an NUPBR market is a non-scalable one. Proposition 2.4 of Fontana and Runggaldier's \cite{fr13} shows that NIA holds if and only if there exists a \emph{market price of risk} (MPR), i.e.~some $\mathbb{R}^d$-valued progressively measurable process $\theta(\cdot)$ such that
\begin{equation} \label{defmpr}
b(t)-r(t)\textbf{1}=\sigma(t)\theta(t)\quad \qquad \mathbb{P}\otimes\text{Leb}\text{-a.e.}
\end{equation}

As the NFLVR condition is allowed to fail in SPT, and only the weaker NUPBR condition is assumed, an ELMM is not guaranteed to exist. The following object will be of greater interest and use to us:
\begin{defn} \label{defZ}
A non-negative process $Z(\cdot)$ is called a \emph{local martingale deflator} (LMD) if it satisfies $Z(0)=1$ and $Z(T)>0\ \mathbb{P}$-a.s., and $Z(\cdot)V^{0,h}(\cdot)$ is a $\mathbb{P}$-local martingale for all $h\in\mathcal{A}$.
\end{defn}
In \cite{kard12} it is shown that the NUPBR condition is equivalent to the existence of at least one LMD.
For a general It\^o model \eqref{model} with the NUPBR condition imposed, we know by Lemma 3.1 of \cite{ct13} that NIA holds, which in turn implies that an MPR exists. If we further assume that there exists a \emph{square-integrable} MPR, i.e.~a $\theta(\cdot)$ satisfying \eqref{defmpr} as well as
$$\int_0^T ||\theta(t)||^2\mathrm{d}t<\infty\quad \text{a.s.} \quad \forall T>0,$$
then it is well-known that the exponential local martingale 
\begin{equation}\label{defItoLMD}
Z^\theta(t):= \exp\left\{ -\int_0^t \theta'(s)\mathrm{d}W(s) -\frac{1}{2} \int_0^t ||\theta(s)||^2\mathrm{d}s \right\}, \quad 0\leq t\leq T,
\end{equation}
is a local martingale deflator. Recall from standard theory that $\mathbb{E}[Z^\theta(T)]=1$ (so $Z^\theta(\cdot)$ is a martingale) if and only if there exists an ELMM; the LMD is then simply the Radon-Nikodym density of the ELMM. We make the following assumption in the remainder of this thesis, which implies NUPBR by the above:

\paragraph*{Standing Assumption}
There exists a square-integrable MPR $\theta(\cdot)$.

\subsection{The existence of relative arbitrage} \label{RAexists}
With the above definitions and relations in place, we ask ourselves the following question: in which It\^o models do relative arbitrages with respect to the market exist?\footnote{As was pointed out to the author by Johannes Ruf, relative arbitrage exists in almost any market, since one can follow a `suicide strategy' which almost surely loses all its money, and thus construct an arbitrage relative to such a strategy.} This question remains largely open, as general (deterministic) conditions on a market model in order for relative arbitrage opportunities to exist have not been found yet. Some progress has been made in the one-dimensional case (i.e.~$n=1$; the case of one stock) by \cite{mu10}, where the authors show the equivalence of the existence of market-relative arbitrage with explicit conditions on the drift and volatility processes $b(\cdot)$ and $\sigma(\cdot)$. It would, however, be very interesting and useful to have a more general result for higher-dimensional markets; and, above all, to have conditions which are easy to check, and do not require knowledge of the drift and volatility processes.

Johannes Ruf, in his Theorem 8 of \cite{ruf11}, proved the following more general characterisation of relative arbitrages in general NUPBR markets: 
\begin{lem} \label{RAequiv}
Let $T>0$ and consider a trading strategy $h(\cdot)\in\mathcal{A}_{\tilde{p}}$ for an initial wealth $\tilde{p}>0$. Then there exists a relative arbitrage opportunity with respect to $h(\cdot)$ over the time horizon $[0,T]$ if and only if $$\mathbb{E}[Z^\nu(T)V^{\tilde{p},h}(T)]< \tilde{p}$$ for all market prices of risk $\nu(\cdot)$.
\end{lem}
If we take $h_i(t)=0,\ i=1,\ldots,n,\ h_0(t)=\tilde{p},\ \forall t\in[0,T]$, so all the money is invested in the riskless asset, then this lemma gives that there exists a non-scalable arbitrage (or $1$-arbitrage) opportunity if and only if all LMDs are strict local martingales. 
For arbitrages relative to the market we get the following: these exist if and only if $$\mathbb{E}[Z^\nu(T)X(T)]< X(0)$$ for all market prices of risk $\nu(\cdot)$, i.e.~if and only if $Z(\cdot)X(\cdot)$ is a strict local martingale.

The following is a reformulation of Proposition 6.1 of \cite{fk09}, which strengthens one direction of Lemma \ref{RAequiv} in the case that an additional assumption on the volatility structure holds:
\begin{prop} \label{prop63}
Suppose the following \emph{bounded volatility} condition holds:
\begin{equation} \label{BV} \tag{BV}
\exists\, K>0\ \textnormal{ such that }\ \xi'a(t)\xi\leq K ||\xi||^2, \quad \forall\xi\in\mathbb{R}^n,t\geq0 \qquad\mathbb{P}\textnormal{-a.s.}.
\end{equation}
Then the existence of a relative arbitrage with respect to the market implies that all local martingale deflators are strict local martingales.\footnote{Note that, by \cite{kard12}, there exists at least one LMD by the NUPBR assumption.}
\end{prop}
The following example shows how this proposition fails if we allow the volatility to be unbounded.
\begin{exmp}
Let us consider the following one-dimensional stock price process, taken from Cox and Hobson \cite{ch05}: 
\begin{equation} \label{exch}
\mathrm{d}S(t)=\frac{S(t)}{\sqrt{T-t}}\mathrm{d}W(t),\ t\in(0,T), \qquad S(0)=s>0,
\end{equation}
where $W(\cdot)$ is a Brownian motion under the considered measure. Then $S(\cdot)$ is a true martingale over $[0,s]$ for all $s<T$, but $S(T)=0$ a.s. This is an example of a so-called ``bubble'', and we can make a relative arbitrage in the following way:
\begin{itemize}
\item Define the portfolio $\pi(t):=0\ \forall t\in[0,T]$, i.e.~an investor following $\pi(\cdot)$ invests all his wealth in the money market;
\item Let $\rho(t):= 1\ \forall t\in[0,T]$; this is a buy-and-hold strategy in which an investor simply puts all his initial wealth into the stock $S(\cdot)$ at time $0$, and is the analogue of the market portfolio in this simple one-dimensional market.
\end{itemize}
Now it is easy to see that $\pi(\cdot)$ is an arbitrage with respect to $\rho(\cdot)$, namely
$$V^\pi(T)=1>0=V^\rho(T)\ \text{ a.s.}$$
Hence this is an example of a market model that allows an ELMM (namely, the measure under which \eqref{exch} holds), but a relative arbitrage with respect to the market still exists --- that is, Proposition \ref{prop63} does not apply in this case. \qed
\end{exmp}

It follows that in a market where the bounded variance assumption \eqref{BV} holds, the existence of a market-relative arbitrage is equivalent to all LMDs being strict local martingales. If we furthermore assume that the filtration is generated by the driving $d$-dimensional Brownian motion $W(\cdot)$, i.e.~$\mathbb{F}=\mathbb{F}^W$, then the above and Proposition 6.2 of \cite{fk09} show that the existence of a relative arbitrage is equivalent to the non-existence of an ELMM. This, in turn, is equivalent to the existence of a free lunch with vanishing risk (FLVR), which, since we are assuming NUPBR, is equivalent to the existence of an arbitrage. This leads to the following corollary:
\begin{cor} \label{RArealA}
Assume \eqref{BV} and $\mathbb{F}=\mathbb{F}^W$. Then there exists a relative arbitrage with respect to the market if and only if there exists an (non-scalable) arbitrage.
\end{cor}




\section{Diverse models} \label{secdiverse} 
The first class of market models for which it was shown that relative arbitrages exist, both over sufficiently long as well as arbitrarily short time horizons, is the class of diverse models. Diversity corresponds to the observation that no single company is allowed to dominate the entire market in terms of relative capitalisation, for instance due to anti-trust regulations. The following definition (i.e.~Definition 2.2.1 of \cite{f02}) formalises this observation mathematically:

\begin{defn} \label{defdiv}
We call a market model \emph{diverse} on $[0,T]$ if\footnote{Recall Definition \ref{defrevorder}; $\mu_{(1)}(\cdot)$ is the maximum process of the collection $\mu_i(\cdot),\  i=1,\ldots,n$.}
\begin{equation} \label{eqdefdiv}
\exists\, \delta\in(0,1)\ \text{ such that }\ \mu_{(1)}(t)<1-\delta \quad \forall\, t\in [0,T]\quad \mathbb{P}\text{-a.s.}
\end{equation}
A model is called \emph{weakly diverse} on $[0,T]$ if
\begin{equation} \label{eqdefwkdiv}
\exists\, \delta\in(0,1)\ \text{ such that }\ \frac{1}{T} \int_0^T \mu_{(1)}(t)\mathrm{d}t < 1-\delta \qquad \mathbb{P}\text{-a.s.}
\end{equation}
\end{defn}

A natural question to ask is whether there exists an It\^o model \eqref{model} that fits our framework at all, or whether Definition \ref{defdiv} of diversity is vacuous. For instance, Remark 5.1 in \cite{fk09} asserts that diversity fails in a market with constant growth rates and where \eqref{BV} and \eqref{ND} hold. It was shown in \cite{fkk05} that there do exist market models which are diverse; namely, let $\delta\in(1/2,1)$, $d=n$, and let $\sigma(\cdot)\equiv\sigma$ be a constant matrix satisfying \eqref{ND}. Let $g_1,\ldots,g_n\geq0$; then, for $t\in[0,T]$, set
\begin{align}
\mathrm{d}\log X_i(t)=\gamma_i(t)\mathrm{d}t +\sum_{\nu=1}^d \sigma_{i\nu}\mathrm{d}W_\nu(t)\quad i=1,\ldots,n,
\end{align}
where, for some constant $M>0$,
\begin{align}
\gamma_i(t):=g_i \mathbbm{1}_{\{ X_i(t)\neq X_{(1)}(t) \}} - \frac{M}{\delta}\frac{\mathbbm{1}_{\{ X_i(t)= X_{(1)}(t) \}}}{\log\big( (1-\delta)X(t)/X_i(t) \big)}.
\end{align}
The authors of \cite{fkk05} show that this system of SDEs has a unique strong solution, and that the diversity property \eqref{eqdefdiv} is satisfied by this model. They go on to construct a model which is weakly diverse, but not diverse. 
The authors of \cite{or06} describe a more general way of constructing diverse market models, using a change of measure technique. We discuss this method in depth in Section \ref{secmeasure}. Other ways to study diverse markets, but which do not fit into our framework (i.e.~are not of the form \eqref{model}, the reason being that companies are allowed to merge or split), are proposed in \cite{sf11}, \cite{sar14} and \cite{ks14}.

\subsection{Relative arbitrage over long time horizons} \label{longRAdiv}
Although the diversity of markets had been studied before, see e.g.\! \cite{fgh98} and \cite{f99}, Fernholz was the first to show in Corollary 2.3.5 of \cite{f02} that relative arbitrages exist (over sufficiently long time horizons) in diverse markets which satisfy an additional non-degeneracy condition on the volatility structure, using what he defined as \emph{entropy-weighted} portfolios (see \eqref{defewp}). This non-degeneracy condition is similar to the \eqref{BV} condition:
\begin{equation} \label{ND} \tag{ND}
\exists\,\varepsilon>0\ \text{ such that }\ \xi'a(t)\xi\geq \varepsilon ||\xi||^2,\quad \forall\xi\in\mathbb{R}^n,t\geq0 \qquad\mathbb{P}\textnormal{-a.s.}
\end{equation}
We quote the following result from Proposition 2.2.2 in \cite{f02}:
\begin{prop} \label{prop222}
If a model is diverse and \eqref{ND} holds, then 
\begin{align} \label{eqdefvol}
\exists\, \zeta>0\ \textnormal{ such that }\ \gamma^*_\mu(t) \geq \zeta \quad \forall\, t\in [0,T]\quad \mathbb{P}\textnormal{-a.s.}
\end{align}
Conversely, if both \eqref{BV} and \eqref{eqdefvol} hold, then diversity follows.
\end{prop}
Equation \eqref{eqdefvol} defines the \emph{sufficient intrinsic volatility} property, and is the topic of Section \ref{secintvol}. There, we demonstrate the construction of a relative arbitrage over sufficiently long time horizons in such a model, using entropy-weighted portfolios --- see computation \eqref{eqVewp}. Alternatively, see Theorem 2.3.4 and Corrolary 2.3.5 of \cite{f02} for a proof that these portfolios outperform the market portfolio in diverse markets.

In \cite{fkk05} the authors showed, in weakly diverse markets, the existence of another relative arbitrage with respect to the market portfolio, namely the \emph{diversity-weighted} portfolio --- see \eqref{defDWP}.\footnote{Several other definitions have been coined by \cite{fkk05} for weaker types of diversity, such as asymptotic diversity, but these have not been studied in the literature.} However, for this they needed to assume the \eqref{ND} assumption as well, which, unlike the assumption of diversity, does not come from observation, thus diminishing the robustness of the result. We now demonstrate how a relative arbitrage was constructed in (the Appendix of) \cite{fkk05}, i.e.~in a market that is non-degenerate in the sense of \eqref{ND} and weakly diverse over $[0,T]$ for $T\geq 2 \log n/p\varepsilon \delta$, using a `diversity-weighted portfolio'. This construction leans heavily on the following lemma (also proved in the Appendix of \cite{fkk05}):

\begin{lem} \label{lem3}
If condition \eqref{ND} holds, then for any long-only portfolio $\pi(\cdot)$ we have 
\begin{equation} \label{ineq2}
\frac{\varepsilon}{2}(1-\pi_{(1)}(t))\leq \gamma^*_\pi(t)\quad \textnormal{a.s.}
\end{equation}
in the notation of Definition \ref{defrevorder}.
\end{lem}
\begin{proof}
By definition of $\tau^\pi_{ij}(t)$, and by condition \eqref{ND}, we have the inequality
\begin{align}
\tau^\pi_{ii}(t)&=\left( \pi(t)-e_i  \right)' a(t)\left( \pi(t)-e_i  \right) \geq \varepsilon || \pi(t)-e_i ||^2\nonumber\\&=\varepsilon \Big( (1-\pi_i(t))^2 + \sum_{j\neq i} \pi^2_j(t) \Big).
\end{align}
Plugging this into equation \eqref{gammaWA}, we conclude that
\begin{align}
\gamma^*_\pi(t)&\geq \frac{\varepsilon}{2}\sum_{i=1}^n \pi_i(t)\Big(  (1-\pi_i(t))^2 + \sum_{j\neq i} \pi^2_j(t) \Big)\nonumber\\
&=\frac{\varepsilon}{2}\Big(\sum_{i=1}^n \pi_i(t)  (1-\pi_i(t))^2 + \sum_{j=1}^n \pi^2_j(t)(1-\pi_j(t) \Big)\nonumber\\
&=\frac{\varepsilon}{2}\sum_{i=1}^n \pi_i(t)  (1-\pi_i(t))\geq \frac{\varepsilon}{2} (1-\pi_{(1)}(t)).
\end{align}
This proves the result.
\end{proof}

\begin{defn}
Define the \emph{diversity-weighted} portfolio $\mu^{(p)}(\cdot)$ with parameter $p\in(0,1)$ by
\begin{equation} \label{defDWP}
\mu^{(p)}_i(t):=\frac{(\mu_i(t))^p}{\sum_{j=1}^n (\mu_j(t))^p}\quad i=1,\ldots,n.
\end{equation}
\end{defn}
One can check that this portfolio is generated, in the sense of Section \ref{fgportfolios}, by the function
\begin{equation} \label{genfctdiv}
\mathbf{G}_p:x\mapsto \Big(\sum_{i=1}^n x_i^p\Big)^{1/p}.
\end{equation}
We compute, for $\mu\in\mathbb{R}^n$ and $i,j=1,\ldots,n$,
\begin{align} \label{dijg}
D^2_{ij}\mathbf{G}_p (\mu)=\begin{cases} (1-p) (\mathbf{G}_p (\mu))^{1-2p}\mu^{p-2}_i\Big(\mu^p_i-(\mathbf{G}_p (\mu))^p\Big)\ &(i=j)\\
(1-p) (\mathbf{G}_p (\mu))^{1-2p}(\mu_i\mu_j)^{p-1}  &(i\neq j)
\end{cases}
\end{align}
and the bounds
\begin{equation} \label{gpbound}
1=\sum_{i=1}^n \mu_i(t) \leq \sum_{i=1}^n \mu^p_i(t)\leq \sum_{i=1}^n \Big(\frac{1}{n}\Big)^p = n^{1-p}.
\end{equation}
Using Lemma \ref{master}, equation \eqref{dijg} implies that the drift process equals (we omit the time-dependence of the processes $\mu(\cdot)$ and $\tau^\mu(\cdot)$ to ease notation)
\begin{align}
\mathfrak{g}(t)&=-\frac{1-p}{2\mathbf{G}_p (\mu)}\left[ -\sum_{i=1}^n (\mathbf{G}_p (\mu))^{1-p}\mu^p_i \tau^\mu_{ii} +   \sum_{i,j=1}^n (\mathbf{G}_p (\mu))^{1-2p}\mu^p_i\mu^p_j \tau^\mu_{ij}  \right]\nonumber\\
&= \frac{1-p}{2}\left[ \sum_{i=1}^n \mu^{(p)}_i \tau^\mu_{ii} - \sum_{i,j=1}^n \mu^{(p)}_i\mu^{(p)}_j \tau^\mu_{ij}  \right]\nonumber\\
&= (1-p)\gamma^*_{\mu^{(p)}}(t)
\end{align}
and therefore that
\begin{equation} \label{divlongra}
\log\left(\frac{V^{\mu^p}(T)}{V^\mu(T)} \right) =\log \left(\frac{\mathbf{G}_p(\mu(T))}{\mathbf{G}_p(\mu(0))} \right) + (1-p)\int_0^T \gamma^*_{\mu^{(p)}}(t) \mathrm{d}t\ \text{ a.s.}
\end{equation}
Now using the bounds \eqref{gpbound}, we get the lower bound
\begin{equation} \label{boundgp1}
\log \left(\frac{\mathbf{G}_p(\mu(T))}{\mathbf{G}_p(\mu(0))} \right) \geq -\frac{1-p}{p}\log n,
\end{equation}
which implies that $V^{\mu^p}(T)/V^\mu(T)\geq n^{-(1-p)/p},\ \mathbb{P}$-a.s., since $\gamma^*_{\mu^{(p)}}(\cdot)$ is a non-negative process for the long-only portfolio $\mu^{(p)}(\cdot)$ by Lemma \ref{lem2}. We use \eqref{ND} and Lemma \ref{lem3}, together with the observation that $\mu^{(p)}_{(1)}(t)\leq\mu_{(1)}(t)$, to get
\begin{equation} \label{lowbd}
\int_0^T \gamma^*_{\mu^{(p)}}(t) \mathrm{d}t \geq \frac{\varepsilon}{2}\int_0^T (1-\mu^{(p)}_{(1)}(t)) \mathrm{d}t\geq\frac{\varepsilon}{2}\int_0^T (1-\mu_{(1)}(t)) \mathrm{d}t > \frac{1}{2}\varepsilon\delta T.
\end{equation}
From equation \eqref{lowbd}, the bound \eqref{boundgp1} and equation \eqref{divlongra}, we conclude that 
\begin{equation} \label{finalvv}
\log\left(\frac{V^{\mu^p}(T)}{V^\mu(T)} \right) > (1-p)\left( \frac{\varepsilon\delta T}{2}-\frac{\log n}{p} \right)\ \text{ a.s.}
\end{equation}
Therefore, if we have $$T\geq 2 \log n/p\varepsilon \delta,$$ (i.e., if T is big enough) we get from equation \eqref{finalvv} that $$\mathbb{P}(V^{\mu^{(p)}}(T)>V^\mu(T))=1.$$ Therefore, the diversity-weighted portfolio is a relative arbitrage with respect to the market over long enough time horizons, under the conditions of weak diversity and non-degeneracy. Note that this is a portfolio and therefore invests only in the stocks, since $\sum_i \mu^{(p)}_i(\cdot)=1$.

\subsection{Relative arbitrage over short time horizons}
The problem of constructing a relative arbitrage over arbitrarily short time horizons was first raised in \cite{f02}, and solved for the case of non-degenerate weakly diverse markets in \cite{fkk05}. The main idea behind the construction is to take a short position in a `mirror image' of the portfolio $e_1$, with respect to which the market portfolio can be shown to be a relative arbitrage, and to take a long position in the market.

\begin{defn}
For any $q\in\mathbb{R}$, define the \emph{$q$-mirror image} of $\pi$ with respect to the market portfolio as
\begin{equation}
\tilde{\pi}^{[q]}(t):=q \pi(t) + (1-q)\mu(t).
\end{equation}
\end{defn}

In analogy with the defining equation \eqref{deftau}, let us define the \emph{relative covariance} of a portfolio $\pi(\cdot)$ with respect to the market as
\begin{equation} \label{deftaumu}
\tau^\pi_{\mu\mu}(t):=(\pi(t)-\mu(t))'a(t)(\pi(t)-\mu(t)).
\end{equation}

The following lemma, which we quote without proof (see Lemma 8.1 in \cite{fkk05}), will be essential:
\begin{lem} \label{lemshortRA}
If there exist $T>0$, $\eta>0$, and $\beta\in(0,1)$ such that
\begin{equation}
\int_0^T \tau^\mu_{\pi\pi}(t)\mathrm{d}t\geq \eta\ \textnormal{ a.s. }\quad and\ \quad V^\pi(T)/V^\mu(T)\leq 1/\beta\ \textnormal{ a.s.},
\end{equation} 
then
\begin{equation}
V^{\tilde{\pi}^{[q]}}(T)<V^\mu(T)\ \textnormal{ a.s.}
\end{equation} 
 for $q> 1+(2/\eta)\log(1/\beta)$.
\end{lem}

\cite{fkk05} then proceed by showing that (see their equation (8.7))
\begin{equation}
\log \bigg( \frac{V^{\tilde{\pi}^{[q]}}(T)}{V^\mu(T)} \bigg)=q\log\left( \frac{V^{\pi}(T)}{V^\mu(T)} \right)+\frac{q(1-q)}{2}\int_0^T \tau^\mu_{\pi\pi}(t)\mathrm{d}t.
\end{equation}
We create a ``seed'' portfolio $\tilde{\pi}^{[q]}(\cdot)$ which is the $q$-mirror image of $e_1$, the first unit vector in $\mathbb{R}^n$. The assumptions of weak diversity and nondegeneracy allow us to use Lemma \ref{lemshortRA}, which with $\beta=\mu_1(0)$ and $\eta=\varepsilon\delta^2 T$ implies that the market portfolio $\mu(\cdot)$ is a relative arbitrage with respect to the seed, provided that $q>q(T):=1+(2/\varepsilon\delta^2 T)\log(1/\mu_1(0))$.
Finally, and as in Example 8.3 of \cite{fkk05}, a relative arbitrage over arbitrary $[0,T]$ is created by going long $\$q/(\mu_1(0))^q$ in $\mu(\cdot)$, and shorting $\$1$ in the seed portfolio. This corresponds to the long-only portfolio defined as $$\xi_i(t):= \frac{1}{V^\xi(t)}\left(\frac{q \mu_i(t)}{(\mu_1(0))^q}V^\mu(t)-\tilde{\pi}^{[q]}_i(t)V^{\tilde{\pi}^{[q]}} (t) \right), \quad i=1,\ldots,n.$$
Now $\xi(\cdot)$ outperforms at $t = T$ the market portfolio with the same initial capital of $z := Z\xi(0) = q/(\mu_1(0))^q - 1 > 0$ dollars, because $\xi(\cdot)$ is long in the market $\mu(\cdot)$ and short in the seed portfolio $\tilde{\pi}^{[q]}(\cdot)$ which underperforms the market at $t = T$;
\begin{equation}
V^{z,\xi}(T)=\frac{q}{(\mu_1(0))^q}V^\mu(T)-V^{\tilde{\pi}^{[q]}}(T)>zV^\mu(T)=V^{z,\mu}(T)\qquad \mathbb{P}\text{-a.s.}
\end{equation}
By choosing $q$ large enough, this can be made to hold over any $[0,T]$. However, note that the minimal required initial wealth tends to infinity as the time horizon becomes shorter: $$z(T) := z\xi(0) = q(T)/(\mu_1(0))^{q(T)} - 1\rightarrow \infty \ \text{ as }\ T\downarrow 0.$$


\section{Sufficiently volatile models} \label{secintvol}
Relative arbitrage over sufficiently long time horizons has also been shown to exist (without any additional assumptions on the volatility structure) in so-called sufficiently volatile markets, as defined in \eqref{eqdefvol} from the previous section. This was first done in \cite{fk05}, Proposition 3.1.
\begin{defn} \label{defvol}
A market satisfies the \emph{sufficient intrinsic volatility} property on $[0,T]$, or is called \emph{sufficiently volatile}, if 
\begin{align} \label{eqdefvol2}
\exists\, \zeta>0\ \text{ such that }\ \gamma^*_\mu(t) \geq \zeta \quad \forall\, t\in [0,T]\quad \mathbb{P}\text{-a.s.}
\end{align}
Furthermore, we say that a model is \emph{weakly sufficiently volatile} if there exists a continuous, strictly increasing function $\Gamma:[0,\infty)\to[0,\infty)$ with $\Gamma(0)=0$ and $\Gamma(\infty)=\infty$, such that
\begin{align} \label{eqdefvolavg}
 \infty > \int_0^t \gamma^*_\mu(s)\mathrm{d}s \geq \Gamma(t) \quad \forall\, t\in [0,T]\quad \mathbb{P}\text{-a.s.}
\end{align}
\end{defn}
Recall equation \eqref{intvol} and the interpretation of $\gamma^*_\mu(\cdot)$ as a measure of the market's `intrinsic volatility' --- this motivates the nomenclature of `sufficient intrinsic volatility' in regard to \eqref{eqdefvol2}. In Figure 1 of \cite{fk05}, the authors argue that the property \eqref{eqdefvolavg} holds in real markets by plotting the function $\int_0^\cdot \gamma^*_\mu(s)\mathrm{d}s$ over a long time period, and visually showing that it lies above a straight line with positive gradient. However, this property might depend on the moment in time at which one starts looking at this function, and further analysis using real-world data would be required to make a stronger case for the sufficient intrinsic volatility of real stock markets.


As will become clear in Section \ref{introVSM}, models of the form \eqref{model} that are sufficiently volatile exist.

\subsection{Relative arbitrage over long time horizons} \label{secintvolRA}
In Proposition 3.1 of \cite{fk05} it was first shown that \emph{entropy-weighted} portfolios, as defined below, are relative arbitrages with respect to the market over sufficiently long time horizons. In this, the authors do not need to assume \eqref{BV} nor \eqref{ND}, but merely \eqref{eqdefvolavg}. We display their construction of these RA opportunities below.

\begin{defn}
Define the \emph{entropy-weighted} portfolio $\pi^c(\cdot)$ with parameter $c>0$ to be the portfolio generated by a version of the Shannon entropy function
\begin{equation}
\textbf{H}_c(x):=c+\textbf{H}(x):=c-\sum_{i=1}^n x_i \log x_i.
\end{equation}
Here, $\textbf{H}$ is the standard Shannon entropy function. One can check that 
\begin{equation} \label{defewp}
\pi^c_i(t)= \frac{ \mu_i(t)(c-\log \mu_i(t))}{\sum_{j=1}^n \mu_j(t)(c-\log \mu_j(t))},\quad i=1,\ldots,n.
\end{equation}
\end{defn}

Once again, we compute for general $\mu\in\mathbb{R}^n$
\begin{equation}
D^2_{ij}\textbf{H}_c (\mu)=-\frac{1}{\mu_i}\delta_{ij}\quad i,j=1,\ldots,n,
\end{equation}
with $\delta_{ij}$ the Kronecker-delta, which with Lemma \ref{master} implies for the drift process
\begin{equation} \label{gVol} 
\mathfrak{g}(t)=\frac{1}{2\textbf{H}_c (\mu(t))} \sum_{i=1}^n \mu_i(t)\tau^{\mu}_{ii}(t)=\frac{\gamma^*_\mu(t)}{\textbf{H}_c (\mu(t))},
\end{equation}
where we have used equation \eqref{gamma+}. The last thing we need for the construction of a relative arbitrage is the bound
\begin{equation}
c < \textbf{H}_c(x)\leq c+\log n;
\end{equation}
using this together with Lemma \ref{master} and the computation \eqref{gVol}, we get that
\begin{align}  \label{eqVewp}
\log\left(\frac{V^{\mu^p}(T)}{V^\mu(T)} \right) &=\log \left(\frac{\textbf{H}_c (\mu(T))}{\textbf{H}_c (\mu(0))} \right) + \int_0^T \frac{\gamma^*_\mu(t)}{\textbf{H}_c (\mu(t))} \mathrm{d}t\nonumber\\
&> -\log \left(\frac{\textbf{H}_c(\mu(0))}{c}\right) + \frac{\zeta T}{c+\log n}\ \text{ a.s.}
\end{align}
We conclude that, if
\begin{equation}  \label{largeT*}
T>\mathcal{T}_*(c):=\frac{1}{\zeta} (c+\log n) \log\left( \frac{c + \textbf{H}(\mu(0))}{c} \right),
\end{equation}
or, alternatively,
\begin{equation*} 
T>\mathcal{T}_*:=\frac{1}{\zeta}\textbf{H}(\mu(0))=\lim_{c\to\infty} \mathcal{T}_*(c),\ \text{ and } c>0\text{ is chosen sufficiently large}, 
\end{equation*}
then by \eqref{eqVewp} the entropy-weighted portfolio $\pi^c(\cdot)$ is a relative arbitrage with respect to the market portfolio over the time horizon $[0,T]$.

As was mentioned in Proposition \ref{prop222}, a market that is diverse and satisfies \eqref{ND} is also sufficiently volatile. Hence it follows from the above that in such markets, the entropy-weighted portfolio beats the market after a sufficiently long time --- see Corollary 2.3.5 of \cite{f02} for a direct proof.

\subsection{Relative arbitrage over short time horizons} \label{MOQ}
It is a \emph{major} open problem whether the sufficient intrinsic volatility property \eqref{eqdefvol2} is a sufficient condition for the existence of relative arbitrage over arbitrarily short time horizons. This question was posed in Remark 11.3 in \cite{fk09}, and it remains unclear what the answer to it is. It has been shown that relative arbitrages over short time horizons exist in several subclasses of the sufficiently volatile model class, one of them being those models with $\gamma^*_\mu(t)\geq\zeta>0$ a.s.\! which are Markovian and non-degenerate in a sense slightly different from \eqref{ND}, namely: for every compact $\mathcal{K}\subset(0,\infty)^n$,
\begin{equation} \label{weird}
\exists\, \varepsilon_\mathcal{K}>0\ \text{ such that }\ \sum_{i,j=1}^n x_i x_j a_{ij}(\mathbf{x}) \xi_i \xi_j \geq \varepsilon ||\xi||^2,\quad \forall \mathbf{x}\in\mathcal{K},\xi\in\mathbb{R}^n;
\end{equation}
see Proposition 2 and the subsequent Corollary in \cite[pp. 1194--1195]{fernk10} for this result. The other two classes for which this major open problem has been answered positively are the so-called volatility-stabilised markets and generalised volatility-stabilised markets, which are the topics of Sections \ref{introVSM} and \ref{introVSMgen}. 

A closely related open question, which was posed in Section 4 of \cite{bf08} as well as Remark 11.4 of \cite{fk09}, is whether short-term relative arbitrage exists for a market with the property that
\begin{equation} \label{eqgammap}
\Gamma(t)\leq\int_0^t \gamma^*_{\mu,p}(s) \mathrm{d}s<\infty\quad \forall t\in[0,T]\qquad \text{a.s.},
\end{equation}
for some $p\in(0,1)$ and continuous, strictly increasing function $\Gamma:[0,\infty)\to[0,\infty)$ with $\Gamma(0)=0$ and $\Gamma(\infty)=\infty$. Here, $\gamma^*_{\mu,p}(\cdot)$ is the \emph{generalised excess growth rate} of the market portfolio, defined as
\begin{equation}
\gamma^*_{\mu,p}(t):=\frac{1}{2}\sum_{i=1}^n(\mu_i(t))^p\tau^\mu_{ii}(t);
\end{equation}
compare to \eqref{intvol}.
\cite{fk05} shows that relative arbitrages exist over sufficiently long time horizons in such markets, but the case for short time horizons remains unanswered for $n\geq3$.\footnote{In the case $n=2$, property \eqref{eqgammap} implies condition \eqref{cond2.7}, so the proof of \cite{bf08} applies and short-term RA exists --- see Section \ref{RAshortintvol}.} Namely, Proposition 3.8 of \cite{fk05} asserts the following:
\begin{prop} \label{prop3.8}
Suppose that for some numbers $p \in (0,1)$, $T \in (0,\infty)$ and $\zeta \in (0,\infty)$ we have the condition 
\begin{equation} \label{fk05cond}
\frac{n^{1-p}}{p}\log n + \zeta \leq\int_0^T \gamma^*_{\mu,p}(t) \mathrm{d}t<\infty\qquad \text{a.s.}.
\end{equation}
Then the $p$-mirror image of the diversity-weighted portfolio with parameter $p$, 
\begin{equation} \label{fk05pi}
\pi_i(t):= p\frac{(\mu_i(t))^p}{\sum_{j=1}^n (\mu_j(t))^p}+(1-p)\mu_i(t),
\end{equation}
is an arbitrage relative to the market portfolio over $[0,T]$.
\end{prop}
Note that Proposition \ref{prop3.8} implies that in a market satisfying \eqref{eqgammap}, we have $\mathbb{P}(V^\pi(T)>V^\mu(T))=1$ when $T>\Gamma^{-1}((1/p)n^{1-p}\log n)$; i.e., $\pi(\cdot)$ of \eqref{fk05pi} beats the market over sufficiently long time horizons. One way to see this is by checking that $\pi(\cdot)$ of \eqref{fk05pi} is generated by $G:x\mapsto\sum_{i=1}^n x^p_i$, which satisfies $1<G(\cdot)\leq n^{1-p}$, that $\mathfrak{g}(\cdot)=p(1-p)\gamma^*_{\mu,p}(\cdot)/G(\mu(\cdot))$, and concluding that
\begin{align*}
\log \left(\frac{V^\pi(T)}{V^\mu(T)}\right) &= \log\left( \frac{G(\mu(T))}{G(\mu(0))} \right)+ p(1-p)\int_0^T\frac{ \gamma^*_{\mu,p}(t)}{G(\mu(t))}\mathrm{d}t \\
&\geq -(1-p)\log n + \frac{p(1-p)}{n^{1-p}}\Gamma(T) >0 \quad \text{a.s.},\\
&\text{provided } T>\Gamma^{-1}((1/p)n^{1-p}\log n).
\end{align*}

\subsection{Volatility-stabilised model} \label{introVSM}
One special case of an explicit market model for which the excess growth rate of the market portfolio is bounded away from zero is the volatility-stabilised model. This model was introduced in \cite{fk05}, and it has been shown in \cite{bf08} that relative arbitrages exist over arbitrarily short time horizons in this model, answering an open question in \cite{fk05} (see the bottom of page 164 of that paper).

Volatility-stabilised models translate the observation that smaller stocks (i.e., stocks of companies with small relative market capitalisations) tend to give higher returns and be more volatile than large-capitalisation stocks. It must be noted, however, that they are an approximation and oversimplification of real markets, unsuitable for capturing all properties of markets (such as stock correlation, to name one).
\begin{defn} \label{defVSM}
Define a \emph{volatility-stabilised model} (VSM) with parameter $\alpha\geq0$ to be a model in which the log-stock price processes follow
\begin{align} \label{eqdefVSM}
&\mathrm{d}\log X_i(t)=\frac{\alpha}{2\mu_i(t)}\mathrm{d}t + \frac{1}{\sqrt{\mu_i(t)}}\mathrm{d}W_i(t),\\
&X_i(0)>0 \qquad i=1,\ldots,n. \nonumber
\end{align} 
\end{defn}
As one can easily see, volatilities and drifts are largest for small stocks in such markets. \cite{fk05} show that diversity fails in such markets (see their Remark 4.6 and the preceding computation), yet there exists what the authors call `stabilisation by volatility' (and by drift, if $\alpha>0$): straightforward computations show that
\begin{align}
a_{\mu\mu}(\cdot)\equiv1,\quad \gamma^*_{\mu}(\cdot)\equiv\gamma^*:=\frac{n-1}{2}>0,\quad \gamma_{\mu}(\cdot)\equiv\gamma:=\frac{(1+\alpha)n-1}{2}>0.
\end{align}
Importantly, VSMs have a constant excess growth rate of the market portfolio, and therefore satisfy the sufficient intrinsic volatility condition \eqref{eqdefvol}, showing that there do indeed exist sufficiently volatile market models \eqref{model}. Using It\^o's formula, one can easily show that the total market capitalisation $X(t)$ is a geometric Brownian motion in this market, namely $X(t) = X(0)e^{\gamma t + \mathcal{W}(t)}$ for $\mathcal{W}(\cdot)=\sum_{i=1}^n \int_0^\cdot \sqrt{\mu_i(s)} \mathrm{d}W_i(s)$ a standard $\mathbb{P}$-BM. The overall market and largest stock have the same growth rate $\gamma$, and if $\alpha>0$ all stocks have the same growth rate. 

The properties of VSMs have been studied in depth. Namely, in Section 12.1 of \cite{fk09} the authors study the asymptotic behaviour of the model \eqref{eqdefVSM} using Bessel processes, and show that if $\alpha=0$ then the (strict) local martingale deflator can be expressed as 
\begin{equation} \label{ZVSM}
Z(t)=\frac{\sqrt{X_1(0)\cdots X_n(0)}}{\mathfrak{R}_1(u)\cdots \mathfrak{R}_n(u)} \exp\left\{ \frac{1}{2}\int_0^u \sum_{i=1}^n \mathfrak{R}^{-2}_i(s)\mathrm{d}s \right\} \Bigg|_{u=\Lambda(t)},
\end{equation}
where $\mathfrak{R}_i(\cdot)$ are independent Bessel processes in dimension $2(1+\alpha)$, and $$\Lambda(t):=\frac{1}{4}\int_0^t X(s)\mathrm{d}s$$ is a time change.
The joint density of market weights in VSMs has been computed in \cite{pal11}, answering an open question (Remark 13.4) of \cite{fk09}. Pal shows that the law of the market weights is identical to that of the multi-allele Wright-Fisher diffusion model from population genetics.

Since the VSM is a special case of a sufficiently volatile model, it follows from Section \ref{secintvolRA} that entropy-weighted portfolios are long-term relative arbitrages with respect to the market. Furthermore, one can show that the diversity-weighted portfolio with parameter $p=1/2$ is an arbitrage relative to the market for time horizons $T>\frac{8\log n}{n-1}$ --- see Example 12.1 of \cite{fk09}. And finally, the $\lambda$-mirror image of the equally-weighted portfolio $$\hat{\pi}_i(t):=\lambda\frac{1}{n}+(1-\lambda)\mu_i(t), \qquad \lambda=\frac{n(1+\alpha)}{2},$$ has the num\'{e}raire property: $V^\pi/V^{\hat{\pi}}(\cdot)$ is a supermartingale for all $\pi(\cdot)$ (see Section \ref{numEU}).

\subsubsection*{Relative arbitrage over short time horizons} \label{RAshortintvol}
The question whether there exist short-term relative arbitrage opportunities in VSMs was first raised in \cite{fk05}, and solved in \cite{bf08} where a relative arbitrage was constructed explicitly. The VSM is Markovian and satisfies \eqref{weird}, so one knows \emph{a priori} that relative arbitrage over short horizons exist by Proposition 2 of \cite[pp. 1194--1195]{fernk10}. 

The way in which Banner and Fernholz construct a short-term relative arbitrage in \cite{bf08} is by generating a portfolio using the standard incomplete Gamma function, and following this portfolio up to a certain stopping time, after which the market portfolio is implemented. Explicitly, they generate the portfolio $\pi(\cdot)$ by the function $G(x_1,\ldots,x_n):=\sum_{i=1}^n f(x_i)$, where $f(y)$ is defined for $y\in[0,1]$ as
\begin{equation}
f(y):= \begin{cases}
\Gamma(c+1,-\log y) := \int_{-\log y}^\infty e^{-r}r^c \mathrm{d}r & \text{ if } 0<y\leq1,\\
0 & \text{ if } y=0.
\end{cases}
\end{equation}
Here the constant is chosen as $c=\frac{8n(n-1)}{T}\int_0^{1/n} -\log r\, \mathrm{d}r $. In equation (3.6) of \cite{bf08}, the authors show the lower bound
\begin{align*}
\log\left(\frac{V^\pi(t)}{V^\mu(t)} \right)\geq&\ \overbrace{\log((n-1)f(\mu_{(n)}(t)) +f(1-(n-1)\mu_{(n)}(t))}^{:=S_1(\mu_{(n)}(t))}-\log(n f(1/n)) \nonumber\\
&+\int_0^t \underbrace{\frac{-\mu_{(n)}(s)f''(\mu_{(n)}(s))}{4\left( f(\mu_{(n)}(s))+(n-1)f\big( \frac{1-\mu_{(n)}(s)}{n-1} \big)  \right)}}_{:= \Theta_1(\mu_{(n)}(s))}\mathrm{d}s \qquad \text{ a.s.}
\end{align*}
Then, for $Y(\cdot)$ the inverse of the function $T_1(Y):= T/2 + \int_{1/n}^Y-\frac{S'_1(r)}{\Theta_1(r)}\mathrm{d}r$, the following stopping rule is defined:
$$T_0:= \inf\{ t\geq T/2 : \mu_{(n)}(t)>Y(t) \}.$$
Finally, a relative arbitrage $\tilde{\pi}(\cdot)$ with respect to the market is defined by setting 
\begin{equation}
\tilde{\pi}(t):= \begin{cases}
\pi(t) & \text{ if } t<T_0,\\
\mu(t) & \text{ if } t\geq T_0.
\end{cases}
\end{equation}
It is shown that the lower bound on the amount of arbitrage guaranteed by $\tilde{\pi}(\cdot)$ tends to zero very quickly as the time horizon becomes shorter. Furthermore, this construction works for any market satisfying the condition \begin{equation} \label{cond2.7}
\tau_{m(t)m(t)}(t)\geq \frac{C}{\mu_{m(t)}(t)}\quad \forall t\geq0\qquad \text{a.s.}
\end{equation}
for some constant $C>0$ and $m(t)$ the index of the stock with smallest capitalisation, i.e.~$\mu_{m(\cdot)}(\cdot)=\mu_{(n)}(\cdot)$. Condition \eqref{cond2.7} holds in VSMs with $C=1/2$, as well as in more general versions of VSMs with $\alpha$ replaced by any drift process $\gamma(\cdot)$ so that the $n$-dimensional SDE \eqref{eqdefVSM} still has a solution.

\subsection{Generalised volatility-stabilised model} \label{introVSMgen}
A generalisation of volatility-stabilised models was introduced in \cite{pic13}, and in the same article the author showed that under certain conditions one can construct relative arbitrages over arbitrarily short time horizons in these generalised models.
\begin{defn} \label{defgenVSM}
Define a \emph{generalised volatility-stabilised model} to be a model of the form
\begin{align} \label{eqdefgenVSM}
&\mathrm{d}\log X_i(t)=\frac{\alpha_i}{2(\mu_i(t))^{2\beta}}\big[K(X(t))\big]^2\mathrm{d}t + \frac{\sigma}{(\mu_i(t))^\beta}K(X(t))\mathrm{d}W_i(t),\\
&X_i(0)>0 \qquad i=1,\ldots,n. \nonumber
\end{align} 
Here, $\alpha_i\geq0$, $\sigma>0$ and $\beta>0$ are given constants, and the given function $K(\cdot):(0,\infty)^n\to(0,\infty)$ is measurable and such that \eqref{eqdefgenVSM} has a weak solution which is unique in distribution.
\end{defn}
Note that the case of $K(\cdot)\equiv1$, $\sigma=1$, $\alpha_i=\alpha$ for $i=1,\ldots,n$, and $\beta=1/2$ corresponds to the VSM \eqref{eqdefVSM}.

Pickov\'a shows in \cite{pic13} that, if $K(\cdot)$ is bounded away from zero, the diversity-weighted portfolio $\mu^{(p)}(\cdot)$ outperforms the market over $[0,T]$ for any $p\leq 2\beta$ and for $T$ sufficiently large. If, in addition, $\beta\geq1/2$, then the same approach as in Proposition 2 in Section 5 of \cite{bf08} can be used to construct a relative arbitrage with respect to the market over any horizon $[0,T]$.

Generalised VSMs have not yet been studied or mentioned in the literature outside of \cite{pic13}, but could offer a general way of modelling the stock market that preserves the sufficient intrinsic volatility property, as well as incorporating the observation that smaller capitalisation stocks tend to have higher volatilities and drifts.

\section{Rank-based models and portfolios}  \label{secrank} 
It has been observed that the log-log capitalisation distribution curve, i.e., the mapping $\log k \mapsto \log \mu_{(k)}(t)$, exhibits great stability over time --- see Figure 5.1 in \cite{f02}. The fact that capital seems to be distributed in a time-independent way according to capitalisation rank (despite the occurrence of extreme events such as crashes) has motivated the study of rank-based models, which were first introduced in \cite{f02}, and in which the drift and volatility coefficients of each stock depend explicitly on its \emph{rank} within the market's capitalisation. These models can be constructed so as to have the stability property described above. The most general type of rank-based model that has been studied in detail so far is the hybrid Atlas model (a type of second-order model\footnote{\emph{Second-order} refers to a model where the coefficients depend on name as well as rank, as opposed to \emph{first-order} models, which have coefficients that depend only on rank.}, see also \cite{fik13a}), introduced by Ichiba \emph{et al} \cite{ipbkf11} as follows:
\begin{align} \label{hybridAtlas}
&d Y_i(t) = \Big( \gamma+\gamma_i+\sum_{k=1}^n g_k\mathbbm{1}_{Q^{(i)}_k}(\mathfrak{Y}(t))\Big)\mathrm{d}t + \sum_{k=1}^n \rho_{ik}\mathrm{d}W_k(t)\nonumber\\
&\qquad\qquad + \sum_{k=1}^n \sigma_k\mathbbm{1}_{Q^{(i)}_k}(\mathfrak{Y}(t)) \mathrm{d}W_i(t),\qquad t\geq0 \\
&Y_i(0)=y_i,\qquad\qquad\qquad\qquad\qquad\qquad\quad\ i=1,\ldots,n,\nonumber
\end{align}
where $Y_i(\cdot):=\log X_i(\cdot)$, $\mathfrak{Y}(\cdot):=(Y_1(\cdot),\ldots,Y_n(\cdot))$, and $\{Q^{(i)}_k\}_{1\leq i,k\leq n}$ is a collection of polyhedral domains in $\mathbb{R}^n$, where $(y_1,\ldots,y_n)\in Q^{(i)}_k$ signifies that coordinate $y_i$ is ranked $k$th among $y_1,\ldots,y_n$. 
We can interpret the above as follows: when $Y(\cdot)\in Q^{(i)}_k$, namely, $Y_i(\cdot)$ is ranked $k$th among $Y_1(\cdot),\ldots,Y_n(\cdot)$, it behaves like a geometric Brownian motion with drift $g_k+\gamma_i+\gamma$ and variance $(\sigma_k+\rho_{ii})^2+\sum_{k\neq i}\rho^2_{ik}$. The constants $\gamma$, $\gamma_i$ and $g_k$ represent a common, a name-based and a rank-based drift respectively, whereas the constants $\sigma_k>0$ and $\rho_{ik}$ represent rank-based volatilities and name-based correlations, respectively. Under additional assumptions on these parameters, see Equation (2.3) in \cite{ipbkf11},
the model \eqref{hybridAtlas} admits a unique weak solution.

It is shown in \cite{ipbkf11} (see for instance their Figure 3) that certain models of the form \eqref{hybridAtlas} indeed lead to the empirical capital distribution curve. The authors also make a brief study of Cover and Jamshidian's universal portfolio in these markets, and show that the conditions for this portfolio to perform extremely well in the long run are naturally met in hybrid Atlas models. However, no further study of the performance of these portfolios is performed.

Note that \eqref{hybridAtlas} is a system of interacting Brownian particles --- this is an active area of research in both mathematical finance as well as pure probability theory, and a lot of progress has been made in recent years. For the sake of brevity, we will not discuss these articles here, but mention a subset of them: \cite{shk12}, \cite{fik13b}, \cite{fikp13}, \cite{iks13}, \cite{ips13}, and \cite{jr13}.


\subsection{Atlas model} \label{introAtlas}
One of the simplest and most studied types of rank-based models is the Atlas model, a first-order model that was introduced in \cite{bfk05} which assigns a non-zero growth rate only to the lowest-ranked stock, which has a positive growth rate and thus ``carries the entire market'' (hence the nomenclature). More precisely, we have \eqref{hybridAtlas} with $\gamma_i=\rho_{ik}=0$ for all $i,k=1,\ldots,n$, $\gamma=g>0,\ g_k=-g$ for $k\leq n-1$, and $g_n=(n-1)g$. It is also assumed that $$\sum_k \sigma^2_k>2 \max_k \{\sigma^2_k\},\qquad 0\leq \sigma^2_2-\sigma^2_1\leq \ldots\leq \sigma^2_n-\sigma^2_{n-1}.$$ In \cite{bfk05} it is shown that all stocks have the same asymptotic growth rate $\gamma$ in the Atlas model, i.e.~$$\lim_{T\to\infty}\frac{1}{T}\log X_i(T)=\gamma\quad \text{a.s.},\ \forall\ i=1,\ldots,n,$$ and that every stock spends roughly $(1/n)^\text{th}$ of the time in any rank: $$\lim_{T\to\infty}\frac{1}{T}\int_0^T \mathbbm{1}_{\mathcal{Q}^{(i)}_k}(\mathfrak{Y}(t))\mathrm{d}t=\frac{1}{n}\quad \text{a.s.},\qquad \forall\ 1\leq i,k\leq n.$$

The dynamics of the market weights have been studied in the Atlas model in \cite{ips13}. This article answers an open question in \cite[p. 170]{fk05} for the case of the Atlas model, namely to determine the distributions of $\mu_i(t)$, of $\mu_{(1)}(t)$ and of $\mu_{(n)}(t)$ for fixed $t>0$. Besides, it solves the problem put forward in Remark 5.3.8 in \cite{f02}, that is to find $$\lim_{T\to\infty}\int_0^T \mu_{(k)}(t)\mathrm{d}t\ \text{ for }\ k=1,\ldots,n.$$ The joint distribution of the long-term relative market weights is studied in a mean-field version of the Atlas model in \cite{jr13}, answering in part the open question in Remark 13.4 in \cite{fk09}. Finally, \cite{shk12} looks at the large-market limit of rank-based models, answering problem 5.3.10 in \cite{f02}.

Portfolio performance remains a practically unstudied topic in rank-based models; some first steps have been taken in \cite{ipbkf11}, but no long-term investment opportunities, let alone relative arbitrages (over finite horizons), have been shown to exist yet. It would be interesting to investigate whether one is able to construct portfolios in say the Atlas model with one of these properties.

\subsection{Rank-based functionally generated portfolios} \label{sec:rankport}
In \cite{f01}, Fernholz generalised the class of functionally generated portfolios of Section \ref{fgportfolios} to allow for functions that do not distinguish between market weights by name, but by rank. Namely, placing ourselves in a general It\^o model \eqref{model}, and applying Ito's rule for convex functions of semimartingales to the ranked market weights, it is shown in \cite[pp. 76--79]{f02} that
\begin{align} \label{semimIto}
\frac{d\mu_{(k)}(t)}{\mu_{(k)}(t)}&= \Big( \gamma_{p_t(k)}(t)-\gamma^\mu(t)+\frac{1}{2} \tau^\mu_{p_t(k)p_t(k)}(t) \Big)\mathrm{d}t + \frac{1}{2}\Big( \mathrm{d}\mathfrak{L}^{k,k+1}(t)-\mathrm{d}\mathfrak{L}^{k-1,k}(t) \Big)\nonumber\\ 
& \qquad +\sum_{\nu=1}^d(\sigma_{p_t(k) \nu}(t)-\sigma_{\mu\nu}(t))\mathrm{d}W_\nu(t),\qquad k=1,\ldots,n.
\end{align}
Here, $p_t(k)$ is the index of the stock that is ranked $k$th at time $t$, so that $\mu_{p_t(k)}(t)=\mu_{(k)}(t)$, and $\mathfrak{L}^{k,k+1}(t)\equiv \Lambda_{\Xi_k}(t)$ is the semimartingale local time at the origin accumulated by the nonnegative process 
\begin{equation} \label{firstdeflocalt}
\Xi_k(t):=\log\big(\mu_{(k)}(t)/\mu_{(k+1)}(t)\big),\quad t\geq0.
\end{equation}
Each $\mathfrak{L}^{k,k+1}(t)$ measures the cumulative effect of the changes that have occurred during the time interval $[0,t]$ between ranks $k$ and $k+1$; we set $\mathfrak{L}^{0,1}(\cdot)\equiv0\equiv \mathfrak{L}^{n,n+1}(\cdot)$.

In this setting, the master equation \eqref{eqmaster} can be generalised as follows: let the generating function $\mathbf{G}:U\to(0,\infty)$ be written as
$$\mathbf{G}(x_1,\ldots,x_n)=\mathfrak{G}(x_{(1)},\ldots,x_{(n)}),\quad x\in U,$$
for some function $\mathfrak{G}\in C^2(U)$ for $U\subset\Delta^n_+$ open, $\Delta^n_+$ as in \eqref{delta+}. Then Theorem 3.1 of \cite{f01} asserts the following counterpart to Lemma \ref{master}:
\begin{lem} \label{genmaster}
Let $\pi(\cdot)$ be the portfolio
\begin{equation} \label{underport}
\pi_{p_t(k)}(t)=\Big( D_k\log \mathfrak{G}(\mu_{(\cdot)}(t)) +1 - \sum_{{l}=1}^n \mu_{({l})}(t)D_{l}\log \mathfrak{G}(\mu_{(\cdot)}(t)) \Big)\cdot \mu_{(k)}(t)
\end{equation}
for $k=1,\ldots,n$, with $\mathfrak{G}\in C^2(U)$. The performance of $\pi(\cdot)$ relative to the market is
\begin{equation} \label{wow}
\log \left( \frac{V^{\pi} (T)}{V^\mu(T)}\right)  = \log \bigg( \frac{ \mathfrak{G}\big(\mu_{(\cdot)}(T)\big)}{ \mathfrak{G}\big(\mu_{(\cdot)}(0)\big)} \bigg)+ \Gamma(T),
\end{equation}
where 
\begin{align} \label{waaa}
\Gamma(T)&:= -\int_0^T \frac{1}{2\mathfrak{G}\big(\mu_{(\cdot)}(t)\big)}\sum_{i,j}^n D^2_{ij} \mathfrak{G}\big(\mu_{(\cdot)}(t)\big)\mu_{(i)}(t)\mu_{(j)}(t)\tau^\mu_{p_t(i)p_t(j)}(t)\mathrm{d}t\nonumber\\
&\qquad + \frac{1}{2}\sum_{k=1}^{n-1}\big(\pi_{p_t(k+1)}(t)-\pi_{p_t(k)}(t)\big)\mathrm{d}\mathfrak{L}^{k,k+1}(t).
\end{align}
We have used the notation $\mu_{(\cdot)}(t):=\big(\mu_{(1)}(t),\ldots,\mu_{(n)}(t)\big)$ here.
\end{lem}



Fernholz applies his generalised master equation \eqref{wow} in two settings in \cite{f01}: first to theoretically explain the `size effect', and second to study `leakage' --- see Sections \ref{secsize} and \ref{secleak} below. The above results have not yet been used to construct relative arbitrages --- we make some first steps towards this in Section \ref{secrankedDWP}.

\subsubsection{The size effect} \label{secsize}
This empirically observed effect refers to the tendency of small stocks to have higher long-term returns relative to large-capitalisation stocks. Equation \eqref{wow} can be used to explain this in the following way:
Let $m\in\{2,\ldots,n-1\}$ and let $\mathcal{G}_L(x)=x_{(1)}+\ldots+x_{(m)}$ and $\mathcal{G}_S(x)=x_{(m+1)}+\ldots+x_{(n)}$. These functions generate the large-stock portfolio 
\begin{equation}
\zeta_{p_t(k)}(\cdot) =
\begin{dcases} \frac{\mu_{(k)}(\cdot)}{\mathcal{G}_L(\mu(\cdot))}, &k=1,\ldots,m\\
0, &k=m+1,\ldots,n,
\end{dcases}
\end{equation}
and the small-stock portfolio 
\begin{equation}
\eta_{p_t(k)}(\cdot) =
\begin{dcases} 0, &k=1,\ldots,m\\
\frac{\mu_{(k)}(\cdot)}{\mathcal{G}_S(\mu(\cdot))}, &k=m+1,\ldots,n,
\end{dcases}
\end{equation}
respectively. Lemma \ref{genmaster} implies that the market-relative performances of these portfolios are
\begin{align} \label{biggie}
\log  \left(\frac{V^{\zeta} (T)}{V^\mu(T)}\right)  &= \log  \bigg(\frac{ \mathcal{G}_L\big(\mu(T)\big)}{ \mathcal{G}_L\big(\mu(0)\big)}\bigg) -\frac{1}{2}\int_0^T \zeta_{(m)}(t)\mathrm{d}\mathfrak{L}^{m,m+1}(t),\\
\log \left( \frac{V^{\eta} (T)}{V^\mu(T)}\right)  &= \log  \bigg(\frac{ \mathcal{G}_S\big(\mu(T)\big)}{ \mathcal{G}_S\big(\mu(0)\big)} \bigg)+ \frac{1}{2}\int_0^T \eta_{(m)}(t)\mathrm{d}\mathfrak{L}^{m,m+1}(t).
\end{align}
Combining these gives the performance of one relative to the other as
\begin{equation}\label{blablaa}
\log  \left(\frac{V^{\eta} (T)}{V^\zeta(T)}\right)  = \log  \bigg(\frac{ \mathcal{G}_S\big(\mu(T)\big) \mathcal{G}_L\big(\mu(0)\big)}{\mathcal{G}_L\big(\mu(T)\big) \mathcal{G}_S\big(\mu(0)\big) }\bigg) +\int_0^T \frac{\zeta_{(m)}(t)+\eta_{(m)}(t)}{2}\mathrm{d}\mathfrak{L}^{m,m+1}(t).
\end{equation}
Hence, as Fernholz remarks, if the market exhibits ``stability'', in the sense that the ratio of relative capitalization of small to large stocks remains stable over time, then the first term in \eqref{blablaa} will remain approximately constant in time. The integral with respect to the local time process, however, is increasing; hence $\log \big(V^{\eta} (T)/V^\zeta(T)\big)$ will be positive.

The above is an illustration of how the generalised master equation \eqref{wow} can be applied to make comparisons between portfolios generated by functions of ranked market weights; perhaps this can be used to make almost sure comparisons as well, and thus to construct relative arbitrages. To the author's knowledge, this has not been tried out yet, and is of great interest.

We inform the reader that in \cite{fk09}, the somewhat surprising observation is made that one can empirically estimate the local times used above, using the generating function for the large-cap portfolio, as follows:
\begin{equation} \label{localt}
\mathfrak{L}^{m,m+1}(\cdot)=\int_0^\cdot \frac{2}{\zeta_{(m)}(t)}\mathrm{d}\log\left(\frac{\mathcal{G}_L(\mu(t))V^\mu(t)}{\mathcal{G}_L(\mu(0))V^\zeta(t)}  \right),\quad m=1,\ldots,n-1.
\end{equation}

\subsubsection{Leakage} \label{secleak}
Another phenomenon that the above Lemma \ref{genmaster} allows us to study explicitly is `leakage', being the loss of value incurred by stocks exiting a portfolio contained in a larger market. Namely, consider, as in Example 4.2 in \cite{f01}, the diversity-weighted index of large stocks with parameter $r\in(0,1)$:
\begin{equation} \label{largeDWP}
\mu^\#_{p_t(k)}(t)=\begin{dcases} \frac{\big(\mu_{(k)}(t)\big)^r}{\sum_{l=1}^m \big(\mu_{(l)}(t)\big)^r}, &k=1,\ldots,m\\
0, &k=m+1,\ldots,n,
\end{dcases}
\end{equation}
which is generated by the function $\mathcal{G}_r(x)=\Big(\sum_{l=1}^m\big(x_{(l)}\big)^r \Big)^{1/r}$, analogously to \eqref{genfctdiv}. Now, \eqref{wow} implies that 
\begin{equation} \label{largeDWPrr}
\log \bigg(\frac{V^{\mu^\#}\!(T)}{V^\mu(T)}\bigg) = \log \bigg(\frac{\mathcal{G}_r(\mu(T))}{\mathcal{G}_r(\mu(0))}\bigg) + (1-r)\int_0^T\gamma^*_{\mu^\#}(t)\mathrm{d}t - \int_0^T \frac{\mu^\#_{(m)}}{2}\mathrm{d}\mathfrak{L}^{m,m+1}(t).
\end{equation}
It is then shown in Example 4.3.5 in \cite{f02} that, using \eqref{biggie}, we get
\begin{align} \label{leakleak}
\log \bigg(\frac{V^{\mu^\#}\!(T)}{V^\zeta(T)}\bigg) &= \log\bigg( \frac{\mathcal{G}_r \big(\zeta_{(1)}(T),\ldots,\zeta_{(m)}(T)\big)}{\mathcal{G}_r \big(\zeta_{(1)}(0),\ldots,\zeta_{(m)}(0)\big)}\bigg) + (1-r)\int_0^T\gamma^*_{\mu^\#}(t)\mathrm{d}t\nonumber\\
&\qquad - \int_0^T \frac{\mu^\#_{(m)}-\zeta_{(m)}(t)}{2}\mathrm{d}\mathfrak{L}^{m,m+1}(t);
\end{align}
this follows from the scale-invariance property 
$$\frac{\mathcal{G}_r(x_1,\ldots,x_n)}{x_1+\ldots+x_n}=\mathcal{G}_r \left( \frac{x_1}{x_1+\ldots+x_n},\ldots,\frac{x_1}{x_1+\ldots+x_n}\right).$$
Since for $r\in(0,1)$ and the diversity-weighted portfolio $\mu^{(r)}_i(\cdot)$ we have
\begin{equation}
\min_i \mu^{(r)}_i(t)=\frac{\big(\min_i \mu_i(t)\big)^r}{\sum_j \big(\mu_j(t)\big)^r}\geq\min_i \mu_i(t),
\end{equation}
we get that $\mu^\#_{(m)}(\cdot)\geq\zeta_{(m)}(\cdot)$; hence the last integral in \eqref{leakleak} is monotonically increasing in $T$. Fernholz typifies it as measuring the ``leakage" that occurs when a cap-weighted portfolio is contained inside a larger market of $n$ stocks, and stocks ``leak" from the cap-weighted to the market portfolio.


\section{Portfolio optimisation}\label{secoptarb} 
In the past, the main forms of portfolio optimisation in financial mathematics have been expected utility maximisation, mean-variance optimisation and Kelly's criterion. What these approaches have in common is that they require one to take the expectation of (functions of) the future wealth, and since one does not know the future behaviour of stocks with certainty, model assumptions and parameter estimations need to be made. Although the volatility process can be estimated to a fair degree of certainty with sufficient data, the drift process $b(\cdot)$ of a stock is notoriously difficult to estimate accurately (see for instance \cite{lcgr01} and \cite{mmima07}). Since all of the resulting optimal strategies in the above approaches depend explicitly on the drift processes of the stocks that they invest in, their applicability in real life is strongly limited by the inherent uncertainty of parameter estimations, and their real-world performance is not as optimal as predicted in theory (see \cite{uppal09}). One of the largest advantages of the SPT approach is that it overcomes this problem by considering only portfolios (namely FGPs, see Section \ref{fgportfolios}) which do not depend on any unobservable quantities, but only the current market configuration, which is observable. Although the performance of these portfolios is not `optimal' in any sense, some of them can be shown to outperform the market index in a pathwise sense, which together with their implementability makes them a very attractive class of investment opportunities. Recently the question has been raised of finding the `best' strategy with these properties, i.e.~the `optimal relative arbitrage' --- this question has been partly answered in \cite{fernk10}, see Section \ref{optRA} below. 

Another portfolio selection criterion which was proposed in \cite{fernk10}, and solved in the complete Markovian case in \cite{bhs12}, is to characterise, in terms of the market configuration and time to maturity, the minimum amount of initial capital with which one can beat the market portfolio with a certain probability. We will not discuss this criterion in detail, and refer the reader to \cite{bhs12}.

\subsection{Num\'eraire portfolio \& expected utility maximisation}   \label{numEU}
Although it is not the portfolio selection criterion of interest, it has been shown that expected utility maximisation can be performed in very general market models. Namely, in \cite{kk07} it is shown that, in a general semimartingale setting, the expected utility maximisation problem to find $\hat{h}(\cdot)\in\mathcal{A}_w$ such that
\begin{equation} \label{EUmax}
\mathbb{E}[U(V^{w,\hat{h}}(T))]=\sup_{h(\cdot)\in\mathcal{A}_w}\mathbb{E}[U(V^{w,h}(T))],
\end{equation}
with $U:(0,\infty)\to\mathbb{R}$ satisfying the Inada conditions, can be solved if and only if NUPBR holds, to wit, if and only if an LMD exists (see Definition \ref{defZ}). Hence an ELMM is not required for solving this problem, even though historically it has been studied mainly in markets where arbitrages do not exist.

The main results of \cite{kk07} relate to a very special portfolio called the \emph{num\'{e}raire portfolio}, which is a portfolio $\rho(\cdot)$ with the property that 
\begin{equation} \label{eqdefnumport}
V^\pi(\cdot)/V^\rho(\cdot) \text{ is a supermartingale for all portfolios } \pi(\cdot).
\end{equation}
Theorem 4.12 in \cite{kk07} shows that this portfolio exists if and only if NUPBR holds. Furthermore, the num\'eraire portfolio is shown to be characterised by the following properties:
\begin{itemize}
\item it maximises the growth rate (i.e., the drift rate of $\log V^\pi(\cdot)$) over all portfolios $\pi(\cdot)$;
\item it maximises the asymptotic growth rate over all portfolios: for any increasing process $H(\cdot)$ with $\lim_{t\to\infty}H(t)=\infty$, $$\limsup_{t\to\infty} \frac{1}{H(t)}\log\bigg(\frac{V^\pi(t)}{V^\rho(t)}\bigg) \leq0\qquad \forall \pi(\cdot);$$
\item it solves the relative log-utility maximisation problem: $$\mathbb{E}\Big[ \log \bigg( \frac{V^\pi(T)}{V^\rho(T)}\bigg) \Big]\leq0\qquad \forall \pi(\cdot).$$
\end{itemize}
We also mention that no relative arbitrage can be constructed with respect to the num\'eraire portfolio --- see, for instance, Remark 6.5 in \cite{fk09}. In VSMs, the num\'eraire portfolio can be computed explicitly. In general, however, the num\'eraire portfolio depends on the drift coefficients as well as the volatility coefficients of the market model, and can therefore not be computed explicitly if we do not assume any expert knowledge on the exact dynamics of these processes, which is what is usually done in SPT.

\subsection{Optimal relative arbitrage} \label{optRA}
The question that naturally arises when constructing relative arbitrages is the following: is there a `best' such investment strategy? This open question was posed in Remark 11.5 in the survey paper \cite{fk09}, and was answered for the case of a complete Markovian NUPBR market in \cite{fernk10}.

The way in which Daniel Fernholz and Ioannis Karatzas interpreted the question of optimal relative arbitrage in \cite{fernk10} is as follows: on a given time horizon $[0,T]$, what is the smallest amount of initial capital from which one can match or exceed at time $t=T$ the market capitalisation $X(T)$? In equation form, find
\begin{equation} \label{fernkoptRA}
\mathfrak{u}(T):=\inf\{w>0\ |\ \exists h(\cdot)\in \mathcal{A}_{wX(0)}\text{ s.t. } V^{wX(0),h}(T)\geq X(T)\quad \mathbb{P}\text{-a.s.}  \}.
\end{equation}
The corresponding strategy $\hat{h}(\cdot)$ which achieves $V^{\mathfrak{u}(T)X(0),\hat{h}}(T)\geq X(T)\ \ \mathbb{P}$-a.s.\! is then the optimal arbitrage.
The authors make many technical assumptions on \eqref{model}, the most important of which are that $\mathbb{F}=\mathbb{F}^W$, $d=n$, that $\sigma(\cdot)$ is invertible, that there exists a square-integrable MPR and that the market is Markovian (i.e., $b(\cdot)$ and $\sigma(\cdot)$ are deterministic functions of time and $\mathfrak{X}(\cdot):= (X_1(\cdot),\ldots,X_n(\cdot))$, the vector of capitalisation processes), under which it holds (by results on hedging in complete markets, e.g.\! in \cite{fk09}, or see Section \ref{hedgeEU}) that $1/\mathfrak{u}(T)$ gives the highest return on investment that one can achieve relative to the market over $[0,T]$;
\begin{equation} \label{fernkrr}
1/\mathfrak{u}(T)=\sup\{q\geq1\ |\ \exists h(\cdot)\in \mathcal{A}_{1}\text{ s.t. } V^h(T)\geq q V^\mu(T)\quad \mathbb{P}\text{-a.s.}  \}.
\end{equation}

With the help of the F\"ollmer measure, see Section \ref{strictRN}, the authors of \cite{fernk10} are able to derive that the optimal arbitrage strategy has the form
\begin{equation} \label{fernkans1}
\frac{\hat{h}_i(t)}{V^{\hat{h}}(t)}=X_i(t) D_i\log U(T-t,\mathfrak{X}(t))+\mu_i(t),\qquad i=1,\ldots,n,\ t\in[0,T],
\end{equation}
with $U:[0,\infty)\times(0,\infty)^n\to(0,1]$ the smallest non-negative solution of the linear parabolic partial differential inequality 
\begin{align}\label{fernkans2}
&\frac{\partial U}{\partial\tau}(\tau,\mathbf{x})\geq \hat{\mathcal{L}}U(\tau,\mathbf{x}),\qquad (\tau,\mathbf{x})\in(0,\infty)\times(0,\infty)^n,\\
&U(0,\cdot)\equiv1,\nonumber
\end{align}
for the linear operator
\begin{equation}\label{fernkans3}
\hat{\mathcal{L}}f:=\frac{1}{2}\sum_{i,j=1}^n x_ix_ja_{ij}(\mathbf{x})D^2_{ij}f+\sum_{i=1}^nx_i\left(\sum_{j=1}^n \frac{x_ja_{ij}(\mathbf{x})}{x_1+\ldots+x_n}  \right)D_if.
\end{equation}
Note that, by the above, $\hat{h}(\cdot)$ depends only on the covariance structure of the market, \emph{not} the rates of return. The authors go on to interpret their results in a probabilistic way, by considering the $([0,\infty)^n\setminus\{\mathbf{0}\})$-valued diffusion process $\mathfrak{Y}(\cdot)$ with infinitesimal generator $\hat{\mathcal{L}}$ and $\mathfrak{Y}(0)=\mathfrak{X}(0)$; for instance, $U(T,\mathfrak{X}(0))$ is identified as the probability that $\mathfrak{Y}(\cdot)$ never hits the boundary of $[0,\infty)^n$ by time $t=T$. Also, in this auxiliary market, the auxiliary market weights $Y_i(\cdot)/(Y_1(\cdot)+\ldots+Y_n(\cdot))$, $i=1,\ldots,n$, have the num\'eraire property, as defined in \eqref{eqdefnumport} in the previous Section \ref{numEU}.

Fernholz and Karatzas extend their characterisation of the highest achievable market-relative return to markets with ``Knightian" uncertainty in \cite{fk11}. They present several different forms of this characterisation, one of which is in terms of the smallest positive supersolution to a non-linear Hamilton-Jacobi-Bellman PDE.

It would be interesting to extend the above results to non-Markovian and, even more so, incomplete markets. Furthermore, the optimal arbitrage \eqref{fernkans1} is characterised in a very indirect way, as the smallest non-negative solution to a complicated equation, to which the explicit solution can only be computed in fairly trivial toy examples. However, it might be possible to use numerical methods to compute optimal portfolios using \eqref{fernkans1}--\eqref{fernkans3}. Still, the strength of these results will always be limited by the quality of covariance estimations.

\subsubsection*{Two-dimensional optimisation} \label{2Dopt}
In \cite{palw13} some initial attempts are made at deriving more concrete results for optimal relative arbitrage, considering only (general) two-dimensional markets and portfolios of the form $\pi(\cdot)=(q(Y(\cdot)),1-q(Y(\cdot)))$, with $q$ a deterministic function and $Y(\cdot):=\log \big(X_2(\cdot)/ X_1(\cdot)\big)$ the log-price ratio. In the special case that one is restricted to constant-weighted portfolios $\pi(\cdot)=(Q,1-Q)$, $Q\in[0,1]$, Pal and Wong show in their Proposition 4.2 that for $[0,\tau]$ an excursion of $Y(\cdot)$, i.e.~an interval such that $\tau=\inf\{t>0\ |\ Y(t)=Y(0)\}$, one has 
\begin{equation}\label{palwone}
\log \bigg(\frac{V^\pi(\tau)}{V^\mu(\tau)}\bigg)-\log \bigg(\frac{V^\pi(0)}{V^\mu(0)}\bigg) =\frac{1}{2}\int_0^\tau q(1-q)\mathrm{d}\left<Y\right>\!(t)=\frac{1}{2}q(1-q)\left<Y\right>\!(\tau).
\end{equation}
Hence, this quantity is maximised when $q=\frac{1}{2}$, so when $\pi(\cdot)$ is the equally-weighted portfolio. For more general functions $q\in C(\mathbb{R},[0,1])$ of finite variation, they show that for $F$ an antiderivative of $q$, and for all $t\geq0$,
\begin{align}\label{palwtwo}
\log\bigg( \frac{V^\pi(t)}{V^\mu(t)}\bigg)-\log \bigg(\frac{V^\pi(0)}{V^\mu(0)}\bigg)&= F(Y(t))-F(Y(0))\\
&\ +\frac{1}{2}\int_{-\infty}^\infty L^y(t)\Big( q(y)(1-q(y))\mathrm{d}y+\mathrm{d}(-q(y)) \Big),\nonumber
\end{align}
with $L^y(t)$ the local time of $Y$ at a point $y\in\mathbb{R}$ over $[0,t]$. This expression is used to explicitly compute the optimal function $q$ in two examples, namely when $Y(\cdot)$ is a Bang-bang process and when it is an Ornstein-Uhlenbeck process. The objective is to maximise \eqref{palwtwo}, and the method used is to assume that one is given a `weight function' $w$ which represents the expected local time at each value of $y$, thus translating beliefs and statistical information about the dynamics of $Y(\cdot)$. 

These preliminary results look very promising, and it would be of great value to extend them to higher-dimensional markets. Pal and Wong make some first attempts in \cite{palw13} at optimising within a class of FGPs, incorporating (statistical) information in this optimisation and applying information theory to SPT. These are all very undiscovered areas, and are of great interest to the author; see Sections \ref{futinfo} and \ref{futinc}.

\section{Hedging in SPT framework} \label{sechedge} 
The theory of hedging in markets without an ELMM has only been developed in recent years. Some initial progress was made in \cite{ch05} and \cite{fk09}, but only later were hedging strategies computed in the European case (and a Markovian market, see \cite{ruf13a}) and optimal stopping times characterised in the American case (see \cite{bkx07}), the latter solving an open problem in \cite{fk09}. Furthermore, \cite{kard12val} derives valuation formulas for both European and American exchange options in general semimartingale models with the possibility of default, showing how traditional parity formulas are altered when NFLVR fails.

\subsection{Hedging European claims} \label{hedgeEU}
The topic of hedging European claims in markets without an ELMM was first explored by Fernholz and Karatzas in Section 10 of their survey \cite{fk09}. They considered those claims given by an $\mathcal{F}(T)$-measurable random variable $Y$ with 
\begin{equation} \label{fkdefY}
0<y:=\mathbb{E}[Z(T)Y]<\infty,
\end{equation}
and as usual defined the upper hedging price of $Y$ as
\begin{equation} \label{fkhedge}
\mathcal{U}^Y(T):=\inf\{w>0\ |\ \exists h(\cdot)\in \mathcal{A}_{w}\text{ s.t. } V^{w,h}(T)\geq Y\quad \mathbb{P}\text{-a.s.}  \}.
\end{equation}
In other words, $\mathcal{U}^Y(T)$ is the minimal amount required at time $t=0$ to be able to super-replicate the claim $Y$ by terminal time $t=T$.\footnote{If $Y=X(T)$, the total market value at time $t=T$, this becomes the problem of finding the optimal relative arbitrage --- see Section \ref{optRA}.} Since an ELMM may not exist, one can no longer rely on the well-known hedging result from ``classical'' mathematical finance (see \cite{ds95c}), namely that 
\begin{equation} \label{classichedge}
\mathcal{U}^Y(T)=\sup_{\mathbb{Q}\in\mathcal{M}}\mathbb{E}^\mathbb{Q}[Y],
\end{equation}
with $\mathcal{M}$ the set of ELMMs. However, as long as the infimum in \eqref{fkhedge} is finite, the fact that $Z(\cdot)V^{w,h}(\cdot)$ is a positive local martingale (and thus supermartingale) for any $w>0$ and $h(\cdot)\in\mathcal{A}_w$ implies that
\begin{equation} \label{fksuperhedge}
w=Z(0)V^{w,h}(0)\geq\mathbb{E}[Z(T)V^{w,h}(T)]\geq \mathbb{E}[Z(T)Y]=y.
\end{equation}
Since this holds for all $w>0$, Fernholz and Karatzas conclude that $\mathcal{U}^Y(T)\geq y$ (which is also trivially true if the infimum in \eqref{fkhedge} is infinite). The authors go on to note that if $\mathbb{F}=\mathbb{F}^W$ and $n=d$ in \eqref{model}, i.e.~the number of driving BMs is equal to the number of stocks, an application of the martingale representation theorem to the process $\mathbb{E}[Z(T)Y\, |\, \mathcal{F}_t]$ implies that any claim $Y$ can be replicated, since this construction only requires the existence of an LMD, not an ELMM. Hence, one can have completeness --- that is, $\mathcal{U}^Y(T)= y$ for every claim $Y$ as in \eqref{fkdefY} --- in a market where NFLVR fails (for instance, a market satisfying \eqref{eqdefdiv}).

It was in \cite{ruf13a} that progress was made in characterising hedging strategies for European claims. Namely, Ruf shows that in a Markovian incomplete setting (i.e.~\eqref{model} with $b(\cdot)$ and $\sigma(\cdot)$ deterministic functions of time and the current market configuration), the optimal hedging strategy, in the sense of smallest initial wealth required, for a European claim is a delta hedge, which is well-known in NFLVR markets. Besides the Markovian assumption, Ruf assumes the existence of a square-integrable MPR, which implies NUPBR (see Section \ref{secarb}); however, an FLVR may exist. It is observed that these assumptions imply the existence of a Markovian MPR $\theta(\cdot)$ in $L^2$. Under the assumption that the claim $Y$ is measurable with respect to $\mathcal{F}^\mathfrak{X}(T)\subset\mathcal{F}(T)$, with $\mathfrak{X}(\cdot)$ the vector of capitalisation processes, Proposition 3.1 in \cite{ruf13a} asserts that for any square-integrable MPR $\nu(\cdot)$, one has the surprising property 
\begin{equation} \label{rufmpr}
\mathbb{E}\left[\frac{Z^\nu(T)}{Z^\nu(t)}Y\ \Big|\ \mathcal{F}(t)\right]\leq \mathbb{E}\left[\frac{Z^\theta(T)}{Z^\theta(t)}Y\ \Big|\ \mathcal{F}(t)\right],
\end{equation}
using the notation \eqref{defItoLMD}. This result, which is interesting on its own, implies that the hedging price $h^p$ of a European claim $Y=p(S(T))$ does not depend on the choice of MPR. It allows the author to prove, in Theorem 4.1 of \cite{ruf13a}, that the optimal hedging strategy is given by $\eta^p(t,s) := \nabla h^p(t,s)$, for 
\begin{equation} \label{rufdelta}
h^p(t,s):=\mathbb{E}\left[\frac{Z^\theta(T)}{Z^\theta(t)}p(S(T))\ \Big|\ S(t)=s\right],
\end{equation}
and requires initial wealth $\nu^p:=h^p(0,S(0))$.\footnote{This result holds under the assumption that $h^p\in C^{1,2}(\mathcal{U}_{t,s})$ for some neighbourhood $\mathcal{U}_{t,s}$ of $(t,s)$, for which the author gives sufficient conditions on the market model and payoff function $p$.} Namely, then $V^{\nu^p,\eta^p}(t)=h^p(t,S(t))$ $\forall t \in [0,T]$, and for any $\tilde{\nu}>0$ and admissible super-replicating strategy $\tilde{\eta}\in\mathcal{A}_{\tilde{\nu}}$ we have $\tilde{\nu}\geq\nu^p$. This is remarkable, since it implies that any claim $Y$ as above is replicable, despite the market not being complete ($d>n$ in general). The hedging price function $h^p$ is also shown to solve a PDE which depends only on the covariance structure of the model. It is then noted that if $h^p$ is sufficiently smooth, it can be characterised as the smallest nonnegative solution to that PDE, with terminal condition $h^p(T,s)=p(s)$. Ruf then derives a modified put-call parity, and investigates a change of measure technique to simplify the computation of $h^p$ --- see Section \ref{strictRN} below for a discussion.

Of course, it would be interesting to see how the above results could be extended to non-Markovian market models and more complex claims. However, as Ruf points out in Remark 4.2 of \cite{ruf13a}, Theorem 4.1 can no longer be expected to hold in a model where the mean rates of return and volatilities have additional stochasticity. Furthermore, by definition one will not be able to hedge all claims in an incomplete market; hence any more general results will have to be weaker and of a different form.

\subsection{Hedging American claims} \label{hedgeUS}
In Remark 10.4 of their survey \cite{fk09}, Fernholz and Karatzas posed the open question of developing a theory for pricing American claims in NUPBR markets, in particular characterising the optimal exercise time. This problem was solved by Bayraktar \emph{et al} \cite{bkx07}, where the authors characterise the hedging strategy and optimal stopping time for American options; in particular, they solve the optimisation problem
\begin{equation} \label{OS} \tag{OS}
\left.\begin{aligned}
&\text{compute }\ v:=\sup_{\tau\in\mathfrak{T}}\mathbb{E}[Z(\tau)g(X_1(\tau))]\quad\\
&\text{find }\ \hat{\tau}\ \text{ s.t. }\ \mathbb{E}[Z(\hat{\tau})g(X_1(\hat{\tau}))]=v\quad
\end{aligned}\right\}
\end{equation}
where $\mathfrak{T}$ is the set of all stopping times, $g:\mathbb{R}_+\to\mathbb{R}_+$ is a convex payoff function, $X_1(\cdot)$ is the one-dimensional price process (a positive semimartingale) and $Z(\cdot)$ is an LMD, so that $Z(\cdot)X_1(\cdot)$ is a local martingale.

The solution to \eqref{OS} is given in Theorem 2.4 of \cite{bkx07} as follows:
\begin{itemize}
\item The value of the problem is $v=\mathbb{E}[Z(\tau^*)g(X_1(\tau^*))]+\delta(\tau^*)$.
\item A stopping time $\hat{\tau}\in\mathfrak{T}$ is optimal if and only if $\tau^*\leq\hat{\tau}$ and $\delta(\hat{\tau})=0$.
\item Optimal stopping times exist if and only if $\delta(\tau^*)=0$, in which case $\tau^*$ is the smallest optimal stopping time.
\end{itemize}
We explain the notation; write $L(\cdot):=Z(\cdot)X_1(\cdot)$ and $Y(\cdot):=Z(\cdot)g(X_1(\cdot))$. The function $\delta:\mathfrak{T}\to\mathbb{R}_+$ is defined in Lemma 2.1 of \cite{bkx07} as\footnote{It is assumed that $X(\infty):=\lim_{t\to\infty}X(t)$ exists.}
\begin{equation}
\delta(\tau):= \uparrow\lim_{n\to\infty}\mathbb{E}[Y(\tau\wedge\sigma^n)]-\mathbb{E}[Y(\tau)],
\end{equation}
where $(\sigma^n)_{n\in\mathbb{N}}$ is a localising sequence for $L(\cdot)$, meaning that $L(\sigma^n\wedge \cdot)$ is a uniformly integrable martingale for all $n\in\mathbb{N}$, and $\sigma^n\to\infty$ as $n\to\infty$. It is shown that $\delta$ is nonnegative and independent of the choice of $(\sigma^n)_{n\in\mathbb{N}}$. Furthermore, $\tau^*$ is a certain stopping time --- see equation (2.3) in \cite{bkx07}. It is remarked that if the problem has a finite horizon $T\in\mathfrak{T}$, the value of \eqref{OS} is $v=\mathbb{E}[Y(T)]+\delta(T)$, which had been shown before in Theorem A.1 in \cite{ch05}.

If the condition $\delta(\tau^*)=0$ does not hold, Bayraktar \emph{et al} show that there does not exist an optimal exercise time. Also, they show explicitly in Example 3.2 that Merton's `no early exercise theorem' may fail in markets without NFLVR, i.e., that it may be optimal to exercise before terminal time $T$. It should be mentioned that \cite{bkx07} proves the above results in a model where there is also a positive, non-increasing discounting process $\beta(\cdot)$, which may become zero at some point. Here, we have ignored this added complication of the model for the sake of presentation.

Despite it being very theoretical, Bayraktar \emph{et al} are able to apply their Theorem 2.4 in some toy examples, showing the optimality of candidate stopping times. However, their paper does not offer a way of searching for such candidates, and thus the applicability of their results in real-world models remains very restricted for now.

\section{Non-equivalent measure changes} \label{secmeasure}
Despite the fact that ELMMs typically do not exist in the general type of It\^o model \eqref{model} that is considered in SPT, as NFLVR is allowed to and usually does fail, one can still make meaningful but possibly non-equivalent changes of measure. In SPT, such measure changes have thus far been made for two purposes: to create a measure that takes the r\^{o}le of an ELMM (see \cite{fernk10} and \cite{ruf13a}), or to construct markets in which NFLVR fails (see \cite{or06}, \cite{rr13} and \cite{ct13}).

\subsection{Strict local martingale Radon-Nikodym derivatives} \label{strictRN}
Building on F\"{o}llmer's work which introduced the ``exit measure'' of a supermartingale, some articles related to SPT have made a change of measure that generalises changing to a riskless measure: \cite{fernk10} does this in order to characterise the optimal arbitrage relative to the market, and \cite{ruf13a} to simplify the computation of hedging strategies of European claims in Markovian markets. For this, the necessary technical assumption is made that the probability space $\Omega$ is the space of right-continuous paths $\omega:[0,T)\to \mathbb{R}^n\cup\{\Delta\}$, $0<T\leq\infty$, where $\Delta$ is the so-called ``absorbing point'': it has the property that 
$$\omega(t)=\Delta,\quad t\in[0,T]\qquad \Rightarrow\qquad \omega(u)=\Delta, \quad \forall u\in[t,T], \forall \omega\in\Omega.$$
Under this condition, the ``exit measure'' (or ``F\"{o}llmer measure'') of a $\mathbb{P}$-local martingale $Y(\cdot)$ is defined in \cite{follmer72} as the measure $\mathfrak{P}$ on the predictable $\sigma$-algebra of $[0,\infty]\times\Omega$ given by
\begin{equation} \label{follmer}
\mathfrak{P}((T,\infty]\times A):=\frac{\mathbb{E}[Y(T)\mathbf{1}_A]}{\mathbb{E}[Y(0)]},\qquad A\in\mathcal{F}(T), T\in[0,\infty).
\end{equation}
As is explained in Section 7 of \cite{fernk10}, Theorem 4 in \cite{ds95} and Theorem 1 and Lemma 4 of \cite{pp07} give that for each positive and continuous $\mathbb{P}$-supermartingale $\Lambda(\cdot)$, a probability measure $\mathbb{Q}$ exists such that
\begin{itemize}
\item $\mathbb{P}\ll \mathbb{Q}$
\item $\Lambda(\cdot)$ is a continuous $\mathbb{Q}$-martingale
\item $\mathrm{d}\mathbb{P}/\mathrm{d}\mathbb{Q}=\Lambda(T)$
\end{itemize}
The case $\Lambda(\cdot):=X(0)/Z^\theta(\cdot)X(\cdot)$, with $X(\cdot)$ the total market capitalisation as defined in \eqref{defXX} and $Z^\theta(\cdot)$ the LMD as defined in \eqref{defItoLMD}, is studied and used in \cite{fernk10}; the case $\Lambda(\cdot):=1/Z^\theta(\tau^\theta\wedge\cdot)$, with $\tau^\theta$ the first hitting time of zero by the process $1/Z^\theta(\cdot)$, is used in Section 5 of \cite{ruf13a}. In both cases, the absorbing point $\Delta$ represents an explosion of $Z^\theta(\cdot)$, which does not occur under $\mathbb{P}$ by assumption but is possible under the probability measure $\mathbb{Q}$ as defined above. That is, let $\tau^\theta_n=\inf\{t\in[0,T]: Z^\theta(t)\geq n \}$ and define $\tau^\theta=\lim_{n\to\infty}\tau^\theta_n $; then in Proposition 1 of \cite{fernk10} it is proved that
\begin{equation} \label{fkfollmer}
\mathfrak{P}((T,\infty]\times \Omega)=\mathbb{Q}(\tau^\theta>T)=\mathfrak{u}(T),
\end{equation}
a representation of the F\"{o}llmer measure where $\mathfrak{u}(T)$ is the smallest wealth required for creating an arbitrage relative to the market, as in \eqref{fernkoptRA} (i.e., the wealth necessary for creating an optimal arbitrage; see Section \ref{optRA}).

Ruf, on the other hand, proves a generalised Bayes' rule for the measure $\mathbb{Q}$ defined above (see Theorem 5.1 in \cite{ruf13a}), which shows that $\mathbb{Q}$ behaves as an ELMM up to the explosion time $\tau^\theta$. He uses it to derive the dynamics up to $\tau^\theta$ of the price processes and of $1/Z^\theta(\cdot)$ under $\mathbb{Q}$. Because all of these are $\mathbb{Q}$-local martingales, the computation of hedging strategies can be done more easily under this measure, and can then be translated back to $\mathbb{P}$ by a known change of measure. An example of this is given when the price process is a three-dimensional Bessel process with drift, and the claim is a European call option --- see Example \ref{rufex} in our paper.

The aforementioned articles are thus far the only works related to SPT which make use of the F\"{o}llmer measure, and there may be more possibilities for applying this technique elsewhere. Other articles which study supermartingales as Radon-Nikodym densities are \cite{imp11} and \cite{pr13}.

\subsection{Constructing markets with arbitrage}
Changing measure to construct a market that allows for arbitrage is not entirely new, see for instance the construction of arbitrage possibilities in Bessel processes in \cite{ds95}. However, this technique has only recently been applied to construct markets with properties relevant to SPT. Other ways of constructing these NUPBR markets in which NFLVR fails have also been suggested, namely through filtration shrinkage \cite{protter13} and enlargement \cite{fontetal13}. 

Osterrieder and Rheinl\"{a}nder were the first to relate techniques of non-equivalent measure changes to SPT, in their article \cite{or06}. They present a way of constructing diverse markets (in the sense of \eqref{eqdefdiv}) which is more generic than the explicit construction in \cite{fkk05}, and show that arbitrages exist in such markets by design.
More precisely, the authors define their diverse market models by making a non-equivalent measure change from a pre-model. Under this pre-model, the time-evolution of the vector of capitalisation processes $\mathfrak{X}(\cdot)$ is governed by a probability measure $\mathbb{P}^0$, and is given explicitly as the stochastic exponential $\mathcal{E}(M(\cdot))$ of an unspecified $n$-dimensional continuous $\mathbb{P}^0$-local martingale $M(\cdot)$. Hence, $\mathbb{P}^0\in\mathcal{M}^e(\mathfrak{X})$, where $\mathcal{M}^e(\mathfrak{X})$ denotes the set of ELMMs for $\mathfrak{X}(\cdot)$, and NFLVR holds in the pre-model. One crucial assumption, baptised the ``non-degeneracy'' assumption, is made on the pre-model. It rules out diversity, and is shown to hold in It\^o models of the form \eqref{model} with assumptions \eqref{BV} and \eqref{ND} on the covariance structure; it is
\begin{equation}\label{orND}
\exists\, T>0,\ \delta\in(0,1)\ \text{ s.t. } \left\{\begin{aligned}
&\inf_{\mathbb{Q}\in\mathcal{M}^e(\mathfrak{X})} \mathbb{Q}\Big( \sup_{0\leq t\leq T} \mu_{(1)}(t)\geq 1-\delta \Big) >0\\
& \mathbb{P}^0\Big( \sup_{0\leq t\leq T} \mu_{(1)}(t)\geq 1-\delta \Big) <1.
\end{aligned}\right.
\end{equation}
Osterrieder and Rheinl\"{a}nder then make the following straightforward measure change:
\begin{equation} \label{ordQdP}
\frac{\mathrm{d}\mathbb{P}}{\mathrm{d}\mathbb{P}^0}=\begin{dcases}
0 & \text{if } \sup_{0\leq t\leq T} \mu_{(1)}(t)\geq 1-\delta \\
c & \text{otherwise},
\end{dcases}
\end{equation}
with $c\in\mathbb{R}$ a normalising constant, which exists by assumption \eqref{orND}. It is clear that diversity holds under $\mathbb{P}$ by construction, and that $\mathbb{P}\ll \mathbb{P}^0$ but $\mathbb{P}\not\sim \mathbb{P}^0$. Finally, using an optional decomposition theorem from \cite{follmk97}, it is shown that the $\mathbb{P}^0$-super-replicating strategy of the claim $$\mathbf{1}_{\big\{\frac{\mathrm{d}\mathbb{P}}{\mathrm{d}\mathbb{P}^0}>0  \big\}}$$ is an arbitrage under $\mathbb{P}$ over the horizon $[0,T]$, in the sense of the first point of Definition \ref{defarb} (i.e., as in \cite{ds94}). Hence NFLVR fails under $\mathbb{P}$.

We note that the results from \cite{or06} are mainly theoretical, in that they apply an absolutely continuous measure change to show \emph{existence} of diverse markets where arbitrage opportunities exist. However, this construction is not suitable for creating \emph{explicit} models, as the computations become undoable. Namely, \cite{or06} derives the $\mathbb{P}$-dynamics of $X(\cdot)$, but the drift process in this expression is given as a function of the density process $Z(\cdot)$ in Lenglart's extension of Girsanov's theorem, and a projection of $Z(\cdot)$ using the Galtchouk-Kunita-Watanabe decomposition --- hence it is implicit. 

A more generic result is proven in \cite{rr13}, where the authors demonstrate a systematic way of constructing markets in which NFLVR fails, but the weaker NUPBR holds. Similarly to \cite{or06}, Ruf and Runggaldier take a pre-model (governed by $\mathbb{P}^0$) in which the capitalisation processes $X_i(\cdot)$, $i=1,\ldots, n$ are non-negative local martingales. They then assume the existence of a $\mathbb{P}^0$-martingale $Y(\cdot)$, the candidate density process, with $Y(0)=1$ and which satisfies assumptions \eqref{rr1} and \eqref{rr2} below. Let $\tau$ denote the first hitting time of $0$ by $Y(\cdot)$; the first assumption is that
\begin{equation} \label{rr1}
\mathbb{P}^0(\tau\leq T)>0\qquad \text{and} \qquad \mathbb{P}^0(\{Y(\tau-)>0\}\cap\{\tau\leq T\})=0.
\end{equation}
The second assumption is as follows:\footnote{We use the notation $\mathcal{A}$ of Definition \ref{tradingstrat} for the set of admissible trading strategies.}
\begin{equation} \label{rr2}
\exists\, x\in(0,1),\, h(\cdot)\in\mathcal{A}\ \text{ s.t. }\ V^{x,h}(T)\geq \mathbf{1}_{\{Y(T)>0\}}\quad \mathbb{P}^0\text{-a.s.}
\end{equation}
Defining a non-equivalent probability measure $\mathbb{P}$ by
\begin{equation} \label{rrdQdP}
\frac{\mathrm{d}\mathbb{P}}{\mathrm{d}\mathbb{P}^0}=Y(T),
\end{equation}
Theorem 1 of \cite{rr13} then shows that under $\mathbb{P}$, the market satisfies NUPBR but not NFLVR; namely, any predictable $h(\cdot)$ as in \eqref{rr2} is an arbitrage under $\mathbb{P}$.
This ostensibly overly artificial construction is systematic in the following way: the authors note that in \emph{any} NUPBR-market with measure $\mathbb{P}$ in which NFLVR fails, one has the existence of a probability measure $\mathbb{P}^0$ and a $\mathbb{P}^0$-martingale $Y(\cdot)$ satisfying \eqref{rr1} such that \eqref{rrdQdP} holds. Thus the above assumptions are not only sufficient, but also necessary for the type of market considered.
Ruf and Runggaldier finish their short paper with several examples of explicit applications of their construction, demonstrating how markets with NUPBR but without NFLVR can be created in some simple cases. 

In their article \cite{ct13}, Chau and Tankov make an identical measure change to \cite{rr13}, but focus more on characterising the \emph{optimal} arbitrage, in the sense of largest guaranteed riskless profit, and consider Radon-Nikodyn derivatives which translate an investor's beliefs. Namely, they assume a pre-model with measure $\mathbb{P}^0$ and assumptions identical to Ruf and Runggaldier, as well as the existence of a $\mathbb{P}^0$-martingale $Y(\cdot)$ satisfying \eqref{rr1} and \eqref{rr2}.\footnote{The authors of \cite{ct13} equivalently formulate \eqref{rr2} as $\sup_{\mathbb{Q}\in\mathcal{M}^e(X)} \mathbb{E}^\mathbb{Q}[\mathbf{1}_{\{Y(T)>0\}}]<1$, to wit, the minimal super-replication price of the claim $\mathbf{1}_{\{Y(T)>0\}}$ is assumed less than 1.} Their contribution is noting that the $\mathbb{P}^0$-superhedging strategy $\hat{h}(\cdot)$ of the claim $\mathbf{1}_{\{Y(T)>0\}}$ is the optimal arbitrage in the $\mathbb{P}$-market; in other words, it guarantees the highest possible lower bound on riskless profit:\footnote{Note the similarity to \eqref{fernkrr}, where the optimal \emph{relative} arbitrage guarantees the highest return \emph{relative} to the market.}
\begin{equation} \label{ctdef}
V^{\hat{h}}(T)\geq \sup\Big\{ c>0 : \exists h(\cdot)\in\mathcal{A}_1 \text{ s.t. } V^h(T)\geq c\quad \mathbb{P}\text{-a.s.}  \Big\} \quad \mathbb{P}\text{-a.s.}
\end{equation}
Finally, the authors propose and study an explicit density 
\begin{equation}\label{ctdens}
\frac{\mathrm{d}\mathbb{P}}{\mathrm{d}\mathbb{P}^0}\bigg|_{\mathcal{F}_t}=Y(t):= \frac{\mathbb{P}^0(\sigma>T|\mathcal{F}_t)}{\mathbb{P}^0(\sigma>T)},
\end{equation}
for $\sigma$ some stopping time, which has the economic interpretation of translating an investor's belief that the event described by $\sigma$ will not happen before time $T$; definition \eqref{ctdens} implies that $\mathbb{P}(\sigma>T)=1$. Chau and Tankov then compute what these measure changes would look like in many explicit examples, both with continuous and discontinuous processes. In a toy example, where the pre-model dynamics are given by $X_1(t)=1+W^{\mathbb{P}^0}(t)$, with $W^{\mathbb{P}^0}(t)$ a one-dimensional $\mathbb{P}^0$-Brownian motion, and the investor believes $\sigma=\inf\{t>0:X_1(t)\leq0\}$, they are able to derive the resulting price dynamics under $\mathbb{P}$ as well as the optimal arbitrage (which is shown to be \emph{fragile} in the sense of \cite{gr11}). However, this toy example is extremely simple, and such explicit results cannot be expected in more complicated models.
Concluding, the results in \cite{ct13} seem too specific to be applicable in SPT. On an unrelated note, it might be interesting to investigate whether one can say anything sensible about the economic equilibrium which arises, in the sense of a dynamic stochastic general equilibrium (DSGE), when agents have diverse beliefs of the form \eqref{ctdens} --- see also \cite{br12}.

A more direct way of constructing markets with desired properties is taken in \cite{sar14}, where the author introduces and studies an explicit diverse market model, with the property that smaller stocks have larger drifts than large-cap stocks, as is the case in VSMs.

\section{Own research so far} \label{secme}
\subsection{Data study} \label{data}
In order to get an idea of how well the portfolios proposed by SPT perform in real markets, we have implemented them using historical market data.
For this empirical study we use 1761 trading days of daily market capitalisation data (computed using end-of-day stock prices and numbers of outstanding shares), starting on 1 January 2007 and ending on 31 December 2013, tracking a subset of 390 of the firms that were in the S\&P 500 index \emph{on the starting date of the data set}. We obtained this data from the CRSP data set\footnote{The CRSP data set is obtainable from \url{http://wrds-web.wharton.upenn.edu/wrds/}}; we track companies from an initial date so as to get a realistic and unbiased investment simulation. For instance, three firms (namely LEH, DYN and CIT) go bankrupt, and several firms merge or are acquired. For the sake of simplicity, we treat M\&A's by setting the discontinued stock's value constant and equal to its last recorded value. We are yet to incorporate dividends, which we expect to only slightly influence our results below (to our disadvantage, since relatively to the market our portfolios invest less in large firms, which tend to pay more dividends).

We use \verb;R; to programme our simulations --- the code is available from the author upon request. We do not display all that we investigated here, but the main directions of study in this discrete-time setting were:
\begin{itemize}
\item implementing and comparing the diversity-weighted, entropy-weighted, equally-weigh-ted and market portfolios, with different parameters;
\item studying the impact of proportional transaction costs on the relative performance of portfolios that rebalance;
\item experimenting with different rebalancing rules (such as total variance or relative entropy with respect to the target portfolio) so as to minimise transaction costs;
\item visualising the market-relative performance of portfolios \eqref{eqmaster};
\item computing the quantities in \eqref{heypal}, thus better understanding the information theoretic approach developed in \cite{palw13}.
\end{itemize}

An example of the first is visualised in Figure \ref{fig1}; the latter has provided us with a new insight into the reason behind the empirically observed good performance of the EWP --- see the end of Section \ref{futinfo}.

\begin{figure} [htp]
\centering
\includegraphics[width=.9\textwidth]{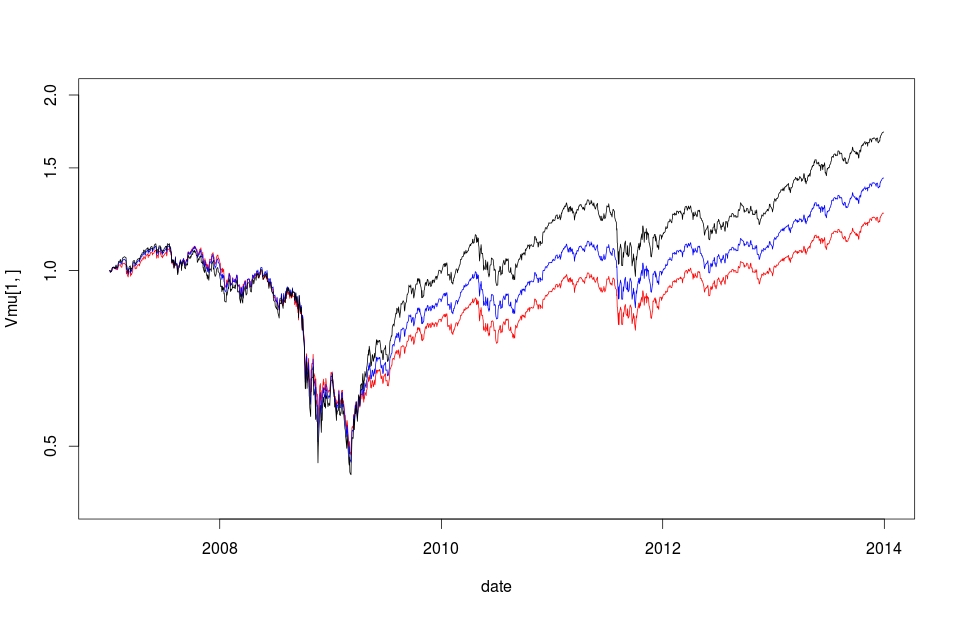}
\caption{Frictionless evolution of wealth processes corresponding to $\pi(\cdot)$ with \textcolor{red}{$p=1$} (market portfolio), \textcolor{blue}{$p=0.5$}, and $p=0$ (EWP).}  \label{fig1}
\end{figure}

\subsection{Diversity-weighted portfolio with negative $p$} \label{DWPneg}
Recall the diversity-weighted portfolio (DWP) $\pi(\cdot)$, defined in \cite{fkk05} as
\begin{equation}
\pi_i(t):= \frac{(\mu_i(t))^p}{\sum_j (\mu_j(t))^p}\qquad i=1,\ldots,n.
\end{equation}
This long-only portfolio is generated by $\mathbf{G}_p:x\mapsto \left(\sum_{i=1}^n x_i^p\right)^{1/p}$. Usually, $p\in(0,1)$, and one can show that in a diverse nondegenerate market, $\pi(\cdot)$ beats the market over sufficiently long time horizons. The portfolio $\pi(\cdot)$ with parameter $p=0$ corresponds to the equally-weighted portfolio (EWP), $\pi_i(\cdot)=1/n$.

\subsubsection{Observations from data} \label{secobs}
One can observe from data that the equally-weighted portfolio outperforms diversity-weigh-ted portfolios with $p\in(0,1)$, which in turn outperform the market --- see Figure \ref{fig1}. More precisely, the performance of $\pi(\cdot)$ seems to increase monotonically as $p$ decreases (and as $\pi(\cdot)$ moves away from the market portfolio, increasing the weights put on the lowest ranked stocks). This motivated us to consider diversity-weighted portfolios with negative parameter $p$; it turns out that such portfolios, for $p$ not too negative, would have performed extremely well over the past ten years, and that this remains true if one includes realistic proportional transaction costs --- see Figure \ref{fig3}.

\begin{figure} [htp]
\centering
\includegraphics[width=.9\textwidth]{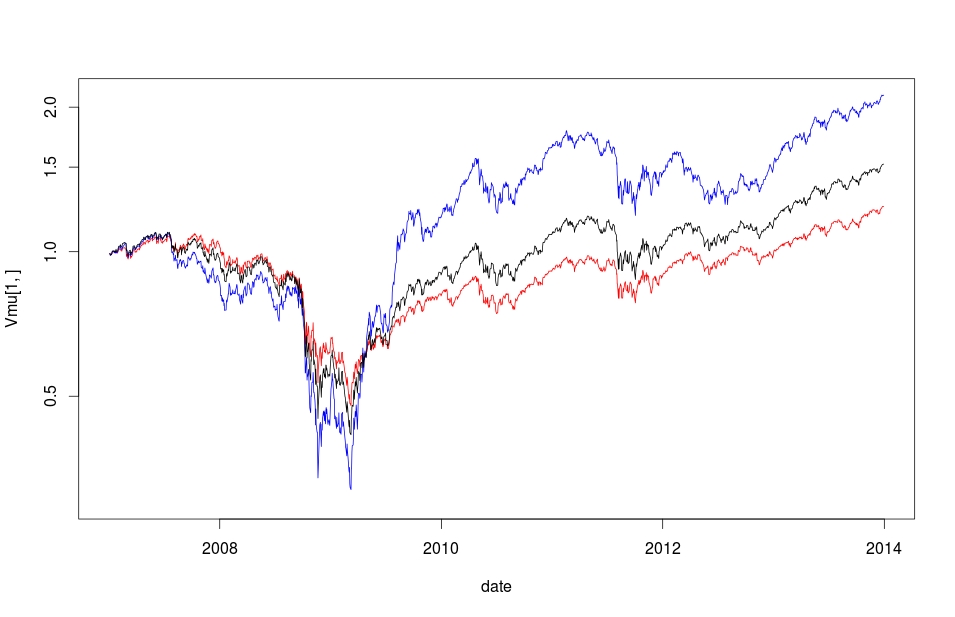}
\caption{Evolution of wealth processes corresponding to $\pi(\cdot)$ with \textcolor{red}{$p=1$} (market portfolio), \textcolor{blue}{$p=-1$}, and $p=0$ (EWP), all with the adjustment that $\pi_i(t)=0$ whenever $\mu_i(t)=0$. We assume proportional transaction costs of 1\% of the value traded, and rebalance once every 20 trading days.} \label{fig3} 
\end{figure}

For our empirical studies we use 1761 trading days of daily market capitalisation data, starting on 1 January 2007 and ending on 31 December 2013, tracking a subset of 390 of the firms that were in the S\&P 500 index \emph{on the starting date of the data set}. We obtained this data from the CRSP data set; we track companies from an initial date so as to get a realistic and unbiased investment simulation. For instance, three firms (namely LEH, DYN and CIT) go bankrupt, and several firms merge or are acquired.
See the Appendix for an alternative visualisation using this data.


\subsubsection{Theoretical motivation} \label{hello}
We ask ourselves: why does $\pi(\cdot)$ with $p<0$ perform so well? One can check that diversity no longer implies that this portfolio is a RA wrt $\mu(\cdot)$, as it would if $p\in(0,1)$. However, let us assume non-degeneracy \eqref{ND} 
as well as the stronger ``non-failure assumption'':\footnote{We assume, WLOG, that $\delta<1/n.$}
\begin{equation} \label{nofail}
\exists\, {\delta}>0\ \text{ s.t. }\ \mu_{(n)}(t)\geq {\delta} \quad \forall t\geq0\quad \mathbb{P}\text{-a.s.}
\end{equation}
This is assumption (21) in Theorem 2.11 of \cite{palw13}, and implies diversity with parameter $(n-1)\delta$. We will restrict the values that $p$ can take later on.

Note that, for $p<0$, we have the bounds \eqref{gpbound}:
\begin{equation} \label{gpboundd}
n^{1-p}=\sum_{i=1}^n \Big(\frac{1}{n}\Big)^p \leq \sum_{i=1}^n \mu^p_i(t)=(\mathbf{G}_p(\mu(t))^p\leq (n-1)\delta^p + (1-(n-1)\delta)^p < n\delta^p; 
\end{equation}
hence we have the bound
\begin{equation} \label{boubound}
\log \bigg(\frac{\mathbf{G}_p(\mu(T))}{\mathbf{G}_p(\mu(0))}\bigg)> \log (n\delta),
\end{equation}
a negative number.\footnote{Note that this lower bound is necessarily negative for any non-constant generating function.}

When $p\in(0,1)$, one has that $$\pi_{(1)}(t)=\frac{\mu^p_{(1)}(t)}{\sum_i\mu^p_i(t)}\leq \mu_1(t),$$ since in that case $x\mapsto x^p$ is an \emph{increasing} function. However, for $p<0$, $x\mapsto x^p$ is a \emph{decreasing} function, and the above inequality no longer holds (hence the usual proof fails). However, under assumption \eqref{nofail} and using \eqref{gpboundd}, we get that 
\begin{equation} \label{ineq}
\pi_{(1)}(t)=\frac{\mu^p_{(n)}(t)}{\sum_i\mu^p_i(t)}\leq \frac{\delta^p}{n^{1-p}}=\frac{(n\delta)^p}{n}<1
\end{equation}
if we make the additional assumption that $p\in(\log n/\log (n\delta),0)$.
Non-degeneracy, by Lemma 3.4 of \cite{fk09} (originally proved in the Appendix of \cite{fkk05} --- see Lemma \ref{lem3}), together with \eqref{ineq}, implies that
\begin{equation} \label{lowerbd}
\int_0^T \gamma^*_\pi(t)\mathrm{d}t\geq \frac{\varepsilon}{2}\int_0^T(1-\pi_{(1)}(t))\mathrm{d}t\geq \frac{\varepsilon}{2}T(1-(n\delta)^p/n).
\end{equation}

Hence, using Fernholz's master equation \eqref{eqmaster}, we conclude that 
\begin{align} \label{resultone}
\log\bigg( \frac{V^\pi(T)}{V^\mu(T)}\bigg)&= \log \bigg(\frac{\mathbf{G}_p(\mu(T))}{\mathbf{G}_p(\mu(0))}\bigg) + (1-p)\int_0^T \gamma^*_\pi(t)\mathrm{d}t\nonumber\\
&>\log (n\delta) + (1-p)\frac{\varepsilon}{2}T(1-(n\delta)^p/n) \geq 0\quad \text{ a.s.},\nonumber\\ & \text{ provided }\ T\geq T^*_\delta:=\frac{-2n \log (n\delta)}{\varepsilon (1-p)(n-(n\delta)^p)};
\end{align}
i.e., the portfolio $\pi(\cdot)$ beats the market over sufficiently long time horizons. \qed

\subsubsection{Outperforming `normal' DWP} \label{outp}
Under one additional assumption besides \eqref{nofail} and non-degeneracy, namely that of bounded variance \eqref{BV}, one can show that DWPs with negative parameter $p$ can be relative arbitrages with respect to DWPs with $p>0$. 

Take $p^-\in(\log n/\log n\delta,0)$ and $p^+\in(0,1)$, and let $\pi^-(\cdot)$ and $\pi^+(\cdot)$ be the diversity-weighted portfolios with parameters $p^-$ and $p^+$, respectively. Note that 
\begin{equation} \label{onex}
1\leq \big(\mathbf{G}_{p^+}(\mu(t))\big)^{p^+} \leq n^{1-p^+}.
\end{equation}
By Lemma 3.4 of \cite{fk09} (an analog of Lemma \ref{lem3}), bounded variance implies the a.s.\! inequality
\begin{equation} \label{twox}
\gamma^*_{\pi^+}(t)\leq 2K(1-\pi^+_{(1)}(t))\quad \forall t\geq0.
\end{equation}
Non-degeneracy implies \eqref{lowerbd}, which, together with \eqref{eqmaster}, \eqref{gpboundd}, \eqref{onex}, \eqref{twox}, and the observation that $\pi^+_{(1)}(t)\geq 1/n$, gives
\begin{align} \label{long}
\log \bigg(\frac{V^{\pi^-}(T)}{V^{\pi^+}(T)}\bigg)&= \log \bigg(\frac{V^{\pi^-}(T)}{V^{\mu}(T)}\bigg)-\log \bigg(\frac{V^{\pi^+}(T)}{V^{\mu}(T)}\bigg)\nonumber\\
&= \log\bigg(\frac{\mathbf{G}_{p^-}(\mu(T))\mathbf{G}_{p^+}(\mu(0))}{\mathbf{G}_{p^-}(\mu(0))\mathbf{G}_{p^+}(\mu(T))}\bigg)\nonumber\\&\quad + \int_0^T \Big((1-p^-)\gamma^*_{\pi^-}(t)-(1-p^+)\gamma^*_{\pi^+}(t)  \Big)\mathrm{d}t\nonumber\\
&> \log \bigg(\frac{n^{1/p^-}\delta\cdot 1}{n^{(1-p^-)/p^-}\cdot n^{(1-p^+)/p^+}} \bigg) \nonumber\\ 
&\quad + \Big( \frac{\varepsilon}{2}(1-(n\delta)^{p^-}/n)(1-p^-)-2K\frac{n-1}{n}(1-p^+) \Big)T\nonumber\\
&=\log \delta n^{2-1/p^+} + C\cdot T >0\quad \text{ a.s.,}
\end{align}
provided that
\begin{equation} \label{eqCT}
CT> -\log \Big(\delta n^{2-1/p^+}\Big),
\end{equation}
a positive number (since $\delta<1/n$). Here, we have used the notation
\begin{equation*}
C:=C(p^-,p^+,\varepsilon,K,\delta,n):=\frac{\varepsilon}{2}(1-(n\delta)^{p^-}/n)(1-p^-)-2K\frac{n-1}{n}(1-p^+).
\end{equation*}

Now in the case that $C\leq0$, \eqref{long} is never positive. However if $C>0$, \eqref{eqCT} gives the condition 
\begin{equation}
T>\frac{-\log \Big(\delta n^{2-1/p^+}\Big)}{C};
\end{equation}
hence in this case $\pi^-(\cdot)$ almost surely outperforms $\pi^+(\cdot)$ over sufficiently long time horizons. Now $C>0$ is possible if $p^+$ is `not too small', i.e.~if $\pi^+$ is close to the market portfolio. Namely, $C>0$ if and only if
$$p^+>1-\frac{\varepsilon (n-(n\delta)^{p^-})(1-p^-)}{4K(n-1)},$$
which is a number less than 1. 



\subsubsection{Under-performing `normal' DWP?}
One might wonder whether the construction in section \ref{outp} could work the other way around, i.e.~to show that $\pi^+(\cdot)$ outperforms $\pi^-(\cdot)$ over certain time horizons. We clearly do not want this to be true, as we would arrive at a contradiction, and we show this explicitly for the sake of completeness.

Note that, under the non-degeneracy and bounded variance assumptions, \eqref{lowerbd} and \eqref{twox} hold for any long-only portfolio. Applied to $\pi^+(\cdot)$, \eqref{lowerbd}, together with the earlier observation that \eqref{nofail} implies diversity with parameter $(n-1)\delta$ and the fact that $\pi^+_{(1)}(t)\leq \mu_{(1)}(t)$, leads to 
\begin{equation} \label{diversek}
\int_0^T \gamma^*_{\pi^+}(t)\mathrm{d}t\geq \frac{\varepsilon}{2}\int_0^T(1-\pi^+_{(1)}(t))\mathrm{d}t\geq \frac{\varepsilon}{2}T(n-1)\delta.
\end{equation}
Hence we get that $\exists T>0$ s.t.
\begin{align}
\log \bigg(\frac{V^{\pi^+}(T)}{V^{\pi^-}(T)}\bigg)&= \log\bigg(\frac{\mathbf{G}_{p^+}(\mu(T))\mathbf{G}_{p^-}(\mu(0))}{\mathbf{G}_{p^+}(\mu(0))\mathbf{G}_{p^-}(\mu(T))}\bigg)\nonumber\\&\quad + \int_0^T \Big((1-p^+)\gamma^*_{\pi^+}(t)-(1-p^-)\gamma^*_{\pi^-}(t)  \Big)\mathrm{d}t\nonumber\\
&> \log \Big(\delta n^{2-1/p^+}\Big) + \Big( \frac{\varepsilon}{2}(n-1)\delta(1-p^+)-2K\frac{n-1}{n}(1-p^-)\Big)\cdot T\nonumber\\
&>0\quad \text{ a.s.}
\end{align}
if and only if (since $\log \delta n^{2-1/p^+}\!<0$)
\begin{equation*}
\left.\begin{aligned}
&\frac{\varepsilon}{2}(n-1)\delta(1-p^+)-2K\frac{n-1}{n}(1-p^-)>0\ \\
&\iff\ K<\varepsilon\frac{(1-p^+)\delta}{4(1-p^-)/n}\\
&\text{but: }\ 1-p^+<1-p^-\ \ \text{ and }\ \ \delta<1/n 
\end{aligned} \right\}
 \quad \Rightarrow\ K<\varepsilon.
\end{equation*}
We have arrived at a contradiction, showing that the above lower bound for\\ $\log V^{\pi^+}(T)/V^{\pi^-}(T)$ is never positive.\qed

\subsubsection{Weakening of non-failure assumption} \label{remove}
We can make a slight adjustment to the DWP with negative $p$, so as to make it a relative arbitrage even in markets where the non-failure condition \eqref{nofail} only holds over the time horizon $[0,T^*_\delta]$, with $T^*_\delta$ the minimal horizon required for this portfolio to beat the market as defined in \eqref{resultone}. Namely, let $\pi(\cdot)$ be the diversity-weighted portfolio with $p\in(\log n/\log n\delta,0)$, fix a threshold value $\delta>0$ and for $t\geq0$ define the portfolio
\begin{equation}\label{pihat}
\hat{\pi}(t):= \begin{cases}
\pi(t) & t<\tau_\delta\\
\mu(t) & t \geq \tau_\delta,
\end{cases}
\end{equation}
where the stopping time $\tau_\delta$ is defined as the first time that condition \eqref{nofail} fails:
\begin{equation}
\tau_\delta:=\inf\{ t>0\ |\ \mu_{(n)}(t)= \delta\}.
\end{equation}
It is easy to see that $\hat{\pi}(\cdot)$ is predictable, since both $\pi(\cdot)$ and $\mu(\cdot)$ are predictable, as well as the event $\{ t \geq \tau_\delta \}$; the latter follows from the fact that we assume the price processes to be adapted and continuous. Note that $\tau_\delta>T^*_\delta$ by assumption.

Now for all $T^*_\delta\leq t<\tau_\delta$, our result \eqref{resultone} that $\pi(\cdot)$ beats the market holds, namely by a linearly increasing log-factor of 
\begin{equation}
f_\delta(t):= -\log (n\delta) + (1-p)\frac{\varepsilon t}{2}\Big(1-\frac{(n\delta)^p}{n}\Big)>0;
\end{equation}
i.e.
\begin{equation}
V^\pi(t)>\exp(f_\delta(t))\cdot V^\mu(t)\quad \forall t\in[T^*_\delta,\tau_\delta)\quad \mathbb{P}\text{-a.s.}
\end{equation}
We therefore have that 
\begin{equation*}
\left.\begin{aligned}
&V^{\hat{\pi}}(t)=V^\pi(t)>V^\mu(t) \quad \forall t\in[T^*_\delta,\tau_\delta) \\
&V^{\hat{\pi}}(\tau_\delta)=V^\pi(\tau_\delta)>V^\mu(\tau_\delta) \\
&V^{\hat{\pi}}(t)=\frac{V^\pi(\tau_\delta)}{V^\mu(\tau_\delta)}\cdot V^\mu(t) > V^\mu(t) \quad \forall t>\tau_\delta\
\end{aligned} \right\}
\ V^{\hat{\pi}}(t) > V^\mu(t)\quad  \forall t\geq T^*_\delta\quad \mathbb{P}\text{-a.s.}
\end{equation*}
and hence $\hat{\pi}(\cdot)$ as defined in \eqref{pihat} is a relative arbitrage with respect to the market portfolio $\mu(\cdot)$ over all time horizons $[0,T]$ with $T\geq T^*_\delta$.

Note also that the above portfolio is still (stochastically) functionally generated, in the generalised sense of \cite{strong13}, by the function $$\mathbf{G}_p(x) \mathbbm{1}_{\{ \tau_\delta > t \}} + \mathbbm{1}_{\{ \tau_\delta \leq  t \}}.$$ A similar construction was used in \cite{bf08} to construct a short-term relative arbitrage in a model class including VSMs. 

An investor who does not want to return to the low-risk market portfolio after stopping time $\tau_\delta$ could do either of the following to hopefully increase her returns: first, she could sell all of her holdings in the lowest-ranked stock when its market weight falls to $\mu_{(n)}(\tau_\delta)=\delta$, investing only in the $n-1$ remaining stocks after time $\tau_\delta$, and so on once another market weight falls to the investor's threshold value $\delta$. Another possibility is to adjust $\hat{\pi}(\cdot)$ as follows:
\begin{equation}\label{pihat2}
\tilde{\pi}(t):= \begin{dcases}
\pi(t) & t<\tau_\delta\\
\frac{V^\mu(\tau_\delta)}{V^\pi(\tau_\delta)} \mu(t) + \frac{V^\pi(\tau_\delta)-V^\mu(\tau_\delta)}{V^\pi(\tau_\delta)} \pi(t)& t \geq \tau_\delta,
\end{dcases}
\end{equation}
This way, we still have $V^{\tilde{\pi}}(t)>V^\mu(t)\ \forall t\in[T^*_\delta,\tau_\delta]\ \mathbb{P}$-a.s., and also 
\begin{align*}
&V^{\tilde{\pi}}(t)=V^\mu(t)+\frac{V^\pi(\tau_\delta)-V^\mu(\tau_\delta)}{V^\pi(\tau_\delta)}V^\pi(t)\geq V^\mu(t) \quad \forall t>\tau_\delta\quad\mathbb{P}\text{-a.s.},\\
&\mathbb{P}(V^{\tilde{\pi}}(t) > V^\mu(t))>0.
\end{align*}
That is, the portfolio \eqref{pihat2} puts only just enough funds aside to ensure beating the market after time $\tau_\delta$, and continues to invest the remaining gains due to investment in $\pi(\cdot)$ up to time $\tau_\delta$ in the DWP with $p<0$, giving potential extra gains if the market does not crash after time $\tau_\delta$ (which for $\tilde{\pi}(\cdot)$ to be a relative arbitrage needs to be assumed to have non-zero probability). If the market does crash, $V^\pi(\cdot)$ will go to zero (see the Discussion in Section \ref{disco!}), and we are left with $V^\mu(\cdot)$, thus matching the market.

\subsubsection{Attempts at removing the non-failure assumption} \label{remove?}
It would be nice to further weaken or even remove entirely the assumption of non-failure, as defined in \eqref{nofail}. Re-examining our first proof in Section \ref{hello}, we see that it works thanks to the non-failure assumption in two ways:
\begin{enumerate}
\item Non-failure gives a lower bound to $\mathbf{G}_p(\mu(\cdot))$, namely the second inequality in \eqref{gpboundd} (the first inequality always holds):
\begin{equation} \label{secondbound}
n^{1-p} \leq (\mathbf{G}_p(\mu(t))^p\ \textcolor{red}{<}\ n\delta^p.
\end{equation}
Without \eqref{nofail}, the function $\mathbf{G}_p(\mu(\cdot))$ is only bounded below by $0$.
\item Condition \eqref{nofail} ensures that $\mu^p_{(n)}(t)\leq \delta^p$, and hence that
\begin{equation} \label{secondpi}
\pi_{(1)}(t)\leq\frac{(n\delta)^p}{n}<1\quad \forall t\geq0.
\end{equation}
\end{enumerate}
Both of these are crucial for the proof; 1.\! for ensuring a lower bound on the log-ratio of $\mathbf{G}_p(\mu(\cdot))$ at different times (i.e.~\eqref{boubound}), and 2.\! for proving a lower bound on the finite-variation term in the master equation (i.e.~\eqref{lowerbd}): namely, non-degeneracy implies, by Lemma 3.4 of \cite{fk09}, that $\gamma^*_\pi(t)\geq\frac{\varepsilon}{2}(1-\pi_{(1)})$, which with \eqref{secondpi} (or the weaker assumption of diversity) gives a positive lower bound on $\gamma^*_\pi(t)$.

Keeping the above in mind, it is tempting to `mix' the diversity-weighted portfolio $\pi(\cdot)$ with another portfolio (i.e.~take certain proportions in $\pi(\cdot)$ and, for instance, the market) so as to ensure \eqref{secondbound} without \eqref{nofail}. One wishes to do this in such a way as to keep the lower bound on the finite-variation term in the master equation (denoted in \cite{fk09} by $\mathfrak{g}(\cdot)$); see Table \ref{table1} for some suggestions for mixes, one should interpret it as $\mathbf{G}^{[1]}$ being $\mathbf{G}_p$, and the other $\mathbf{G}^{[i]}$ representing the generating functions of for instance the entropy-weighted, market and equally-weighted portfolios.

\begin{table}[htp] 
\centering
\begin{tabular}{| c | c | c |} 
\hline
$\hat{\mathbf{G}}= $ & $\hat{\pi}(\cdot)=$ & $\hat{\mathfrak{g}}(\cdot)=$\\
\hline \hline
$a+b\mathbf{G}^{[1]}$ & $\frac{\mathbf{G}^{[1]}(\cdot)\pi^{[1]}(\cdot) + a\mu(\cdot)}{\mathbf{G}^{[1]}(\cdot)+a}$ & $\frac{\mathbf{G}^{[1]}(\cdot)}{\mathbf{G}^{[1]}(\cdot)+a}\mathfrak{g}^{[1]}(\cdot)$\\
\hline
$\big(\mathbf{G}^{[1]}\big)^q$ & $q\pi^{[1]}(\cdot) + (1-q)\mu(\cdot)$ & $\frac{q(q-1)}{\mathbf{G}^{[1]}(\cdot)}\mathfrak{g}^{[1]}(\cdot)$  \\
\hline
$\exp(\mathbf{G}^{[1]})$ & $\mathbf{G}^{[1]}(\cdot)(\pi^{[1]}(\cdot) + (1-\mathbf{G}^{[1]}(\cdot))\mu(\cdot)$ & $\mathbf{G}^{[1]}(\cdot) \mathfrak{g}^{[1]}(\cdot)$  \\
\hline
$\prod_i \mathbf{G}^{[i]}$ & $\sum_i \pi^{[i]}(\cdot)-(n-1)\mu(\cdot)$ & \ldots \\
\hline
$\sum_i \mathbf{G}^{[i]}$ & $\frac{\sum_i \mathbf{G}^{[i]}(\cdot)\pi^{[i]}(\cdot)}{\sum_i \mathbf{G}^{[i]}(\cdot)}$ &  $\frac{\sum_i \mathbf{G}^{[i]}(\cdot)\mathfrak{g}^{[i]}(\cdot)}{\sum_i \mathbf{G}^{[i]}(\cdot)}$\\
\hline
\end{tabular}
\caption{Let $\pi^{[i]}(\cdot)$ be generated by $\mathbf{G}^{[i]}$, $\ i=1,\ldots,n$, and let $\hat{\pi}(\cdot)$ be generated by $\hat{\mathbf{G}}$. The above shows the relation between these portfolios for different relations between the generating functions; we write $\mathbf{G}^{[i]}(\cdot):=\mathbf{G}^{[i]}(\mu(\cdot))$.} \label{table1}
\end{table}

Our progress thus far with finding `good' mixes is limited to the following: let $\hat{\mathbf{G}}:=\mathbf{G}_{p^+}+\mathbf{G}_{p^-}$, with $p^+\in(0,1)$ and $p^-<0$ as before in Section \ref{outp} (although $p^-$ may now take any negative value). Then by the above, $\hat{\mathbf{G}}$ generates the portfolio
\begin{equation} \label{crazypi}
 \hat{\pi}(t):=\mathfrak{p}(t) \pi^+(t) +  (1-\mathfrak{p}(t)) \pi^-(t),
\end{equation}
where the time-dependent proportion is given as a deterministic function of the market weights:\footnote{We define $\mathbf{G}_{p^-}(\bar{x}):=\lim_{x\to\bar{x}}\mathbf{G}_{p^-}(x)=0$ for any $\bar{x}$ in the simplex with $\prod_i \bar{x}_i=0$.}
\begin{equation}
\mathfrak{p}(t):=\frac{\mathbf{G}_{p^+}(\mu(t))}{\mathbf{G}_{p^+}(\mu(t))+\mathbf{G}_{p^-}(\mu(t))}\in(0,1].
\end{equation}
Also by the above table, we have that the drift process is
\begin{align}
 \hat{\mathfrak{g}}(t)&=\mathfrak{p}(t) \mathfrak{g}^+(t) +  (1-\mathfrak{p}(t)) \mathfrak{g}^-(t)\nonumber\\
&= (1-p^+) \mathfrak{p}(t) \gamma^*_{\pi^+}(t)+ (1-p^-)(1- \mathfrak{p}(t)) \gamma^*_{\pi^-}(t).
\end{align}
If we now assume non-degeneracy and diversity over the horizon $[0,T]$ with parameter $\delta$, then we have as in \eqref{diversek} that 
$$\gamma^*_{\pi^+}(t)\geq \frac{\varepsilon}{2}(1-\pi^+_{(1)}(t))\geq \frac{\varepsilon}{2}(1-\mu_{(1)}(t))\geq \frac{\varepsilon}{2} \delta.$$
Since also $\gamma^*_{\pi^-}(t)\geq0$, as for any portfolio, noting that 
$$\mathfrak{p}(t)\geq \frac{1}{1+n^{1/p^--1}}>0,$$
and using the bounds \eqref{gpboundd} and \eqref{onex}, we conclude by the master equation that
\begin{align}
\log \bigg(\frac{V^{\hat{\pi}}(T)}{V^\mu(T)}\bigg)&= \log\left( \frac{\mathbf{G}_{p^+}(\mu(T))+\mathbf{G}_{p^-}(\mu(T))}{\mathbf{G}_{p^+}(\mu(0))+\mathbf{G}_{p^-}(\mu(0))} \right) + (1-p^+)\int_0^T \mathfrak{p}(t) \gamma^*_{\pi^+}(t)\mathrm{d}t\nonumber\\
&\qquad +  (1-p^-)\int_0^T(1- \mathfrak{p}(t)) \gamma^*_{\pi^-}(t)\mathrm{d}t\nonumber\\
&\geq - \log\left( n^{1/p^+-1}+n^{1/p^--1} \right) + \frac{\varepsilon}{2}T(1-p^+)\frac{1-\delta}{1+n^{1/p^--1}}>0\quad \text{ a.s.},\nonumber\\ 
& \text{ provided }\ T > \frac{2(1+n^{1/p^--1}) \log\left( n^{1/p^+-1}+n^{1/p^--1} \right)}{\varepsilon (1-p^+)(1-\delta)}.
\end{align}
Hence the portfolio defined in \eqref{crazypi} beats the market over sufficiently long time horizons under the assumptions of non-degeneracy and (weak) diversity, a property that the portfolio $\pi^+(\cdot)$ also has on its own.

At the end of the next section we make another attempt at removing the non-failure assumption.

\subsubsection{Rank-based diversity-weighted portfolios} \label{secrankedDWP}
\emph{This section contains errors, apologies. The assumed bounds on semimartingale local time do not hold.}\\
In Example 4.2 in \cite{f01} (see also Remark 11.9 in \cite{fk09}), Fernholz considers a variation of the diversity-weighted portfolio with parameter $r$ which only invests in the $m<n$ highest-ranked stocks (by capitalisation), namely:
\begin{equation} \label{largeDWPp}
\mu^\#_{p_t(k)}(t)=\begin{dcases} \frac{\big(\mu_{(k)}(t)\big)^r}{\sum_{l=1}^m \big(\mu_{(l)}(t)\big)^r}, &k=1,\ldots,m\\
0, &k=m+1,\ldots,n,
\end{dcases}
\end{equation}
with $p_t(k)$ the index of the stock that is ranked $k$th at time $t$, so that $\mu_{p_t(k)}(t)=\mu_{(k)}(t)$. Portfolio \eqref{largeDWPp} is generated by the function $\mathcal{G}_r(x)=\big(\sum_{l=1}^m (x_{(l)})^r \big)^{1/r}$, analogous to $\mathbf{G}_p$ above --- see also Section \ref{sec:rankport}. The master equation \eqref{wow} applied to \eqref{largeDWPp} gives \eqref{largeDWPrr}:
\begin{equation} \label{prettyimportant}
\log \bigg(\frac{V^{\mu^\#}\!(T)}{V^\mu(T)}\bigg) = \log \bigg(\frac{\mathcal{G}_r(\mu(T))}{\mathcal{G}_r(\mu(0))} \bigg) + (1-r)\int_0^T\gamma^*_{\mu^\#}(t)\mathrm{d}t - \int_0^T \frac{\mu^\#_{(m)}(t)}{2}\mathrm{d}\mathfrak{L}^{m,m+1}(t).
\end{equation}
Here, $\mathfrak{L}^{k,k+1}(t)\equiv \Lambda_{\Xi_k}(t)$ is, as defined in equation \eqref{firstdeflocalt}, the semimartingale local time at the origin accumulated by the nonnegative process 
\begin{equation}
\Xi_k(t):=\log\big(\mu_{(k)}(t)/\mu_{(k+1)}(t)\big),\quad t\geq0.
\end{equation}

\subsubsection*{Case 1: $\mathbf{r\in(0,1)}$}
Assume non-degeneracy and diversity with parameter $\delta>0$, and let $r\in(0,1)$. Then straightforward modifications of \eqref{onex} and \eqref{diversek}, together with the observations that\\ $\mu^\#_{(m)}(t)\leq 1/m$ and $\mathfrak{L}^{m,m+1}(T)\leq T$, imply by \eqref{prettyimportant} that 
\begin{align}
\log \bigg(\frac{V^{\mu^\#}\!(T)}{V^\mu(T)}\bigg) &\geq -\frac{1-r}{r}\log n  + \frac{\varepsilon}{2}\delta (1-r)T - \int_0^T \frac{\mu^\#_{(m)}(t)}{2}\mathrm{d}\mathfrak{L}^{m,m+1}(t)\nonumber\\
&\geq -\frac{1-r}{r}\log n  + \Big[\frac{\varepsilon}{2}\delta (1-r) - \frac{1}{2m}\Big]T>0,
\end{align}
provided that $m$ can be chosen big enough (i.e.~we require a large market) such that $\frac{\varepsilon}{2}\delta (1-r) - \frac{1}{2m}>0$, and that $T$ is sufficiently large. Hence, under these conditions the large-stock DWP with $r\in(0,1)$ beats the market over long time horizons.

A similar result can be obtained for the \emph{small-stock} diversity-weighted portfolio with $r\in(0,1)$, namely the portfolio defined as
\begin{equation} \label{smallDWP}
\mu^\$_{p_t(k)}(t)=\begin{dcases} 0, &k=1,\ldots,m-1\\
\frac{\big(\mu_{(k)}(t)\big)^r}{\sum_{l=m}^n \big(\mu_{(l)}(t)\big)^r}, &k=m,\ldots,n.
\end{dcases}
\end{equation}
Again assume non-degeneracy, as well as
\begin{equation}\label{kappass}
\exists\, \kappa>0\ \text{ s.t. }\ \mu_{(m)}(t)\geq\kappa,\quad  \forall t\in[0,T],\quad \mathbb{P}\text{-a.s. }
\end{equation}
which says that no more than $n-m$ companies will go bankrupt by time $T$.\footnote{Note that $\kappa\in(0,1/m)$ and that \eqref{kappass} implies diversity with parameter $(m-1)\kappa$.} Now \eqref{prettyimportant} becomes 
\begin{equation} \label{prettyimportant2}
\log \bigg(\frac{V^{\mu^\$}\!(T)}{V^\mu(T)}\bigg) = \log \bigg(\frac{\mathcal{G}_r(\mu(T))}{\mathcal{G}_r(\mu(0))} \bigg)+ (1-r)\int_0^T\gamma^*_{\mu^\$}(t)\mathrm{d}t + \int_0^T \frac{\mu^\$_{(m)}(t)}{2}\mathrm{d}\mathfrak{L}^{m-1,m}(t).
\end{equation}
The equivalent of \eqref{onex} is
\begin{equation*} 
\kappa^r\leq \big(\mathcal{G}_r(\mu(t))\big)^r= \sum_{l=m}^n \big( \mu_{(l)} (t)\big)^r \leq (n-m+1)n^-r < n^{1-r},
\end{equation*}
which together with the fact that $\mathfrak{L}^{m-1,m}(T)\geq 0$ implies the following:
\begin{equation} \label{inspire}
\log \bigg(\frac{V^{\mu^\$}\!(T)}{V^\mu(T)}\bigg) > \frac{1}{r}\log \left(\frac{\kappa^r}{n^{1-r}}  \right) + \frac{\varepsilon}{2}(m-1)\kappa (1-r)T >0
\end{equation}
for $T$ sufficiently large; i.e., $\mu^\$(\cdot)$ is a relative arbitrage with respect to $\mu(\cdot)$ over long time horizons.

\subsubsection*{Case 2: $\mathbf{r<0}$}
When $r<0$, equations \eqref{prettyimportant} and \eqref{prettyimportant2} still hold for the large-stock and small-stock DWPs, respectively. Let us consider the large-stock portfolio \eqref{largeDWPp} with $r<0$ first: assume non-degeneracy \eqref{ND} and \eqref{kappass}, so that the following bounds hold:
\begin{equation} \label{theother}
m^{1-r} \leq \big(\mathcal{G}_r(\mu(t))\big)^r = \sum_{l=1}^m \big( \mu_{(l)} (t)\big)^r < m\kappa^r,
\end{equation}
and
\begin{equation} \label{theother2}
\gamma^*_{\mu^\#}(t)\geq\frac{\varepsilon}{2}\big(1-\mu^\#_{(1)}(t)\big)
\end{equation}
by non-degeneracy, as usual. Because of assumption \eqref{kappass}, we have
\begin{equation}
\mu^\#_{(1)}(t)=\frac{\mu^r_{(m)}(t)}{\sum_{l=1}^m \mu^r_{(l)}(t)}\leq \frac{\kappa^r}{m^{1-r}}<1,
\end{equation}
since necessarily $\kappa<1/m$, which together with \eqref{prettyimportant}, \eqref{theother} and \eqref{theother2} implies
\begin{align}
\log \bigg(\frac{V^{\mu^\#}\!(T)}{V^\mu(T)}\bigg)&> -\log (m\kappa)  + \Big[\frac{\varepsilon}{2}\Big(1-\frac{\kappa^r}{m^{1-r}}\Big) (1-r)-\frac{1}{2m}\Big] T \nonumber\\
&>-\log (m\kappa)  + \frac{1}{m}\Big[\frac{\varepsilon}{2}(m-1) (1-r)-\frac{1}{2}\Big] T>0,
\end{align}
provided that $m$ can be chosen big enough and that $T$ is sufficiently large, as in Case 1. Hence in this case the large-stock diversity-weighted portfolio with $r<0$ beats the market over sufficiently long time horizons. The assumption \eqref{kappass} that no more than $n-m$ companies will crash can be made reasonable by choosing large $n$, and $m$ not too close to $n$; it can be empirically observed that the number of companies that crash in a given time interval is typically quite small. Hence we expect the portfolio $\mu^\#(\cdot)$ with $r<0$ to perform well in real markets.

For the small-stock DWP \eqref{smallDWP} with $r<0$ one can see that a stronger assumption than \eqref{kappass} is required (besides non-degeneracy), namely our original non-failure assumption \eqref{nofail}. Under these assumptions, it can once again be shown, using \eqref{prettyimportant2}, that $\mu^\$(\cdot)$ with $r<0$ beats the market over sufficiently long time horizons (without any restrictions on $m$).

\begin{figure} [htp]
\centering
\includegraphics[width=.9\textwidth]{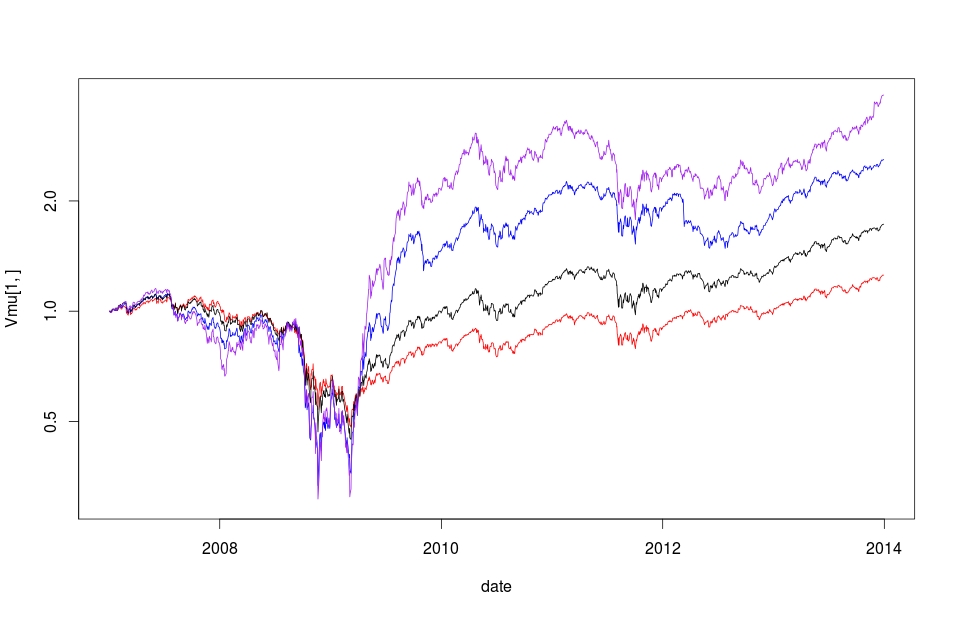}
\caption{Frictionless evolution of wealth processes corresponding to the large-stock portfolio \textcolor{Purple}{$\mu^\#(\cdot)$} with $r=-4$ and $m=n-10$, $\pi(\cdot)$ with \textcolor{red}{$p=1$} (market portfolio), $p=0$ (EWP), and \textcolor{blue}{$p=-1$}, the latter two with the adjustment that $\pi_i(t)=0$ whenever $\mu_i(t)=0$.} \label{largepic} 
\end{figure}

\paragraph{Implementation}
Using the same CRSP data set as described in Section \ref{secobs}, we implement the large-stock diversity-weighted portfolio with $r<0$ as defined in \eqref{largeDWPp} to compare its performance to the market and positive-parameter diversity-weighted portfolios. The resulting wealth process is visualised in Figure \ref{largepic}.

\subsubsection{Discussion} \label{disco!}
The assumption \eqref{nofail} says that no capitalisation can go to zero, i.e.~no company can fail. This is obviously not true in the real world, and we see that this is where the danger of diversity-weighted portfolios with $p<0$ lies: if $\exists k\in\{1,\ldots,n\}, t_k$ such that $\mu_k(t)\to0$ as $t\to t_k$, then $\lim_{t\to t_k}\pi_k(t)=1$. In other words, this portfolio will put all of the investor's wealth in the crashing stock, causing her to go bankrupt. 

Several ways in which this could be avoided in practical applications have been demonstrated in Sections \ref{remove} -- \ref{secrankedDWP}. Other possibilities include capping the maximal proportion invested in any single stock, so that only this proportion is lost at bankruptcy, or using a portfolio of the form $\pi_i(t) \propto \mu_i(t)^{k-1}e^{-\mu_i(t)/\theta}$ (see Figure \ref{gammapic}) or $\pi_i(t) \propto \mu_i(t)^\alpha (1-\mu_i(t))^\beta$, which liquidate positions in crashing stocks but have similar behaviour to the DWP with $p<0$ in the upper range of market weights.

\begin{figure} [htp]
\centering
\includegraphics[width=.9\textwidth]{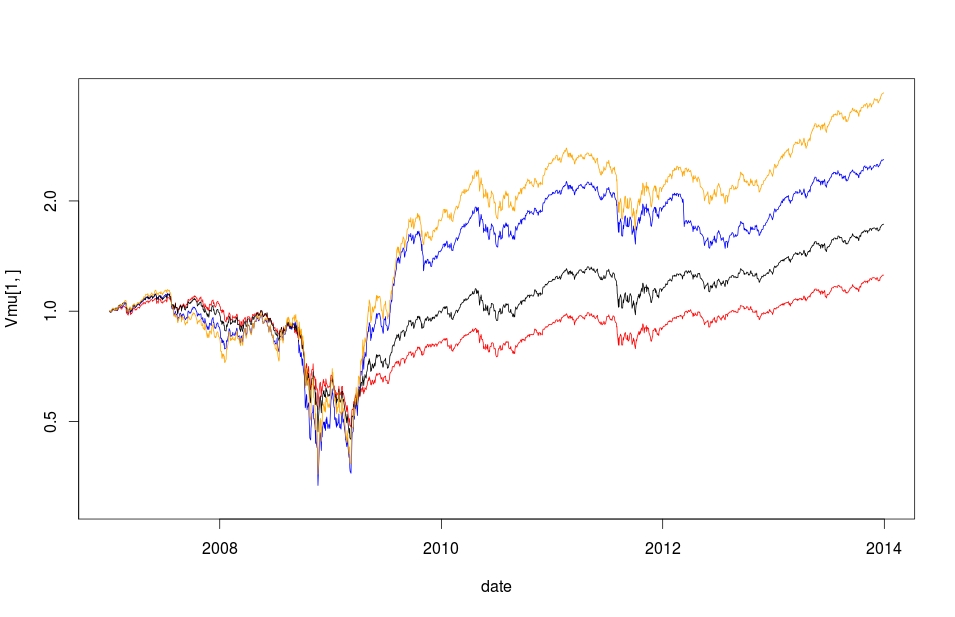}
\caption{Frictionless evolution of wealth processes corresponding to the `gamma-distributed' portfolio \textcolor{YellowOrange}{$\pi_i(t) \propto \mu_i(t)^{k-1}e^{-\mu_i(t)/\theta}$} with $k=1.5$ and $\theta=0.0001$, $\pi(\cdot)$ with \textcolor{red}{$p=1$} (market portfolio), \textcolor{blue}{$p=-1$}, and $p=0$ (EWP), all with the adjustment that $\pi_i(t)=0$ whenever $\mu_i(t)=0$.} \label{gammapic} 
\end{figure}

One might be tempted to apply the technique used in Section \ref{remove} using other stopping times, each representing the first time that a certain condition fails (these could be non-degeneracy, bounded variance, or even the `sufficient intrinsic volatility' condition $\gamma^*_\mu(t)\geq0\ \forall t\geq0\ \mathbb{P}$-a.s.). This would allow one to prove that certain portfolios which are relative arbitrages under these conditions also beat the market without these conditions, if one follows them up to the corresponding stopping time of the condition failing. However, one must keep in mind that this may only be done using stopping times \emph{with respect to the filtration} $\mathbb{F}^\mathfrak{X}$, with $\mathfrak{X}(\cdot)=(X_1(\cdot),\ldots,X_n(\cdot))$ the vector of capitalisations; we assume the investor does not have any additional information besides the observed stock prices. Thus a portfolio of the form 
\begin{equation}\label{notpihat}
\bar{\pi}(t):= \begin{cases}
\pi(t) & t<\tau_\delta\wedge \tilde{\tau}_\epsilon\\
\mu(t) & t \geq \tau_\delta\wedge \tilde{\tau}_\epsilon,
\end{cases}
\end{equation}
with
\begin{equation}
\tilde{\tau}_\epsilon := \inf\{ t>0\ |\ \text{non-degeneracy with parameter } \varepsilon \text{ fails }\},
\end{equation}
is \emph{not} a predictable portfolio, and is therefore not allowed, since one cannot determine whether $\tilde{\tau}_\epsilon$ has occurred or not by only observing stock prices.

The above insight does give the following clue: in the quest for a relative arbitrage over arbitrarily short time horizons in sufficiently volatile markets, a major open problem, one could attempt to find a portfolio that has this desirable property under an additional observable assumption on the behaviour of $X_i(\cdot),i=1,\ldots,n$, and then use the construction \eqref{pihat} to make a relative arbitrage in the more general market where this assumption does not hold.

For this portfolio, we plan on researching the following:
\begin{itemize}
\item attempting to construct a short-term relative arbitrage using the technique of mirror portfolios as developed in \cite{fkk05}, or as in \cite{bf08},
\item extending our results to more general semimartingale market models, in the spirit of Section 7 of \cite{kard08},
\item investigating the validity and effect of applying the continuous-time theory in discrete time,
\item mixing the DWP with the EWP or market portfolio, in order to construct a relative arbitrage that holds under milder conditions than \eqref{nofail}.
\end{itemize}
For the latter we plan on using the results in Table \ref{table1}.

\section{Future research} \label{secdirections}
We present some ideas that we think would be worthwhile looking at in detail. Some of these were mentioned above, at the end of Section \ref{secme}.

\subsection{Optimal relative arbitrage and incorporation of information}\label{futinc}
We wish to attempt generalising the characterisation of optimal relative arbitrage in \cite{fernk10} to incomplete markets with possibly $\mathbb{F}\neq\mathbb{F}^W$. Michael Monoyios has suggested doing this through a `fictitious completion' of the market as first proposed in \cite{klsx91}, i.e.~ hedging the claim $X(T)$ by embedding the incomplete model into an unconstrained market. It would also be of interest to look for a characterisation of the optimal (long-only) portfolio or FGP. One possible way of conducting this study is by using the Monge-Kantorovich optimal transport briefly described in Section 4 of \cite{palw14}.

The large advantage of SPT is clearly its robustness, in the sense that it does not require parameter estimation. All of the portfolios proposed by the theory are implementable using only the current capitalisations of companies in the market. However, an investor might want to include certain beliefs, for instance on the terminal values and drifts of certain stock prices or because she has insider information, in selecting investment strategies --- this is not yet possible within the framework of SPT, and would be an interesting research direction. \cite{strong13} has made a first step in the direction of incorporating additional information, by defining an extension of FGPs which allows for the possibility of having the generating function depend on additional finite-variation arguments. Examples of information to be included in such arguments are  fundamental economic data and information extracted from Twitter feeds.
Strong does not show how the inclusion of such information would influence the resulting portfolios, or how one might go about optimising a relative arbitrage given certain information. In \cite{ct13} one possible approach to doing this optimisation is described, which is to change measure using a density which translates an investor's belief that a certain event (i.e.~stopping time) will not occur before some horizon $T$, thus changing the dynamics of the stock price. However, these measure changes will be impossible to do explicitly in general models, and only the one-dimensional case is treated in \cite{ct13}.

A more promising approach has been taken in \cite{palw13}, where the authors are able to find the optimal portfolio given a weight function translating beliefs or statistical information in some two-stock markets --- see Section \ref{2Dopt}. They only do this in two toy examples of market models, but it would be interesting to attempt to generalise this approach to higher-dimensional and less explicit market models. Given a weight function $w$ as in Section \ref{2Dopt}, one could for instance try to find the optimal strategy within a class of FGP portfolios, e.g.\! diversity-weighted portfolios. One could then investigate whether this bears any relation to the optimal relative arbitrage characterisation in \cite{fernk10}, see Section \ref{optRA}.

\subsection{Information theoretic approach}\label{futinfo}
Recently, Pal and Wong have developed an alternative approach to studying portfolio performance in \cite{palw13}, which is inspired by information theory and completely model-independent (see, for instance, \cite{coverbook12}). They derive `master equations' for general portfolios in both discrete and continuous time. In discrete time, the market-relative performance of a portfolio $\pi(\cdot)$ is given as
\begin{align} \label{heypal}
\log \bigg(\frac{V^\pi(T)}{V^\mu(T)}\bigg) &= \sum_{t=0}^{T-1} \gamma^*_\pi(t) + H(\pi(0)|\mu(0)) - H(\pi(T)|\mu(T))\nonumber\\ &\quad + \sum_{t=0}^{T-1}\Big( H(\pi(t+1)|\mu(t+1))-H(\pi(t)|\mu(t+1)) \Big),
\end{align}
where $H:\Delta^n_+\times \Delta^n_+ \to \mathbb{R}$ is the entropy of a probability distribution $\pi$ with respect to a distribution $\mu$, defined as $$H(\pi|\mu)=\sum_{i=1}^n \pi_i \log \frac{\pi_i}{\mu_i},$$
and $\gamma^*_\pi(\cdot)$ is a discrete analogue of the excess growth rate, referred to by the authors as the `free energy'. 
That is, Pal and Wong interpret portfolios as discrete probability distributions over $n$ atoms, with $n$ the number of stocks in the market.\footnote{This interpretation might be extended when $n\to\infty$ in the large-market limit, giving an interpretation of large-market portfolios.}
The authors take this further in \cite{palw14}, where they show how the problem of finding RAs given some property of the market can be approached as a Monge-Kantorovich optimal transport problem. It is of great interest to further develop this approach. 

Note that for the EWP $\pi(\cdot)\equiv 1/n$, the last term in \eqref{heypal} is zero, and thus the discrete-time performance of the EWP with respect to the market can be decomposed into its cumulative free energy, or excess growth rate, which is monotonically increasing, and the negative change in entropy:
\begin{equation}
\log \bigg(\frac{V^\pi(T)}{V^\mu(T)}\bigg)=\sum_{t=0}^{T-1}\underbrace{\gamma^*_\pi(t)}_{\geq0} + \frac{1}{n}\sum_{i=1}^n \log \frac{\mu_i(T)}{\mu_i(0)}\geq D(\mu(T))- D(\mu(0)).
\end{equation}
Here, $D:\Delta^n_+\to \mathbb{R}$ defined by $x\mapsto 1/n \sum_{i=1}^n \log x_i$ is a measure of diversity, as defined in Definition 3.4.1 of \cite{f02}; i.e., it is $C^2$, symmetric and concave.\footnote{Note that $D$ as defined above does not match the definition in \cite{f02} exactly, as it is not positive; however, this is not important since we are only considering the \emph{change} in the value of $D$, i.e.~the decrease or increase in diversity.}
Since $D$ is also increasing, this proves the following interesting implication:
$$\text{market diversity increases}\ \Rightarrow\ \text{EWP beats the market}$$
Figure \ref{ewpfig} shows that despite a \emph{decrease} in market diversity over the observed period, the EWP still beat the market thanks to the cumulative free energy term, which is due to rebalancing. Here, $D$ has been extended to $\Delta^n$ by setting $D(x)=0\ \forall x\in\mathbb{R}^n$ s.t. $\prod_i x_i=0$.

\begin{figure} [htp]
\centering
\includegraphics[width=.9\textwidth]{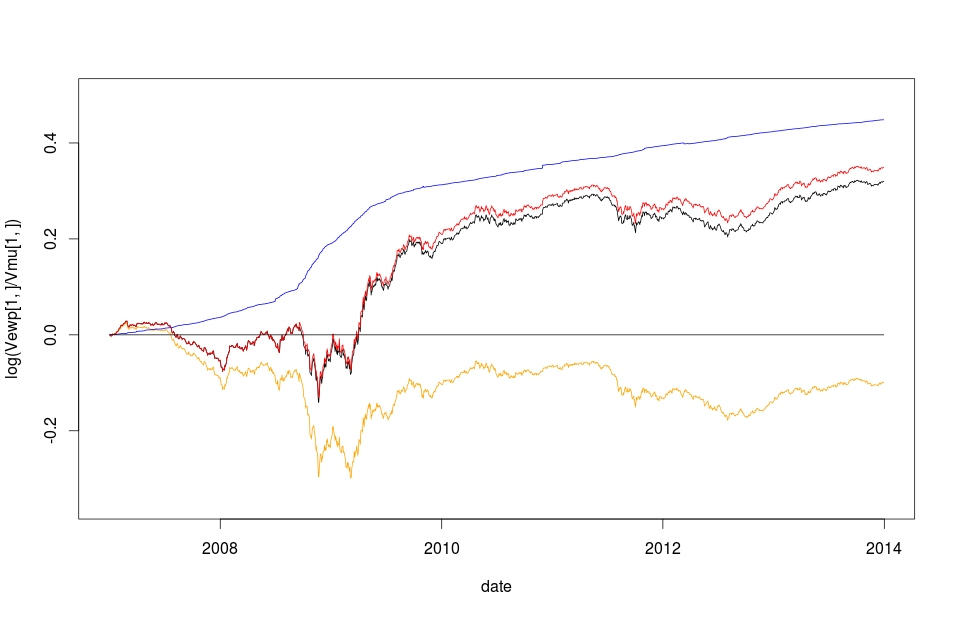}
\caption{Decomposition of frictionless log-relative performance of EWP with respect to the market (in black), into \textcolor{blue}{cumulative excess growth rate} and \textcolor{YellowOrange}{change in diversity}. The sum of these terms is displayed in \textcolor{red}{red}, and is different from actual wealth due to stocks crashing.}  \label{ewpfig}
\end{figure}

The information-theoretic approach to portfolio analysis in \cite{palw13} is entirely new, although it shares some features with Cover's construction of his universal portfolio --- see \cite{cover91} and \cite{jam92}. It would be interesting to see how this approach is related to Fernholz's approach in SPT, and whether this might be used to find any relation between the universal portfolio and FGPs, thus resolving the open question posed in Remark 11.7 in \cite{fk09}. However, since Cover's universal portfolio depends on the entire stock returns history, and FGPs only depend on the current market configuration, it remains unclear to the author what such a relation might look like. Perhaps, however, the idea of `selecting the best portfolio with hindsight' could be applied to FGPs, thus modifying Cover's algorithm by taking a performance-weighted average not over constant-proportion portfolios but over FGPs.

\subsection{Implementation and performance in real markets} \label{realimpl}
Some imperfections have been neglected in SPT thus far, the most notable of which is the presence of transaction costs, which largely limits the implementability of FGPs in a continuous setting since these are typically of infinite variation. The only existing theoretical work regarding this is in Section 6.3 in \cite{f02}, where Fernholz makes a rough approximation of the turnover of a diversity-weighted portfolio $\pi(\cdot)$ with parameter $p\in(0,1)$ when it is rebalanced every time the portfolio weights differ from the desired weights by a fixed multiple $\delta$ of the target weights. He finds that the total turnover up to time $t$ is
\begin{equation}
\frac{2}{\delta}(1-p)^2\gamma^*_\pi(t).
\end{equation}
This approximation is only a first attempt at quantifying the total amount of trading due to rebalancing, and many assumptions are made in order to make it. It would be interesting to try and improve this approximation, by making fewer assumptions, and by doing it for more general portfolios and rebalancing criteria. Ideally, we would like to develop a theory of transaction costs (and possibly optimisation) in the context of SPT.

Hardly any of the theoretical results in SPT have been tested using real market data. In Chapter 6 of \cite{f02}, the author uses historical stock price data over an 80-year window to compute the wealth processes of an investor who would have implemented the diversity- and entropy-weighted portfolios, showing that they would outperform the market significantly. However, this simple calculation ignores crucial imperfections such as transaction costs, the indivisibility of shares and market regulations --- with the introduction of proportional transaction costs, an investor naively following an SPT strategy would go bankrupt due to continuous rebalancing. In order to be able to study how the computed strategies could be implemented in real markets, by for instance a hedge fund, and how they compare to Cover's universal portfolio\footnote{This topic has only slightly been touched upon, in \cite{ipbkf11}.}, the portfolios suggested in \cite{palw13}, the equally-weighted portfolio, and the num\'eraire portfolio (see Section \ref{secoptarb}), it would be interesting to use real-world data and incorporate realistic market frictions --- we have made a first attempt at this in Section \ref{data}. Using data from real markets, for instance from the CRSP data set, one could also test the efficiency, in terms of trade-off between transaction costs and return, of different rebalancing frequencies or criteria (which comes down to choosing between different measures of distance between the traded and target portfolios).

Moreover, Michael Monoyios suggests investigating whether relative arbitrages are good from a utility point of view: even though they are not optimal in that they do not maximise expected utility, they do not require any drift or volatility estimation, unlike utility-optimal portfolios which \emph{do} depend explicitly on the drift and volatility of the stocks they invest in. One could try and quantify the amount of uncertainty required in a market for FGPs to do better than for instance the num\'eraire portfolio. Since it is typically very difficult to explicitly compute expected utility in the models considered in SPT (such as VSM),\footnote{Note that it is possible, by \cite{klsx91} and \cite{kk07}, to perform expected utility maximisation in a NUPBR market; an ELMM is not required.} one might have to do simulations or use real-world data in order to do this computation.

\subsection{Large markets}
In \cite{shk12} and \cite{shk13}, Shkolnikov studies the limiting behaviour of rank-based models and VSMs, respectively, when the number of stocks goes to infinity --- this problem was put forward in Remark 11.6 in \cite{fk09}. Although he is more interested in the resulting dynamics of the stocks in the large-market limit, we would like to see how these dynamics influence portfolio behaviour in such markets, and what a portfolio even means in this case. For instance, can one still construct relative arbitrages or long-term growth opportunities in large markets? Can these asymptotics give us any idea about how to invest in large markets, such as the American stock market? To answer these ideas, we would like to make use of the progress made by Hambly, Reisinger and others (see \cite{hamreis11} and \cite{pr13}) in studying large markets. An alternative approach might be the one in \cite{platrend12}, where Platen and Rendek directly apply the Law of Large Numbers and Central Limit Theorem to show that the equally-weighted portfolio converges to the num\'eraire portfolio when the number of market constituents tends to infinity.

\subsection{Others}
Other, more specific, research questions we would like to tackle include the following:
\begin{itemize}
\item Although of low priority, it would be nice to generalise the framework and results of SPT to semimartingale market models where there is also a jump component. One would have to think of analogues of the diversity and sufficient intrinsic volatility conditions, derive the theory of FGPs in this general setting, and see whether for instance the diversity- and entropy-weighted portfolios are still relative arbitrages and over what time horizons. Constantinos Kardaras has mentioned to the author that this can be done.
\item In a personal communication, Samuel Cohen has made an argument against the modelling approach of SPT, saying he finds it undesirable to make almost sure assumptions (such as diversity, or sufficient intrinsic volatility) when modelling stock prices, especially on events of very small probability, which he says can lead to perverse models. He suggests weakening such assumptions by having them hold with \emph{high} probability instead, or by making different but similar assumptions. This would probably no longer lead to almost sure comparisons as with relative arbitrages, but to statistical arbitrages. Michael Monoyios says this would relate to \cite{palshk10}'s study of the concentration of measure. We note that Section 5 of \cite{palw14} makes some remarks on statistical arbitrage, attempting to explain observations in \cite{fm07}.
\item Does relative arbitrage exist over finite (or even arbitrarily short) time horizons in (certain) rank-based models (e.g.~Atlas)? One would have to derive the dynamics of portfolios in such models first.
\item \cite{fk09} suggest computing the maximal relative return and shortest time to beat the market by a certain factor when one is only allowed to use long-only portfolios instead of general trading strategies.
\item \cite{fk05} suggest computing the shortest time to beat the market by a certain factor in VSMs with zero growth rates ($\alpha=0$).
\item We have looked at a few approaches to the major open question of constructing a short-term relative arbitrage in sufficiently volatile markets (see Section \ref{MOQ}). We plan to write up our ideas and decide whether this is a feasible project.
\item Idem for the condition that the \emph{generalised} growth rate $\gamma^*_{\pi,p}(\cdot)$ is bounded from below, as posed in \cite[p. 452]{bf08}.
\item Could one construct short-term RAs using the mirror-image method, but constructing a `seed' using the diversity-weighted portfolio with $p>1$ instead of $e_1$?
\item How valid is it to apply the continuous-time theory of SPT in discrete time?
\item Study the capital-distribution curve at a shorter time-scale; does stability prevail? Also during crashes?
\end{itemize}

\section*{Acknowledgements}
\addcontentsline{toc}{section}{Acknowledgements}
The author would first of all like to thank his supervisor, Michael Monoyios, for proofreading this work and providing feedback. He would also like to express his gratitude to Johannes Ruf and Christoph Reisinger for several fruitful discussions and helpful comments, and to Ioannis Karatzas for his involved correspondence and suggestions regarding the diversity-weighted portfolio with negative parameter. Finally, Thibaut Lienart and Samuel Cohen have made useful contributions to the empirical study.

\bibliographystyle{alpha}
\bibliography{sptrefs}

\newcommand{\etalchar}[1]{$^{#1}$}
\begin{thebibliography}{BHH{\etalchar{+}}11}

\bibitem[BF08]{bf08}
Adrian~D. Banner and Daniel Fernholz.
\newblock Short-term relative arbitrage in volatility-stabilized markets.
\newblock {\em Ann. Finance}, 4:445--454, 2008.

\bibitem[BFK05]{bfk05}
Adrian~D. Banner, Daniel Fernholz, and Ioannis Karatzas.
\newblock Atlas models of equity markets.
\newblock {\em Ann. Appl. Probab.}, 15(4):2296--2330, 2005.

\bibitem[BHH{\etalchar{+}}11]{hamreis11}
Nick Bush, Ben~M Hambly, Helen Haworth, Lei Jin, and Christoph Reisinger.
\newblock Stochastic evolution equations in portfolio credit modelling.
\newblock {\em SIAM J. Financial Math.}, 2(1):627--664, 2011.

\bibitem[BHS12]{bhs12}
Erhan Bayraktar, Yu-Jui Huang, and Qingshuo Song.
\newblock Outperforming the market portfolio with a given probability.
\newblock {\em Ann. Appl. Probab.}, 22(4):1465--1494, 2012.

\bibitem[BKX12]{bkx07}
Erhan Bayraktar, Constantinos Kardaras, and Hao Xing.
\newblock Strict local martingale deflators and valuing {A}merican call-type
  options.
\newblock {\em Finance Stoch.}, (16):275--291, 2012.

\bibitem[BR12]{br12}
Angus~A Brown and Leonard~CG Rogers.
\newblock Diverse beliefs.
\newblock {\em Stochastics An International Journal of Probability and
  Stochastic Processes}, 84(5-6):683--703, 2012.

\bibitem[CH05]{ch05}
Alexander~M.G. Cox and David~G. Hobson.
\newblock Local martingales, bubbles and option prices.
\newblock {\em Finance Stoch.}, 9:477--492, 2005.

\bibitem[Cov91]{cover91}
Thomas~M Cover.
\newblock Universal portfolios.
\newblock {\em Math. Finance}, 1(1):1--29, 1991.

\bibitem[CT12]{coverbook12}
Thomas~M Cover and Joy~A Thomas.
\newblock {\em Elements of {I}nformation {T}heory}.
\newblock John Wiley \& Sons, 2012.

\bibitem[CT13]{ct13}
Huy~N. Chau and Peter Tankov.
\newblock Market models with optimal arbitrage.
\newblock {\em arXiv preprint arXiv:1312.4979v1}, 2013.

\bibitem[DGU09]{uppal09}
Victor DeMiguel, Lorenzo Garlappi, and Raman Uppal.
\newblock Optimal versus naive diversification: How inefficient is the 1/n
  portfolio strategy?
\newblock {\em Review of Financial Studies}, 22(5):1915--1953, 2009.

\bibitem[DS94]{ds94}
Freddy Delbaen and Walter Schachermayer.
\newblock A general version of the fundamental theorem of asset pricing.
\newblock {\em Math. Ann.}, 300(3):463--520, 1994.

\bibitem[DS95a]{ds95}
Freddy Delbaen and Walter Schachermayer.
\newblock Arbitrage possibilities in {B}essel processes and their relations to
  local martingales.
\newblock {\em Probab. Theory Relat. Fields}, 102(3):357--366, 1995.

\bibitem[DS95b]{ds95b}
Freddy Delbaen and Walter Schachermayer.
\newblock The existence of absolutely continuous local martingale measures.
\newblock {\em Ann. Appl. Probab.}, 5(4):926--945, 1995.

\bibitem[DS95c]{ds95c}
Freddy Delbaen and Walter Schachermayer.
\newblock The no-arbitrage property under a change of num{\'e}raire.
\newblock {\em Stochastics and Stochastic Reports}, 53(3-4):213--226, 1995.

\bibitem[Fer99a]{f99}
Robert Fernholz.
\newblock On the diversity of equity markets.
\newblock {\em J. of Math. Econ.}, 31:393--417, 1999.

\bibitem[Fer99b]{f95}
Robert Fernholz.
\newblock Portfolio generating functions.
\newblock {\em Quantitative Analysis in Financial Markets, River Edge, NJ.
  World Scientific}, 1999.

\bibitem[Fer01]{f01}
Robert Fernholz.
\newblock Equity portfolios generated by functions of ranked market weights.
\newblock {\em Finance Stoch.}, 5:469--486, 2001.

\bibitem[Fer02]{f02}
Robert Fernholz.
\newblock {\em Stochastic {P}ortfolio {T}heory}.
\newblock Springer, 2002.

\bibitem[FGH98]{fgh98}
Robert Fernholz, Robert Garvy, and John Hannon.
\newblock Diversity-weighted indexing.
\newblock {\em The Journal of Portfolio Management}, 24(2):74--82, 1998.

\bibitem[FIK13a]{fik13a}
Robert Fernholz, Tomoyuki Ichiba, and Ioannis Karatzas.
\newblock A second-order stock market model.
\newblock {\em Ann. Finance}, 9:439--454, 2013.

\bibitem[FIK13b]{fik13b}
Robert Fernholz, Tomoyuki Ichiba, and Ioannis Karatzas.
\newblock Two {B}rownian particles with rank-based characteristics and
  skew-elastic collisions.
\newblock {\em Stochastic Process. Appl.}, 123:2999--3026, 2013.

\bibitem[FIKP13]{fikp13}
Robert Fernholz, Tomoyuki Ichiba, Ioannis Karatzas, and Vilmos Prokaj.
\newblock Planar diffusions with rank-based characteristics and perturbed
  {T}anaka equations.
\newblock {\em Ann. Finance}, 156:343--374, 2013.

\bibitem[FJS14]{fontetal13}
Claudio Fontana, Monique Jeanblanc, and Shiqi Song.
\newblock On arbitrages arising with honest times.
\newblock {\em Finance Stoch.}, 18:515--543, 2014.

\bibitem[FK97]{follmk97}
Hans F{\"o}llmer and Dmitry Kramkov.
\newblock Optional decompositions under constraints.
\newblock {\em Probab. Theory Relat. Fields}, 109(1):1--25, 1997.

\bibitem[FK05]{fk05}
Robert Fernholz and Ioannis Karatzas.
\newblock Relative arbitrage in volatility-stabilized markets.
\newblock {\em Ann. Finance}, 1:149--177, 2005.

\bibitem[FK09]{fk09}
Robert Fernholz and Ioannis Karatzas.
\newblock Stochastic portfolio theory: A survey.
\newblock In Alain Bensoussan and Qiang Zhang, editors, {\em Handbook of
  {N}umerical {A}nalysis. {V}ol. {XV}. {S}pecial volume: mathematical modeling
  and numerical methods in finance}, volume~15 of {\em Handbook of Numerical
  Analysis}. Elsevier/North-Holland, Amsterdam, 2009.

\bibitem[FK10]{fernk10}
Daniel Fernholz and Ioannis Karatzas.
\newblock On optimal arbitrage.
\newblock {\em Ann. Appl. Probab.}, 20(4):1179--1204, 2010.

\bibitem[FK11]{fk11}
Daniel Fernholz and Ioannis Karatzas.
\newblock Optimal arbitrage under model uncertainty.
\newblock {\em Ann. Appl. Probab.}, 21(6):2191--2225, 2011.

\bibitem[FKK05]{fkk05}
Robert Fernholz, Ioannis Karatzas, and Constantinos Kardaras.
\newblock Diversity and relative arbitrage in equity markets.
\newblock {\em Finance Stoch.}, 9(1):1--27, 2005.

\bibitem[FMJ07]{fm07}
Robert Fernholz and Cary Maguire~Jr.
\newblock The statistics of statistical arbitrage.
\newblock {\em Financial Analysts Journal}, 63(5):46--52, 2007.

\bibitem[F{\"o}l72]{follmer72}
Hans F{\"o}llmer.
\newblock The exit measure of a supermartingale.
\newblock {\em Probab. Theory Relat. Fields}, 21(2):154--166, 1972.

\bibitem[Fon13]{font13}
Claudio Fontana.
\newblock Weak and strong no-arbitrage conditions for continuous financial
  markets.
\newblock {\em arXiv preprint arXiv:1302.7192}, 2013.

\bibitem[FR13]{fr13}
Claudio Fontana and Wolfgang~J Runggaldier.
\newblock Diffusion-based models for financial markets without martingale
  measures.
\newblock In {\em Risk {M}easures and {A}ttitudes}, pages 45--81. Springer,
  2013.

\bibitem[GR11]{gr11}
Paolo Guasoni and Mikl{\'o}s R{\'a}sonyi.
\newblock Fragility of arbitrage and bubbles in diffusion models.
\newblock {\em Abailable at SSRN 2317344 (2013)}, 2011.

\bibitem[IKS13]{iks13}
Tomoyuki Ichiba, Ioannis Karatzas, and Mykhaylo Shkolnikov.
\newblock Strong solutions of stochastic equations with rank-based
  coefficients.
\newblock {\em Probab. Theory Relat. Fields}, 156:229--248, 2013.

\bibitem[IP11]{imp11}
Peter Imkeller and Nicolas Perkowski.
\newblock The existence of dominating local martingale measures.
\newblock {\em arXiv preprint arXiv:1111.3885}, 2011.

\bibitem[IPB{\etalchar{+}}11]{ipbkf11}
Tomoyuki Ichiba, Vassilios Papathanakos, Adrian Banner, Ioannis Karatzas, and
  Robert Fernholz.
\newblock Hybrid {A}tlas models.
\newblock {\em Ann. Appl. Probab.}, 21(2):609--644, 2011.

\bibitem[IPS13]{ips13}
Tomoyuki Ichiba, Soumik Pal, and Mykhaylo Shkolnikov.
\newblock Convergence rates for rank-based models with applications to
  portfolio theory.
\newblock {\em Probab. Theory Relat. Fields}, 156:415--448, 2013.

\bibitem[Jam92]{jam92}
Farshid Jamshidian.
\newblock Asymptotically optimal portfolios.
\newblock {\em Math. Finance}, 2(2):131--150, 1992.

\bibitem[JR13]{jr13}
Benjamin Jourdain and Julien Reygner.
\newblock Capital distribution and portfolio performance in the mean-field
  {A}tlas model.
\newblock {\em arXiv preprint arXiv:1312.5660}, 2013.

\bibitem[Kar08]{kard08}
Constantinos Kardaras.
\newblock Balance, growth and diversity of financial markets.
\newblock {\em Ann. Finance}, 4:369--397, 2008.

\bibitem[Kar12a]{kard12}
Constantinos Kardaras.
\newblock Market viability via absence of arbitrage of the first kind.
\newblock {\em Finance Stoch.}, 16:651--667, 2012.

\bibitem[Kar12b]{kard12val}
Constantinos Kardaras.
\newblock Valuation and parity formulas for exchange options.
\newblock {\em arXiv preprint arXiv:1206.3220}, 2012.

\bibitem[KK07]{kk07}
Ioannis Karatzas and Constantinos Kardaras.
\newblock The num\'eraire portfolio in semimartingale financial models.
\newblock {\em Finance Stoch.}, 11:447--493, 2007.

\bibitem[KLSX91]{klsx91}
Ioannis Karatzas, John~P. Lehoczky, Steven~E. Shreve, and Gan-Lin Xu.
\newblock Martingale and duality methods for utility maximization in an
  incomplete market.
\newblock {\em SIAM J. Control Optim.}, 29(3):702--730, 1991.

\bibitem[KS14]{ks14}
Ioannis Karatzas and Andrey Sarantsev.
\newblock Diverse market models of competing {B}rownian particles with splits
  and mergers.
\newblock {\em Preprint. Available at arXiv:1404.0748}, 2014.

\bibitem[Mon07]{mmima07}
Michael Monoyios.
\newblock Optimal hedging and parameter uncertainty.
\newblock {\em IMA J. Manag. Math.}, 18(4):331--351, 2007.

\bibitem[MU10]{mu10}
Aleksandar Mijatovic and Mikhail Urusov.
\newblock Deterministic criteria for the absence of arbitrage in diffusion
  models.
\newblock {\em Finance Stoch.}, 2010.
\newblock Available at http://www.ma.ic.ac.uk/\textasciitilde amijatov/.

\bibitem[OR06]{or06}
J\"org~R. Osterrieder and Thorsten Rheinl\"ander.
\newblock Arbitrage opportunities in diverse markets via a non-equivalent
  measure change.
\newblock {\em Ann. Finance}, 2:287--301, 2006.

\bibitem[Pal11]{pal11}
Soumik Pal.
\newblock Analysis of market weights under volatility-stabilized market models.
\newblock {\em Ann. Appl. Probab.}, 21(3):1180--1213, 2011.

\bibitem[PH06]{ph06}
Eckhard Platen and David Heath.
\newblock {\em A {B}enchmark {A}pproach to {Q}uantitative {F}inance}.
\newblock Springer, 2006.

\bibitem[Pic14]{pic13}
Radka Pickov{\'a}.
\newblock Generalized volatility-stabilized processes.
\newblock {\em Ann. Finance}, 10(1):101--125, 2014.

\bibitem[PP07]{pp07}
Soumik Pal and Philip Protter.
\newblock Strict local martingales, bubbles, and no early exercise.
\newblock {\em Preprint, Nov}, 2007.

\bibitem[PR12]{platrend12}
Eckhard Platen and Renata Rendek.
\newblock Approximating the num{\'e}raire portfolio by naive diversification.
\newblock {\em Journal of Asset Management}, 13(1):34--50, 2012.

\bibitem[PR13]{pr13}
Nicolas Perkowski and Johannes Ruf.
\newblock Supermartingales as {R}adon-{N}ikodym densities and related measure
  extensions.
\newblock {\em arXiv preprint arXiv:1309.4623}, 2013.

\bibitem[Pro13]{protter13}
Philip Protter.
\newblock Strict local martingales with jumps.
\newblock {\em arXiv preprint arXiv:1307.2436}, 2013.

\bibitem[PS10]{palshk10}
Soumik Pal and Mykhaylo Shkolnikov.
\newblock Concentration of measure for systems of {B}rownian particles
  interacting through their ranks.
\newblock {\em arXiv preprint arXiv:1011.2443}, 2010.

\bibitem[PW13]{palw13}
Soumik Pal and Ting-Kam~Leonard Wong.
\newblock Energy, entropy, and arbitrage.
\newblock {\em Preprint. Available at arXiv:1308.5376}, 2013.

\bibitem[PW14]{palw14}
Soumik Pal and Ting-Kam~Leonard Wong.
\newblock The geometry of relative arbitrage.
\newblock {\em Preprint. Available at arXiv:1402.3720}, 2014.

\bibitem[Rog01]{lcgr01}
L.~C.~G. Rogers.
\newblock The relaxed investor and parameter uncertainty.
\newblock {\em Finance Stoch.}, 5(2):131--154, 2001.

\bibitem[RR13]{rr13}
Johannes Ruf and Wolfgang Runggaldier.
\newblock A systematic approach to constructing market models with arbitrage.
\newblock {\em arXiv preprint arXiv:1309.1988}, 2013.

\bibitem[Ruf11]{ruf11}
Johannes Ruf.
\newblock {\em Optimal Trading Strategies Under Arbitrage}.
\newblock PhD thesis, Columbia University, 2011.

\bibitem[Ruf13]{ruf13a}
Johannes Ruf.
\newblock Hedging under arbitrage.
\newblock {\em Math. Finance}, 23(2):297--317, 2013.

\bibitem[Sar14]{sar14}
Andrey Sarantsev.
\newblock On a class of diverse market models.
\newblock {\em Ann. Finance}, 10(2):291--314, 2014.

\bibitem[SF11]{sf11}
Winslow Strong and Jean-Pierre Fouque.
\newblock Diversity and arbitrage in a regulatory breakup model.
\newblock {\em Ann. Finance}, 7:349--374, 2011.

\bibitem[Shk12]{shk12}
Mykhaylo Shkolnikov.
\newblock Large systems of diffusions interacting through their ranks.
\newblock {\em Stochastic Process. Appl.}, 122(4):1730--1747, 2012.

\bibitem[Shk13]{shk13}
Mykhaylo Shkolnikov.
\newblock Large volatility-stabilized markets.
\newblock {\em Stochastic Process. Appl.}, 123:212--228, 2013.

\bibitem[Str13]{strong13}
Winslow Strong.
\newblock Generalizations of functionally generated portfolios with
  applications to statistical arbitrage.
\newblock {\em arXiv preprint arXiv:1212.1877}, 2013.

\end{thebibliography}

\end{document}